\documentclass[12pt]{amsart}
\usepackage{graphicx}
\usepackage{epstopdf}

\usepackage{amsmath}
\usepackage{xcolor}
\usepackage{float}
\usepackage{mathtools}
\usepackage{amsthm}
\usepackage{cite}
\usepackage{stix}
\usepackage{amsfonts}
\usepackage{setspace}
\usepackage{adjustbox}
\usepackage{subcaption}

\usepackage{tikzit}
\usepackage{blindtext}
\usepackage[breaklinks]{hyperref}
\usepackage[toc,page]{appendix}
\usepackage[font={small}]{caption}

\usepackage{tikzcode}

\newcommand{\bmb}{\left( \begin{array}{rr}}
\newcommand{\enm}{\end{array}\right)}

\newcommand{\cT}{\mathcal T}

\newcommand{\cO}{\mathcal O}

\newcommand{\cS}{\mathcal S}

\newcommand{\ci}{{\mathrm i}}

\newcommand{\io}{{\iota}}

\newcommand{\tS}{{\widetilde S}}

\newcommand{\C}{{\mathbb C}}
\newcommand{\Z}{{\mathbb Z}}

\newcommand{\R}{{\mathbb R}}

\newcommand{\epsijk}{{\epsilon(j,k)}}
\newcommand{\fjk}{{(n+1)f(\frac{j+k}{n+1})}}

\newcommand{\rr}{{r(n;j,k)}}

\newcommand{\bone}{{\mathbf 1}}
\newcommand{\al}{{\alpha}}

\newcommand{\ds}{{\mathrm ds}}

\newcommand{\dw}{{\mathrm dw}}

\numberwithin{equation}{section}

\theoremstyle{definition}
\newtheorem{thm}{Theorem}[section]
\newtheorem{prop}[thm]{Proposition}
\newtheorem{defn}[thm]{Definition}
\newtheorem{lemma}[thm]{Lemma}

\newtheorem{cor}[thm]{Corollary}
\newtheorem{remark}[thm]{Remark}
\newtheorem{example}[thm]{Example}

\newtheorem{assumption}[thm]{Assumption}

\begin{document}

\title{Perfect t-embeddings of uniformly weighted generalized tower graphs}
\author{David Keating} 
\address{Department of Mathematics, University of Illinois, Urbana, IL 61821, U.S.A. 
\break  e-mail: dkeating@illinois.edu
}
\author{Hieu Trung Vu}
\address{Yau Mathematical Sciences Center, Tsinghua University, Beijing, China
\break  e-mail: htrungvu@mail.tsinghua.edu.cn}
\begin{abstract}
In this paper, we study sequences of perfect t-embeddings of a uniformly weighted family of graphs we call generalized tower graphs. We show that the embeddings of these graphs satisfy certain technical assumptions, in particular, the rigidity assumption of Berggren-Nicoletti-Russkikh in \cite{tower_AD_perfect}.  As a result, we confirm the convergence of the gradients of the height function fluctuations of these graphs to those of the Gaussian free field. 
\end{abstract}

\maketitle
\date{\today}
\tableofcontents

\section{Introduction}

Let $G=(V \cup E,\nu)$ be a edge-weighted, planar, bipartite graph with vertex set $V$, edge set $E$, and edge weights  $\nu: E \to \R$. A dimer configuration, or dimer cover, on the graph $G$ is any subset of $E$ such that each vertex belongs to precisely one edge of the subset. We denote the set of all possible dimer configurations on $G$ by $\mathcal M(G)$. A dimer model on $G$ is a probability measure on $\mathcal M(G)$, defined for $m \in \mathcal M(G)$ as 
$$
\mathbb{P}(m)=\frac{\prod_{e \in m} \nu(e)}{\sum_{m^{\prime}\in \mathcal{M}(G)} \prod_{e \in m^{\prime}} \nu(e)}.
$$
The denominator in the above is known as the partition function of the model. We refer to \cite{MR2523460,MR2198850,MR4299268,MR3444135} for surveys and expository materials on dimer models and give a brief summary of the features important to the current work below.

The dimer model was introduced in the physics and chemistry communities to study the absorption of diatomic molecules on the surface of a crystal by Fowler and Rushbrooke \cite{edselc.2-52.0-3704915511519370101} in 1937. Later in the 1960's the model was studied from the mathematical point of view (independently) by Kasteleyn \cite{Kasteleyn}, Percus \cite{MR250899}, and Fisher and Temperley \cite{Temperley1961DimerPI}, giving the first enumerative result.  A special feature of the dimer model for certain families of graphs $(G_n)_{n>0}$ is that they display the \emph{arctic curve phenomenon} \cite{ProppShor, MR2952086} in the limit $n\to \infty$.  For such graphs, there exists phase separation in which regions where local configurations are fixed in one orientation of dimer, called \emph{frozen/solid regions}, are separated from regions where the local configurations are disordered, called \emph{liquid/rough regions}. The boundary separating the different regions approximates a limiting curve known as the \emph{arctic curve}. One of the most well-studied dimer model is the \emph{Aztec diamond} \cite{EKLP} for which, with uniform weights, the arctic curve is known to be a circle \cite{JPS}. See Section \ref{subsect:AD} for a precise description of the Aztec diamond. Figure \ref{fig:introAD} gives examples of domino tilings of the Aztec diamond of various ranks.

\begin{figure}
    \centering
    \resizebox{\textwidth}{!}{
    \begin{tabular}{ccc}
         \includegraphics[]{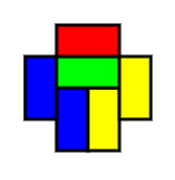}
         & 
         \includegraphics[scale=0.25]{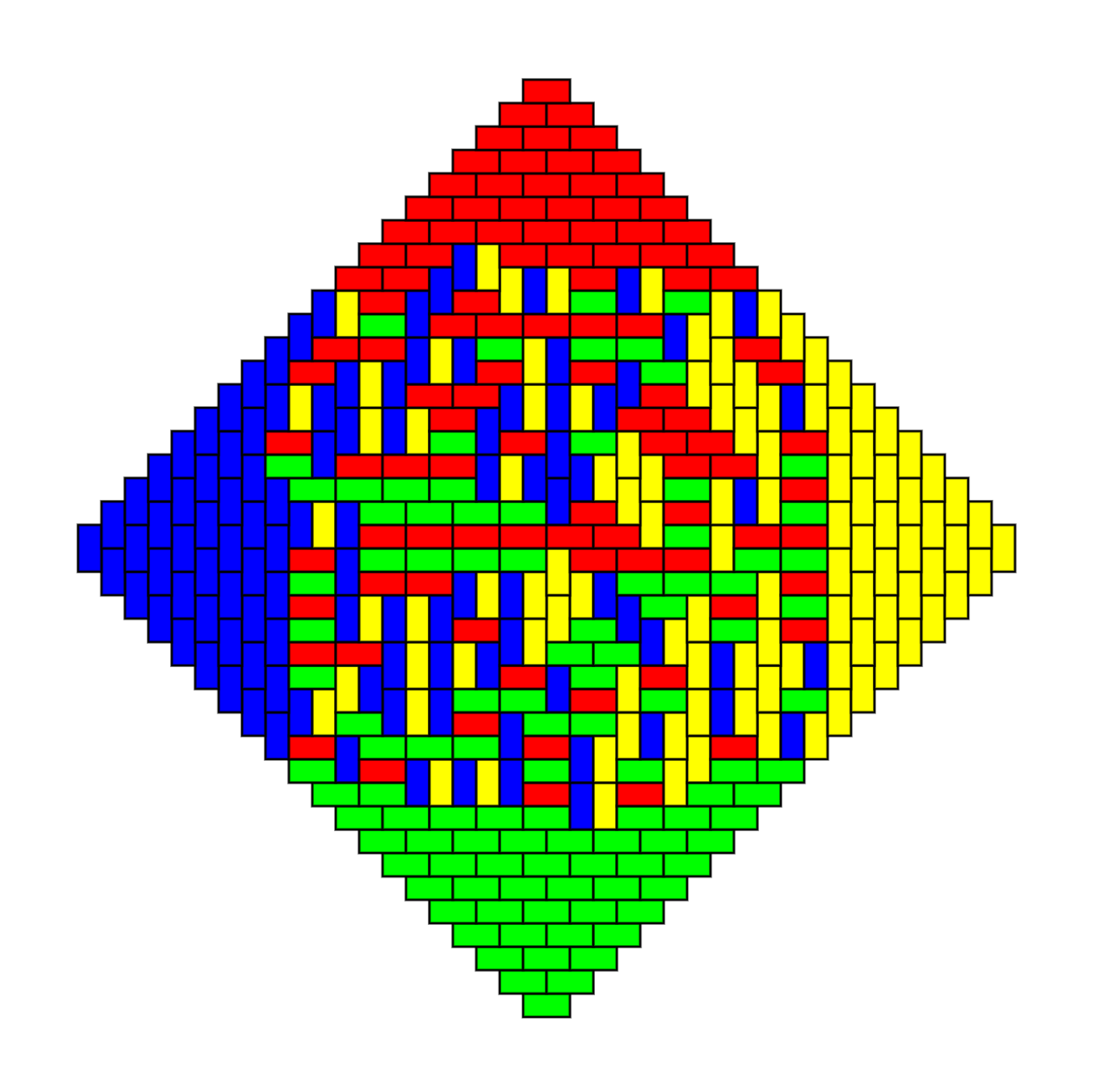}
         & 
         \includegraphics[scale=0.25]{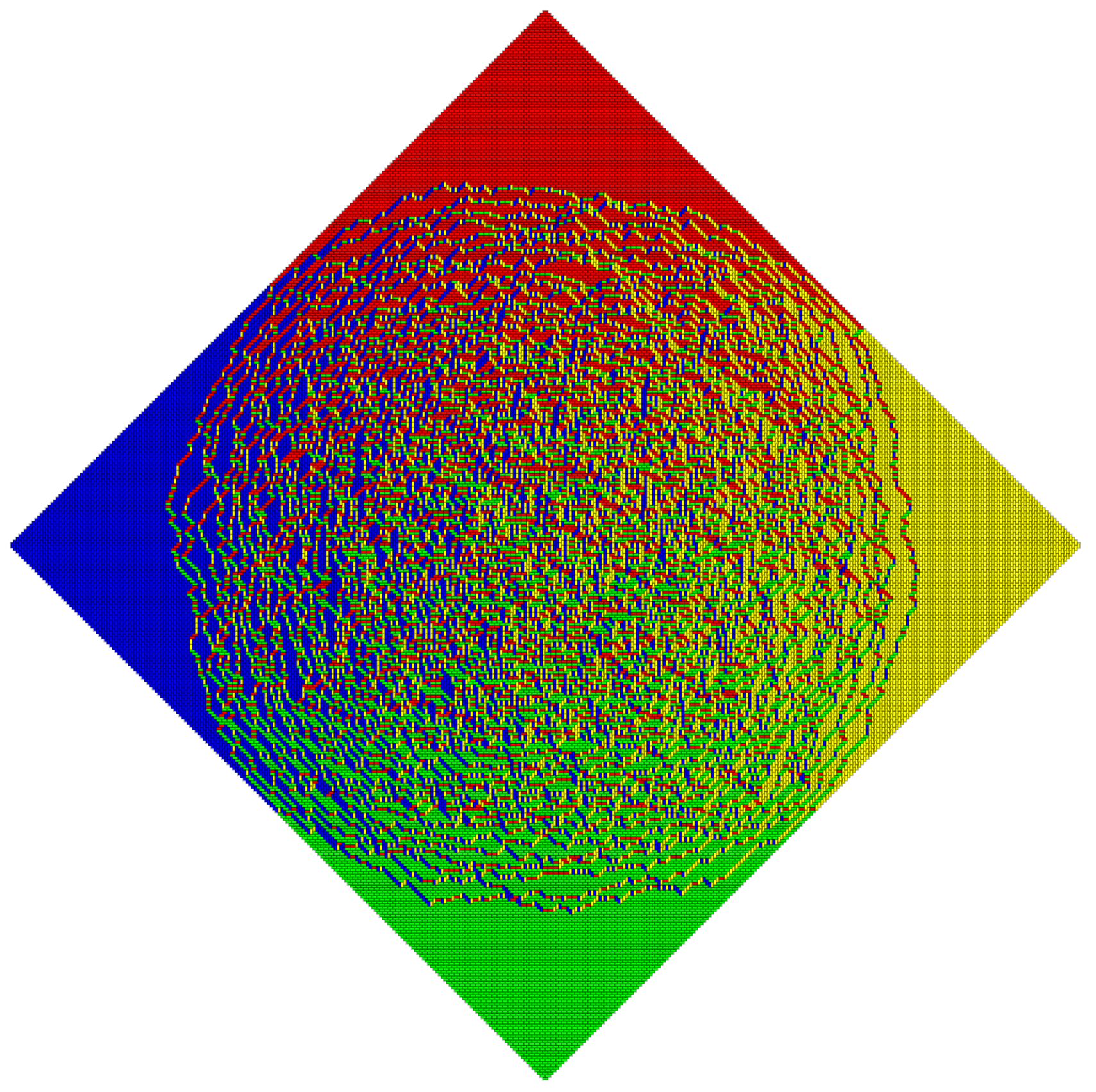}
    \end{tabular}
    }
    \caption{From left to right, tilings of the Aztec diamond of rank 2, 20, and 200, respectively. Each domino covers two faces of the dual graph to the Aztec diamond graph and corresponds to an edge belonging to the dimer cover. The different colors distinguish the four possible dimer orientations. Note that one can clearly see the arctic curve appearing in the right tiling.}
    \label{fig:introAD}
\end{figure}

A remarkable feature of dimer models is the \emph{shuffling algorithm}, the general version of which was introduced by Propp \cite{MR1990768}. The shuffling algorithm consists of certain local modifications to the graph $G$ along with its edge weights for which you can keep track of the change in the dimer model. A precise description of the shuffling algorithm is given in Section \ref{sect: gen_shuffling}. This algorithm allows us to recursively compute the probability that a randomly-chosen matching of an Aztec diamond will include a particular edge as well as sample a random matching. In \cite{circle_pattern}, the shuffling algorithm is shown to be equivalent to certain transformation of the complex planar embedding for dimer models on planar graphs. For more applications of the shuffling algorithm, see \cite{nicoletti2022localstatisticsshufflingdimers, rail-yard, MR3821249, anderson2025localstatisticsmndimermodel} for a non-exhaustive list. 

The dimer model can also be viewed as a stepped surface via \emph{Thurston's height function} \cite{MR4556490}. For a finite graph $G$, let $G^*$ be the planar dual of $G$ with vertices corresponding to the faces of $G$. Given any dimer cover $M$ of a planar, bipartite graph $G$, the height function $h: V(G^*) \to \R$ is a function on the set of vertices in the augmented dual $G^*$, obtained by fixing a reference matching $M_0$ of $G$ and the boundary values $h(v^*)=0$ for outer vertices of $G^*$. For an edge $(bw^*) \in V(G^*)$ oriented such that the white vertex is on the left, the height function $h(v^*)-h(u^*)=-\bone_{bw\in M} + \bone_{bw\in M_0}$. See Figure \ref{fig:height_function} for an example.
\begin{figure}
    \centering
    \begin{subfigure}{.45\textwidth}
    \centering
    \includegraphics[scale=.2]{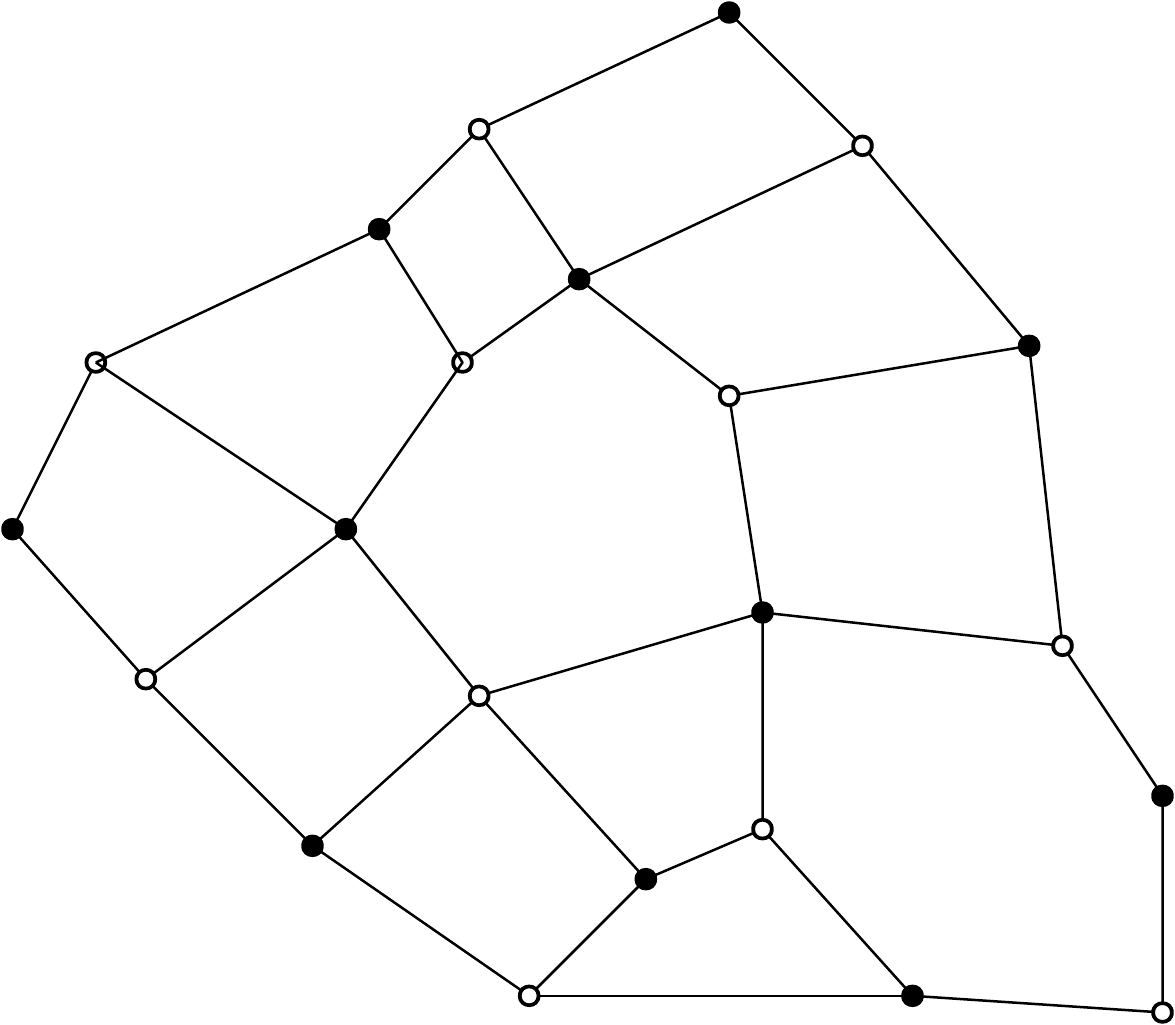}
    \end{subfigure}
    \begin{subfigure}{.45 \textwidth}
    \centering
        \includegraphics[scale=.2]{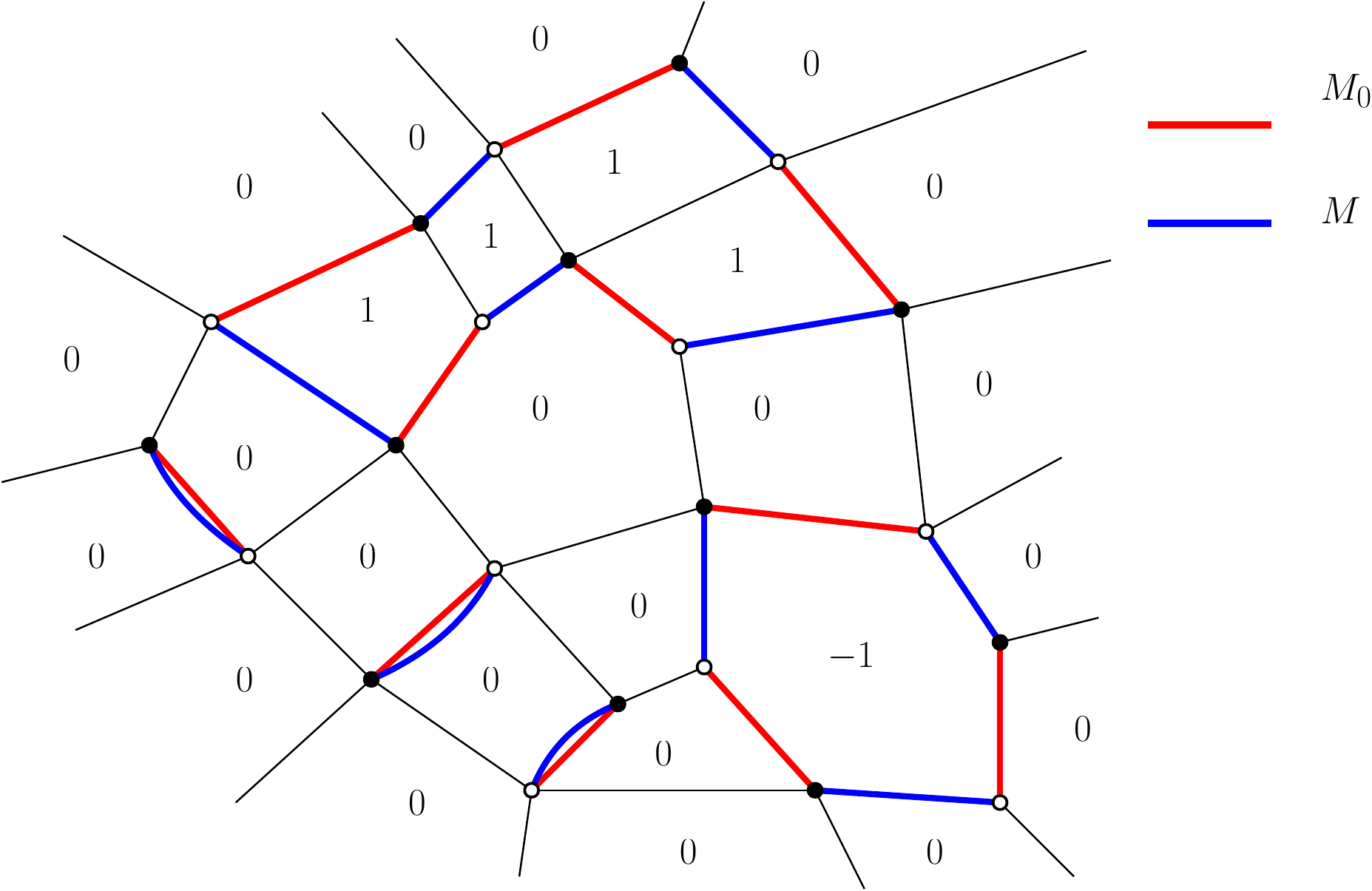}
    \end{subfigure}
    \caption{Example of a height function.}
    \label{fig:height_function}
\end{figure}

The height function plays an important role in describing the limiting structure of a dimer model. Namely, the fluctuations of the height function away from its mean were first shown to be related to the Gaussian Free Field (GFF) in \cite{Ken00, Ken01} for the cases of dimer models on Temperleyan domains which, roughly speaking, correspond to domains for which there are no frozen region. These results were later extended to several other cases by Russkikh \cite{Russ18,Russ21}. More generally, it was conjectured by Kenyon and Okounkov in 2007, that the height fluctuations in the liquid regions converge to the Gaussian free field in a non-trivial complex conformal structure \cite{OK1, OK2}. For background on the Gaussian free field, we refer the reader to \cite{MR4466634,MR2322706,berestycki2025gaussianfreefieldliouville}. This has motivated much subsequent research, and the conjecture has been proven in a handful of cases \cite{MR3821250,MR3148098,MR3020314,MR3278913,MR3298467,MR4448539,MR3825881,MR3861715,MR4653769, MR3821249,huang2020heightfluctuationsrandomlozenge}. More general results on height function fluctuations can also be found in \cite{MR3861715,MR3325273}.

For some 2-dimensional statistical mechanics models, one can consider the underlying objects (graphs, lattices, etc.) embedded in the complex planes. Then, a correct notion of discrete holomorphicity and discretization of the plane can be used to analyze the asymptotic behavior of these models by taking the limit as the mesh size goes to 0. This technique is extremely fruitful, namely, in the study of the critical Ising model on $\Z^2$ and isoradial graphs; see, e.g. \cite{CS11,CS12,HS13,CHI15,Smi06,Smir10,Smir10ICM}. 

Recently, in \cite{circle_pattern,MR4588721}, the authors introduced a new embedding of planar graphs in the complex plane called a \emph{t-embedding} $\mathcal{T}$ as well as a mapping to the plane called the \emph{origami map} $\mathcal{O}$. These embeddings were studied in \cite{MR4588721} and extended to the notion of \emph{perfect t-embeddings} in \cite{chelkak2021bipartitedimermodelperfect}, see Def. \ref{defn: perfect_t_embedding}. It is shown in \cite{chelkak2021bipartitedimermodelperfect,MR4588721} that for a sequence of graph $(G_n)_{n>0}$, if a corresponding sequence $\left(\mathcal{T}_n, \mathcal{O}_n\right)$ of perfect t-embeddings and associated origami maps exists and satisfies some technical assumptions (Assumptions \ref{assumption:LIP} and \ref{assumption:exp-fat}), and if the graphs of $\mathcal{O}_n$ over $\mathcal{T}_n$ converge to a maximal surface $S$ in the Minkowski space $\mathbb{R}^{2,1}$, then the gradients of $k$-point correlation functions converge to those of the standard Gaussian free field in the intrinsic metric of the surface $S$. Note that this is, however, is a weaker statement than the conjecture by Okounkov and Kenyon in \cite{OK1}, see Thm. \ref{thm:Main_thm_CLR21} and Section 7 of \cite{MR4588721}. Proving these technical assumptions is, in general, not a straightforward task. Until now, there are only a handful of cases in which these technical assumptions have been verified \cite{tower_AD_perfect, berggren2024perfecttembeddingslozengetilings,berggren2025perfecttembeddingsdoublyperiodic,berggren2025perfecttembeddingsoctahedronequation}. 

The main focus of this paper is the study of perfect t-embeddings of a large class of graphs we call \emph{generalized tower graphs}. For a function $f: \R \to \R$, satisfying certain constraints given in Section \ref{sect: shufupdown}, we may approximate $f$ by a an \emph{up-down path} on $\Z^2$. Using a version of the shuffling algorithm, we show how to construct a planar, bipartite graph whose boundary is determined by the up-down path approximating $f$. By scaling the integer lattice to have step size $\frac{1}{n+1}$, we can construct a family of graphs $(G_n^f)_{n>0}$ we call the \emph{generalized tower graphs approximating $f$}. Examples of these graphs include the Aztec diamond which is indexed by the function $f(x)=1$ and the tower graph which is indexed by $f(x)=\frac{1}{3}x+\frac{4}{3}$. Other examples are given in Figure \ref{fig:introEx}. These graphs can be seen as an example of \emph{rail-yard graphs} \cite{rail-yard}, for which the asymptotic behavior can be studied via the Schur process introduced by Reshetikhin and Okounkov in \cite{OR_Schur_03} (see Rmk. \ref{rmk:railyard}), although we do not pursue that here.  For linear functions $f$, this family of graphs also appear as a special case of the graphs of \cite{MR4782740} in the context of octahedron recurrence (see Rmk. \ref{rmk:2d}).   
\begin{figure}
    \centering
    \includegraphics[width=0.3\textwidth]{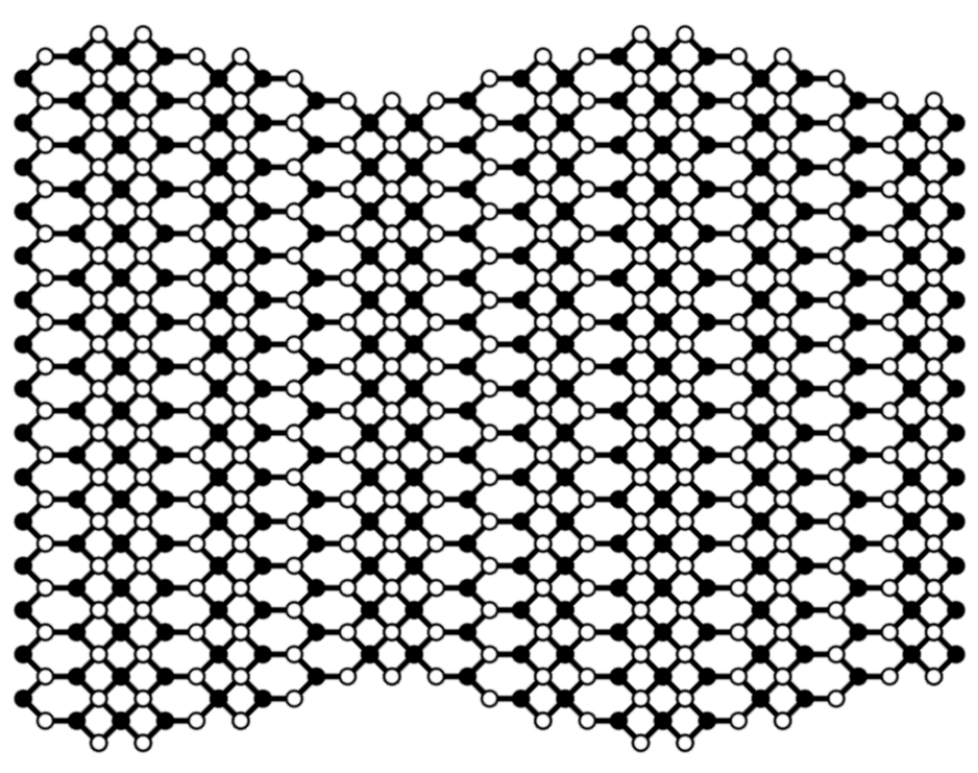}
     \qquad 
    \includegraphics[scale=.2]{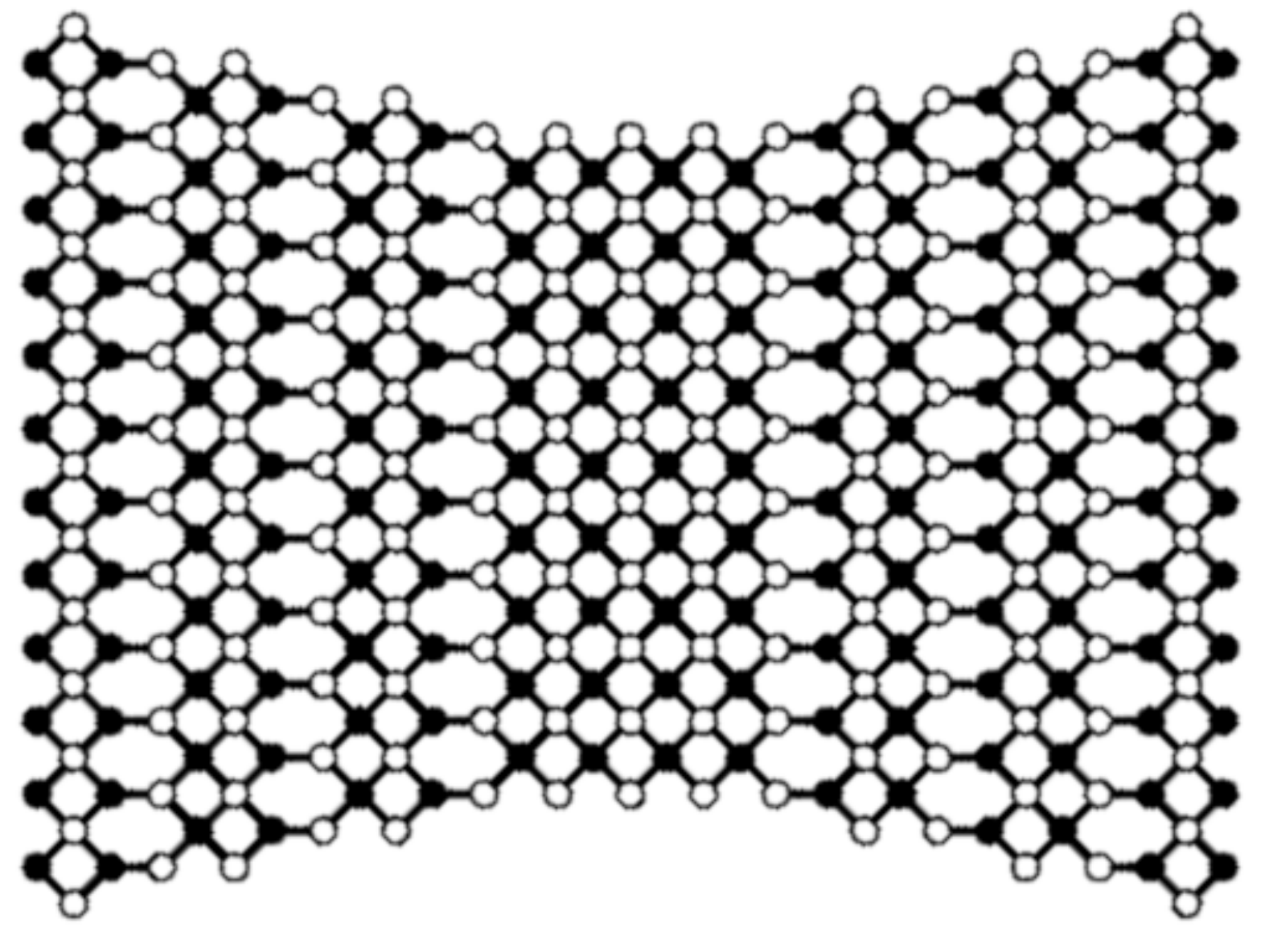}
    \caption{Left: Generalized tower graph of rank 15 approximating the function $f(x)=\frac{1}{3\pi} \sin (2\pi x)+1$. Right: Generalized tower graph of rank 15 approximating $f(x) = \frac{1}{3}x^2+\frac{2}{3}$.}
    \label{fig:introEx}
\end{figure}

The goal of this paper is to verify the assumptions in \cite{chelkak2021bipartitedimermodelperfect,MR4588721} and confirm that the gradients of $k$-point correlation functions, see Thm. \ref{thm:Main_thm_CLR21} for the precise formula, converge to those of the standard GFF.  As a byproduct of this analysis, we also compute the arctic curves for the generalized tower graphs.

The organization of this paper is the following. In Section \ref{sect: background}, we provide the necessary background on dimer models and t-embeddings, and present the main result of \cite{chelkak2021bipartitedimermodelperfect,MR4588721} that leads to the convergence of the gradients of $k$-point correlation of height function fluctuations as in Thm. \ref{thm:Main_thm_CLR21}. In Section \ref{sect: gen_shuffling}, we introduce the generalized tower graphs and describe their construction via a shuffling algorithm. Section \ref{sect: t_embedding} describes the relationship between the t-embeddings and shuffling. In particular, we relate the t-embedding of the generalized tower graphs to that of the Aztec diamond (Lem. \ref{lem: gen_tower_AD_coord}).  This allows us to use known edge probability integral formulas for the Aztec diamond to study the generalized tower graphs. Our main result is that the perfect t-embeddings of the generalized tower graphs approximating a function $f$ satisfy all the technical assumptions in Thm. \ref{thm:Main_thm_CLR21}. In Section \ref{sect: rigidity}, we prove the technical \emph{rigidity assumptions} \ref{assumption: Rigidity}) for the perfect t-embedding for the \emph{generalized tower graph}. 

\textbf{Acknowledgments.} The authors would like to thank Tomas Berggren, Matthew Nicoletti, Leonid Petrov, and Marianna Russkikh for insightful discussions. We are also grateful for the hospitality during the long program \emph{Geometry, Statistical Mechanics, and Integrability} at Institute for Pure and Applied Mathematics, supported by the
National Science Foundation (Grant No. DMS-1925919), where this work was started.  HTV is supported by the David G. Bourgin Departmental Mathematics Fellowship. DK is supported by the NSF RTG grant DMS 1937241.

\section{Background on Dimer Models and t-Embeddings}
\label{sect: background}
In this section, we follow the presentation of the background material in \cite{tower_AD_perfect} and \cite{chelkak2021fluctuations}. All of the material in this section is stated precisely with the same notations and conventions as in \cite{tower_AD_perfect}. Later on, we will clarify and provide some new notation for the family of graphs we are considering. We also suggested to the readers \cite{chelkak2021bipartitedimermodelperfect,MR4588721} for the motivation and many technical details behind t-embedding and the detail of the proof of Thm. \ref{thm:Main_thm_CLR21} that motivates this project. 

\subsection{t-embeddings and perfect t-embeddings}
Consider the graph $G$ and add a vertex $v_{\text {out }}$ to the graph by connecting it with all vertices of the outer face of ${G}$. Define \emph{augmented dual} ${G}^*$ of $G$ as the dual graph of ${G} \cup v_{\text {out }}$.
For the rest of this paper, when we mention $G^*$, we will always mean the augmented dual if we do not specify otherwise.
By an embedding of $G$ on the complex plane we mean an embedding of the augmented dual graph ${G}^*$ on the complex plane, that is, an assignment to each vertex of $G^*$ a complex number. Since we only consider bipartite graphs $G$, we will always consider the faces of the embedding to be colored black or white depending on whether they correspond to black or white vertex of $G$. A \emph{t-embedding} of $G$ is an embedding on the complex plane satisfying certain properties we now state. 

\begin{figure}
    \centering
    \includegraphics[width=\linewidth]{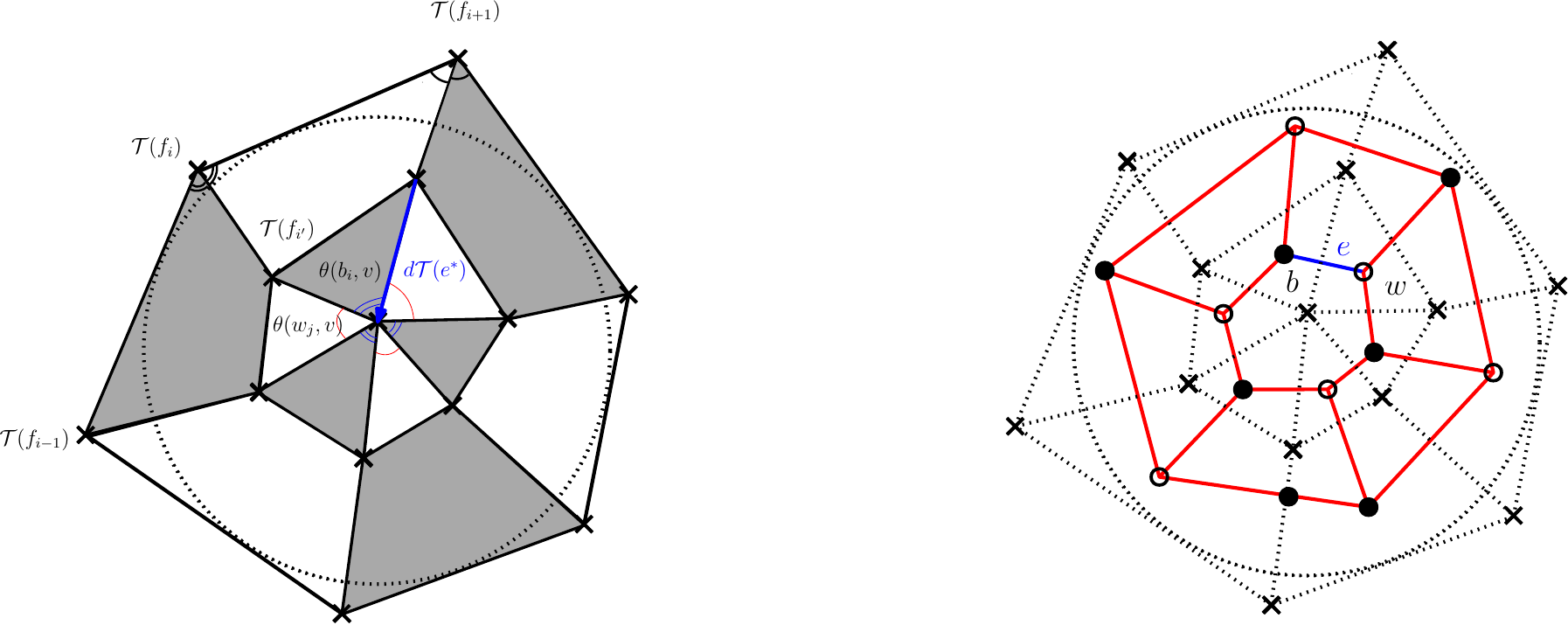}
    \caption{An example of a perfect t-embedding.}
    \label{fig:perfect_t_embedding_defn_example}
\end{figure}

\begin{defn}[\emph{t-embedding}] 
\label{def:t_embedding}
     Given weighted planar bipartite graph $({G}, \nu)$, a t-embedding $\mathcal{T}\left({G}^*\right)$ is a proper embedding of an augmented dual graph $G^*$ such that the following conditions are satisfied:
     \begin{enumerate}
         \item The sum of the angles at each inner vertex of $\mathcal{T}\left({G}^*\right)$ at the corners corresponding to black faces is equal to $\pi$ (and similarly for white faces),
         \[
\sum_{i}\theta(b_i,v)=\sum_j\theta(w_j,v)=\pi
         \]
         \item The geometric weights (dual edge lengths) $\left|\mathcal{T}\left(\nu_1^*\right)-\mathcal{T}\left(\nu_2^*\right)\right|$ are gauge equivalent to $\nu_e$, where $v_{1,2}^*$ are vertices of the dual edge $e^*$.
     \end{enumerate}
\end{defn}
\noindent From here on out, we will denote $d \mathcal{T}\left(bw^*\right):=\mathcal{T}\left(v^{\prime}\right)-\mathcal{T}(v)$, where $vv'$ is an oriented edge on the complex plane such that $w$ is on the left. 

The existence of t-embeddings is guaranteed for some families of graphs.
\begin{thm}[\cite{circle_pattern}]
    \label{circle_pattern_t_embedding_existence}
    t-embeddings of the (augmented) dual graph $G^*$ exist at least in the following case: 
    \begin{enumerate}
        \item $(G, \nu)$ is a non-degenerate bipartite finite weighted graph admitting a dimer cover and with outer face of degree $4$.  
        \item  $(G, \nu)$ is a infinite doubly periodic weighted bipartite graph, equipped with an equivalence class of doubly periodic edge weights.
    \end{enumerate}
\end{thm}

In this paper, we consider a more specific type of t-embedding which is more susceptible to boundary conditions.
\begin{defn}[\emph{Perfect t-embedding}]
\label{defn: perfect_t_embedding}
	A t-embedding is perfect if the outer face of $\cT(G^*)$ is a tangential polygon to a circle and all the non-boundary edges adjacent to boundary vertices lie on the bisectors of the corresponding angle. See Figure \ref{fig:perfect_t_embedding_defn_example}.  
\end{defn}

\subsection{Origami Maps}
For each oriented edge, let $bw^*$ be the corresponding edge of $G^*$ oriented so that $w$ is to the left.  To each t-embedding $\cT(G^*)$ one can associate the \emph{origami map} $\cO: G^*\to \C$. 

\begin{defn}
    \label{defn:origami_sqrt}
     A function $\eta: V(\mathcal{G})=W \cup B \rightarrow \mathbb{T}$ where $\mathbb T$ denotes the unit circle, is said to be an origami square root function if it satisfies the identity
$$
\bar{\eta}_b \bar{\eta}_w=\frac{d \mathcal{T}\left(bw^*\right)}{\left|d \mathcal{T}\left(bw^*\right)\right|}
$$
for all pairs $(w,b)$ of white and black neighboring faces of $\mathcal{G}^*$. 
\end{defn}

\begin{defn}
    \label{def:origami_map}
    The origami map $\cO$ is the primitive of the piecewise differential form $d\cO$ defined by 
    \begin{equation}\label{eq:origami_diff_form}
        d\cO(z)=\begin{cases}
            \eta_w^2dz & \textrm{if}\ z \in \cT(w)\\
            \bar{\eta_b}^2d\bar z & \textrm{if} \ z \in \cT(b)
        \end{cases}.
    \end{equation}
\end{defn}
Let us describe the geometric interpretation for these definitions. One can think of the map $z \mapsto \mathcal O(z)$ as folding along the edges of $G^*$ passing through a face path connecting the root face $b_0$ to some other face corresponding to some vertex $v \in V(G)$. This procedure is locally consistent due to the angle condition in the Def. \ref{def:t_embedding}. That is, the order of folding preserves the image of $\mathcal O$, see Figure \ref{fig:origami_map}. The primitive of $d\cO(z)$ is defined up to a global additive constant and thus one can set $\eta_{w_0}=1$ and $\cO(w_0)=\cT(w_0)$. From this folding procedure, we can see that if $\cT(G^*)$ is a perfect t-embedding, then the image of the outer faces of $G^*$ obtained from the origami map by fixing one outer face as root face coincides.  

\begin{figure}
    \centering
    \includegraphics[width=0.8\linewidth]{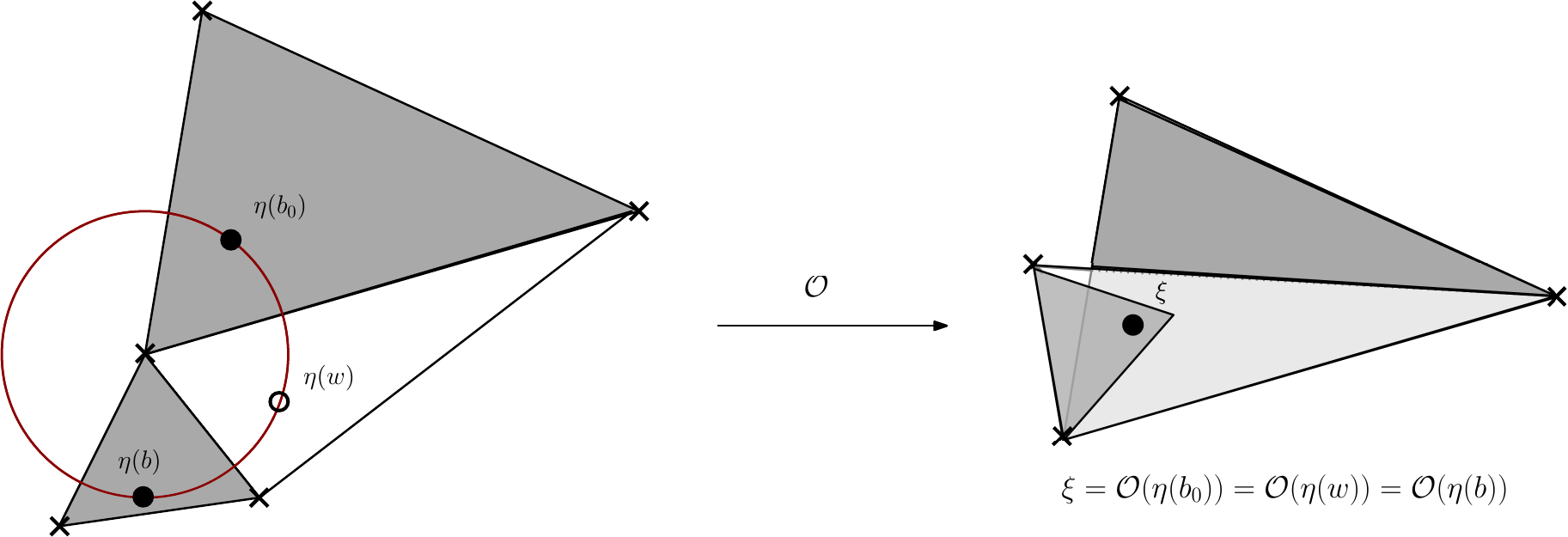}
    \caption{The image of a point under the origami map by folding the faces along adjacent edges.}
    \label{fig:origami_map}
\end{figure}

\subsection{Technical assumptions for perfect t-embeddings}

After one shows that the structure of perfect t-embedding and origami map exists, there are several technical assumptions from \cite{chelkak2021bipartitedimermodelperfect,MR4588721,chelkak2021fluctuations} that one must verify in order to conclude the convergence of the gradients of $k$-point correlation functions to that of the GFF. We first consider the pseudo-Riemannian manifold $\R^{2,1}$ as $\mathbb{R}^3$ equipped with the metric $d x^2+d y^2-d \vartheta^2$. We denote by $\Omega_{\mathcal{T}}$ the region covered by the union of faces of $\mathcal{T}(G^*)$. We are now ready to formulate the main theorem of \cite{MR4588721}.

\begin{thm}[\cite{MR4588721}]
    \label{thm:Main_thm_CLR21}
    Suppose there exists a sequence of perfect t-embeddings $\mathcal{T}_n$ of weighted planar bipartite graphs ${G}_n$ satisfying the following properties:
\begin{enumerate}
    \item \label{hyp: domain} The regions $\Omega_{\mathcal{T}_n}$ approximate (in the Hausdorff sense) a simply connected domain $\Omega \subset \mathbb{C}$ as $n \rightarrow \infty$.
    \item \label{hyp:convergence}We have the convergence of origami maps $\mathcal{O}_n(z) \rightarrow \vartheta(z)$, uniformly on compact subsets, for some function $\vartheta: \Omega \rightarrow \mathbb{R}$. Furthermore, the limiting graph $S_{\Omega}=\{(z, \vartheta(z))\}_{z \in \Omega}$ of $\mathcal{O}_n$ over $\mathcal{T}_n$ is a space-like surface with zero mean curvature in $\mathbb{R}^{2,1}$.
    \item \label{hyp:technical} There is a sequence of scales $\delta=\delta_n \rightarrow 0$, such that the sequence $\mathcal{T}_n$ satisfies Assumptions \ref{assumption:LIP} and \ref{assumption:exp-fat} on compact subsets.
\end{enumerate}
Let $\zeta_{m, i} \in \mathbb{D}, m=1, \ldots, k, i=1,2$ be points of the disk, and let $v_{m, i}^{(n)}$ be vertices of $\mathcal{T}_n$ which approximate the points $z\left(\zeta_{m, i}\right)$, where $(z(\zeta), \vartheta(\zeta))$ is a conformal parameterization of $S_{\Omega}$. Suppose that for $i=1,2$, the points $v_{1, i}^{(n)}, \ldots, v_{k, i}^{(n)}$ stay uniformly away from each other and from the boundary $\partial \Omega$. Then,

$$
\mathbb{E}\left[\prod_{m=1}^k\left(\bar{h}_n\left(v_{m, 2}^{(n)}\right)-\bar{h}_n\left(v_{m, 1}^{(n)}\right)\right)\right] \rightarrow \sum_{r_1, \ldots, r_k \in\{1,2\}}(-1)^{r_1+\cdots+r_k} G_k\left(\zeta_{1, r_1}, \ldots, \zeta_{k, r_k}\right)
$$
where $\bar{h}_n(v)=h_n(v)-\mathbb{E}\left[h_n(v)\right]$ denotes the mean subtracted height function on $\mathcal{T}_n$, and $G_k$ is the $k$-point correlation function of the Gaussian free field on $\mathbb{D}$.

\end{thm}

The two assumptions, Assumptions \ref{assumption:LIP} and \ref{assumption:exp-fat}, in the theorem are called assumptions $\operatorname{LIP}$ and $\operatorname{EXP-FAT}$, respectively, in \cite{MR4588721}. Verifying these assumptions will be the main goal of this paper.   

\begin{defn}
    \label{defn:Lipschitz_defn}
    Given constants $\delta>0$ and $\kappa \in(0,1)$, we say that a t-embedding $\mathcal{T}$ satisfies the assumption $\operatorname{LIP}(\kappa, \delta)$ on a compact subset $K \subset \Omega$ if
$$
\left|\mathcal{O}\left(z^{\prime}\right)-\mathcal{O}(z)\right| \leq \kappa\left|z^{\prime}-z\right|
$$
for all $z, z^{\prime} \in K$ such that $\left|z^{\prime}-z\right| \geq \delta$.
\end{defn}

\begin{assumption}[LIP assumption, \cite{MR4588721}]
\label{assumption:LIP}
    Given a sequence $\delta=\delta_n \rightarrow 0$, we say that a sequence of t-embeddings $\mathcal{T}_n$ satisfies the Lipschitz assumption on compact subsets if for each compact $\mathcal{K} \subset \Omega$ there exists $\kappa \in(0,1)$ such that for $n$ large enough $\mathcal{T}_n$ satisfies $\operatorname{LIP}(\kappa, \delta)$ on $\mathcal{K}$.
\end{assumption}
Intuitively speaking, Def. \ref{defn:Lipschitz_defn} implies that the function does not increase in length with respect to the Euclidean norm at the discrete level. Thus, this assumption implies that the origami map is a 1-Lipschitz function with constant strictly less than 1. 

The second assumption is a non-degeneracy condition on the faces of the sequence of t-embeddings, which roughly asserts that for almost every face, the radius of the largest circle which can be inscribed in the face cannot decay exponentially fast as $n \rightarrow \infty$. We say that a face of $\mathcal{T}$ is $\rho$-fat if there exists a disk of radius $\rho$ contained in the face. To formulate the second assumption, we also need to recall the definition of a splitting of a t-embedding introduced in \cite{MR4588721}. A \emph{splitting} $\mathcal{T}_{\text {spl }}^{\circ}$ (resp. $\mathcal{T}_{\text {spl }}^{\bullet}$ ) of $\mathcal{T}$ is obtained from $\mathcal{T}$ by adding diagonals to each white (resp. black) face of degree larger than 3 such that all of these faces are split into triangles.

\begin{assumption}[EXP-FAT assumption, \cite{MR4588721}]
    \label{assumption:exp-fat}
    Given a sequence $\delta=\delta_n \rightarrow 0$, we say that a sequence of t-embeddings $\mathcal{T}_n$ satisfies the assumption EXP-FAT( $\delta$ ) on compact subsets if for each compact $\mathcal{K} \subset \Omega$, there exists some sequence $\delta^{\prime}=\delta_n^{\prime} \rightarrow 0$ such that:
    \begin{itemize}
        \item There exist splittings $\left(\mathcal{T}_n\right)_{\text {spl }}^{\bullet}$ such that after removing all $\exp \left(-\delta^{\prime} \delta^{-1}\right)$-fat white faces and black triangles, the size of any remaining vertex connected component in $\mathcal{K}$ converges to 0 as $n \rightarrow \infty$.
        \item There exists splittings $\left(\mathcal{T}_n\right)_{\mathrm{spl}}^{\circ}$ satisfying the same condition.
    \end{itemize}
\end{assumption}

In \cite{tower_AD_perfect}, the authors provide a stronger but more straightforward condition, called the \emph{rigidity assumption}, that would imply the validity of the two assumptions $\operatorname{LIP}$ and $\operatorname{EXP-FAT}$ with the scaling $\delta=\delta_n=\frac{\log n}{n}$.  

\begin{assumption}[Rigidity Assumption, \cite{tower_AD_perfect}]
    \label{assumption: Rigidity}
   Given a compact set $\mathcal{K} \subset \Omega$, there exist positive constants $N_{\mathcal{K}}, C_{\mathcal{K}}$ and $\varepsilon_{\mathcal{K}}$ which only depend on $\mathcal{K}$, such that for all pairs of adjacent vertices $v, v^{\prime}$ of the dual graph $\mathcal{G}_n^*$ such that both $\mathcal{T}_n(v)$ and $\mathcal{T}_n\left(v^{\prime}\right)$ are contained in $\mathcal{K}$ we have
$$
\frac{\mu_n}{C_{\mathcal{K}}} \leq\left|\mathcal{T}_n\left(v^{\prime}\right)-\mathcal{T}_n(v)\right| \leq \mu_n C_{\mathcal{K}}
$$
for all $n>N_{\mathcal{K}}$. In addition, the angles of the faces of the perfect t-embedding inside $\mathcal{K}$ are contained in $\left(\varepsilon_{\mathcal{K}}, \pi-\varepsilon_{\mathcal{K}}\right)$ for all $n>N_{\mathcal{K}}$.
\end{assumption}

We refer to Section 6 of \cite{tower_AD_perfect} for the detail discussion of the \emph{rigidity assumption}. Let us now give some background and motivation on these technical assumptions LIP and EXP-FAT, which are given in a much more detail in the introduction of \cite{MR4588721}. In \cite{Ken00, Ken01}, Kenyon used the fact that the entries of the inverse Kasteleyn matrix (also known as the coupling function in \cite{Ken00}) of the dimer model on the square grid satisfy a discrete version of the Cauchy-Riemann equation to prove the convergence of the height fluctuations to the Gaussian Free Field for the dimer models on Temperleyan domains (domains, intuitively speaking, do not possess frozen regions in the scaling limit). However, for other cases, the scaling limits possess a more complicated complex structure, first witnessed in \cite{CKP} by Cohn-Kenyon-Propp and later described in detail by Kenyon-Okounkov in \cite{OK1}. This is due to the fact that for many cases, the entries of the inverse Kasteleyn matrix are not uniformly bounded and do not converge as the mesh $\delta \to 0$. These assumptions are steps toward developing a framework to remove the exponential growth of the model but still be able to utilize the discrete complex analysis toolbox to study discrete equations arising from the non-trivial complex structures. 

Currently, there are several instances where perfect t-embeddings are found satisfying the properties in Thm. \ref{thm:Main_thm_CLR21} \cite{tower_AD_perfect,berggren2024perfecttembeddingslozengetilings,berggren2025perfecttembeddingsdoublyperiodic}. The goal of the current work is to extend this to a large family of graphs.  Include in our family of graphs will be the Aztec diamond and the tower graph, defined by Borodin and Ferrari in \cite{MR3821250}.

\subsection{Convergence of t-embeddings and origami maps}
For the discussion in this section, it is more convenient to consider the rotated version of the origami map. Let us assume that a perfect t-embedding $\cT_n$ exists, with boundary contour $C_\lozenge = \{z=x+\mathrm i y:|x|+|y|=1 \}$. For the origami map $\cO_n$ corresponding to folding faces of $\cT_n$, define,
$$
\mathcal{O}_n^{\prime}:=\mathrm{e}^{i \frac{\pi}{4}}\left(\mathcal{O}_n-\frac{1+i}{2}\right),
$$
which is a composition of a translation and rotation of the origami map $\mathcal{O}_n$ defined in Section \ref{sect: t_embedding}. We define $S_{\lozenge}$ as the maximal surface in $\R^{2,1}$, in the sense that $S_\lozenge$ locally maximizes the area in the Minkowski space $\mathbb{R}^{2,1}$ with the boundary contour $C_{\lozenge}$. This is the unique space-like surface in $\mathbb{R}^{2,1}$ with this boundary contour, which has vanishing mean curvature, see \cite{Kobayashi83}.
The second assumption in Thm.  \ref{thm:Main_thm_CLR21} is that the sequence of t-embeddings and origami maps $(\cT_n,\cO_n )$  converges to 
this maximal surface in the Minkowski space $\R^{2,1}$.

\section{General Shuffling}
\label{sect: gen_shuffling}
In this section, we define the class of graphs that we will consider. These graphs will be constructed by generalizing the shuffling algorithm used to construct the Aztec diamond graph \cite{MR1990768}. Included in our family of graphs will be the Aztec diamond as well as the tower graph studied in \cite{tower_AD_perfect,nicoletti2022localstatisticsshufflingdimers,MR3821250}. The building block of this construction are certain local transformations applied to the dimer model graph $G$, which preserve the partition function up to an overall factor. In words, these transformations are: 
\begin{enumerate}
\item combining a double edge into a single edge or vice-versa,
\item adding or contracting a degree 2 vertex,
\item or performing a spider move. 
\end{enumerate}
The description of how these operations change the graph and its weights is shown in Figure \ref{fig:elt_moves}. We will also need the additional operation of applying a gauge to the weights. Let $W$ be the set of white vertices of our graphs and $B$ the set of black vertices.  We say that two choices of edge weights, $w$ and $\widetilde w$, are \emph{gauge equivalent} if there exist functions $F:W\to \R_{>0}$ and $G:W\to \R_{>0}$ such that for every edge $e=(wb)$ we have
\[
\widetilde w(e) =  w(e) F(w)G(b).
\]
Note that gauge-equivalent weights have the same partition function up to an overall factor depending on the gauge and thus define the same dimer model.

\begin{figure}
    \centering
    \includegraphics[width=0.9\linewidth]{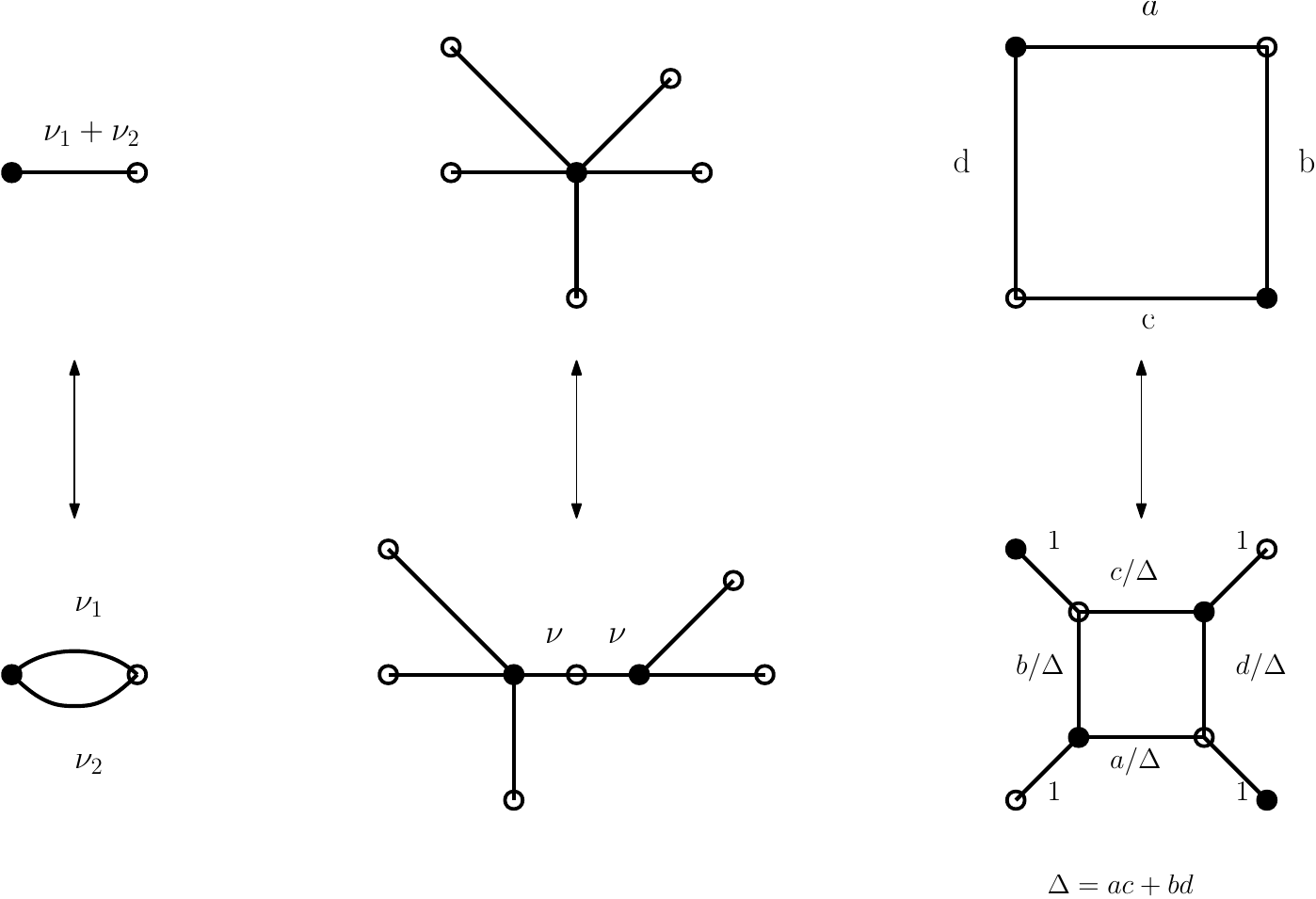}
    \caption{Elementary moves on $G$ along with their adjusted weights. Note that going from the top row to the bottom the leftmost transformation the partition function is unchanged, for the middle transformation the partition function gains an overall factor of $\nu$, and in the rightmost transformations (spider move) the partition function gains an overall factor of $\Delta^{-1}$.}
    \label{fig:elt_moves}
\end{figure}

\subsection{The Aztec Diamond}\label{subsect:AD}
Here we recall the shuffling algorithm for the Aztec diamond. Consider the square lattice $\left(\mathbb{Z}+\frac{1}{2}\right)^2$ where we label the faces by pairs of integers $(j,k)\in \mathbb{Z}^2$. We can assign a bipartite coloring to the vertices by choosing the top-left corner of the face at $(0,0)$ to be black. We define the \emph{Aztec diamond} of rank $n$ to be the set of faces $(j,k)$ such that $|j| + |k| \le n-1$ and denote it by $AD_{n}$. By convention we let $AD_0$ be the empty graph. Figure \ref{fig:AD} shows the Aztec diamond of rank 4. For each $c\in\mathbf{Z}$, we call the set of faces such that $j+k = c$ the $c$-th \emph{column} of the graph.

\begin{figure}
    \centering
   \includegraphics{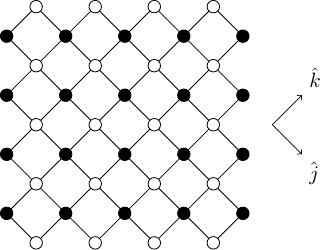}
    \caption{Aztec diamond of rank 4. We draw the graph such that the axis are rotated by $\pi/2$ so that faces such that $j+k$ is constant form vertical columns. }\label{fig:AD}
\end{figure}

It is shown in \cite{MR1990768} that one can construct $AD_{n}$ from $AD_{n-1}$ using what is known as the domino shuffling algorithm. By iterating this algorithm, one may start from $AD_0$, construct $AD_1$, and so on, until they reach $AD_{n}$. In terms of the dimer model, we present the algorithm below.
\begin{center}
 \textbf{Shuffling algorithm for the Aztec diamond}  
 \begin{enumerate}
 \item Start at rank $0$, with no faces or edges included in the graph.
 \item To get from rank $m$ to rank $m+1$ do the following:
 \begin{enumerate}
     \item \textbf{Decorate the boundary:} Add the edges and vertices around the faces located at $|j|+|k|=m$ to the graph. To the newly added vertices of degree $2$, add a new edge connecting to a new vertex of degree $1$. Give all the newly added edges weight 1. 
     \item \textbf{Shuffle:} Perform a spider move at all faces such that $m-|j|-|k|\mod 2 = 0$. 
     \item \textbf{Contract:} Contract all degree-$2$ vertices.
     \item \textbf{Gauge:} Multiply the weight of the edges adjacent to each white vertex by 2.
 \end{enumerate}
 This gives an Aztec diamond of rank $m+1$ with uniform weights.
 \item Repeat step 2 until you reach an Aztec diamond of rank $n$.
 \end{enumerate}
\end{center}
An example of step 2 can be seen in Figure \ref{fig:ADShuffle}.

Similar shuffling algorithms can be used to generate other families of graphs. Figure \ref{fig:towergraphshuffling} shows such an algorithm for the so-called \emph{tower graph} defined in \cite{MR3821250} (see \cite{tower_AD_perfect}). 
\begin{defn}[\emph{Tower graph}]
    \label{def: tower graph}
    A tower graph of size $n$ is a union of $3 n-1$ columns of faces with $n$ columns of hexagonal faces separated by two columns of square faces. The first column contain $n$ hexagons, and the last column has $2 n$ squares. For $p \geq 0$ the ($3 p+1$)st column contains $n+p$ hexagons and it is adjacent to a column of $n+p-1$ squares on the left (if $p>0$) and to a column of $n+p+1$ squares on the right. See Figure \ref{fig:towergraphshuffling} for examples of the rank 2 and 3 tower graphs.
\end{defn}

\begin{figure}
    \centering
    \includegraphics{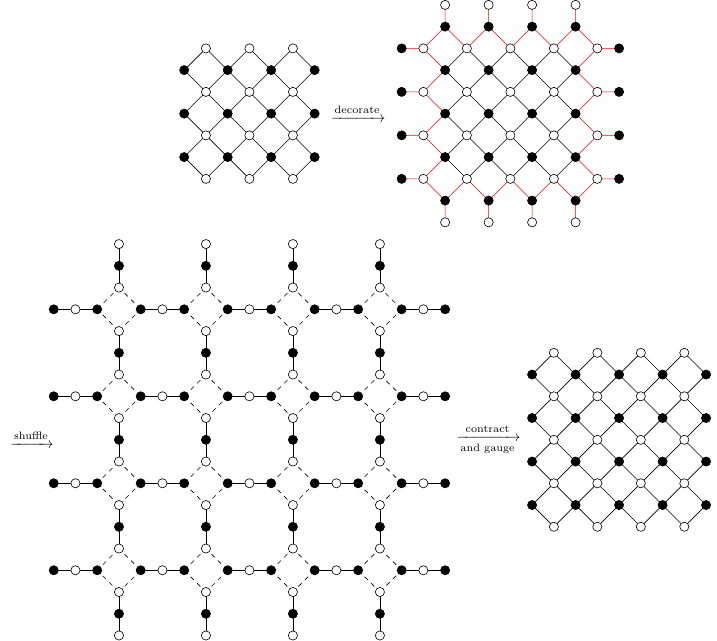}
        \caption{Steps in the shuffling algorithm from the Aztec diamond of rank 3 to the Aztec diamond of rank 4. The red edges are those added as decoration to the boundary. The dashed edges have weight $1/2$ while all other edges have weight $1$.}\label{fig:ADShuffle}
\end{figure}

\begin{figure}
    \centering
    \includegraphics{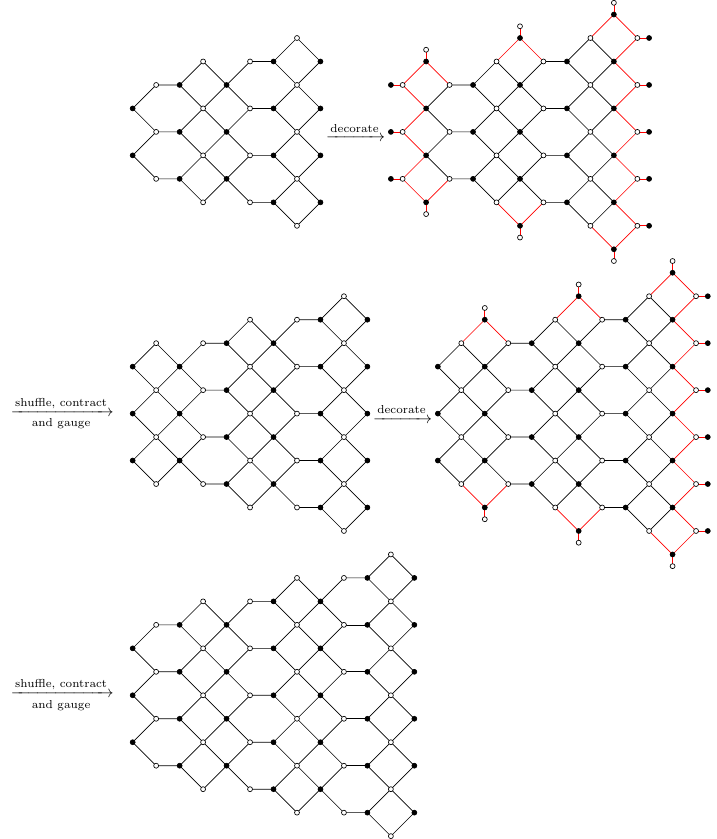}
    \caption{Generating the rank-3 tower graph from the rank-2 tower graph using a shuffling algorithm. During the shuffle step, a spider move is done at every face of the decorated columns. Unlike the shuffling algorithm for the Aztec diamond, the tower graph requires two rounds of decorating then shuffling.} \label{fig:towergraphshuffling}
\end{figure}

\subsection{Column-wise Shuffling}
We would like to highlight two features of the shuffling algorithm for the Aztec diamond. First, during each step, we always shuffle all faces in a column of the graph. Moreover, the number of faces in a column increases by two after shuffling. Second, we can shuffle non-adjacent faces in any order. In particular, the different columns of the Aztec diamond can be shuffled in any order and the resulting weighted graph would be the same. We will use the above observation to describe an alternative formulation of the shuffling algorithm that will be useful for us in what follows.

For this alternative description, we will consider shuffling only at a single column at a time, and so will need to understand how this affects our graphs. To that end, we say that the \emph{height} of a column is the number of faces of the graph it contains, and we call a column \emph{shufflable} if it is a column of square faces such that its height is less than its neighboring columns. Shuffling at a column takes the form shown in Figure \ref{fig:colShuf}. Note the height of the column increases by two after the shuffle.

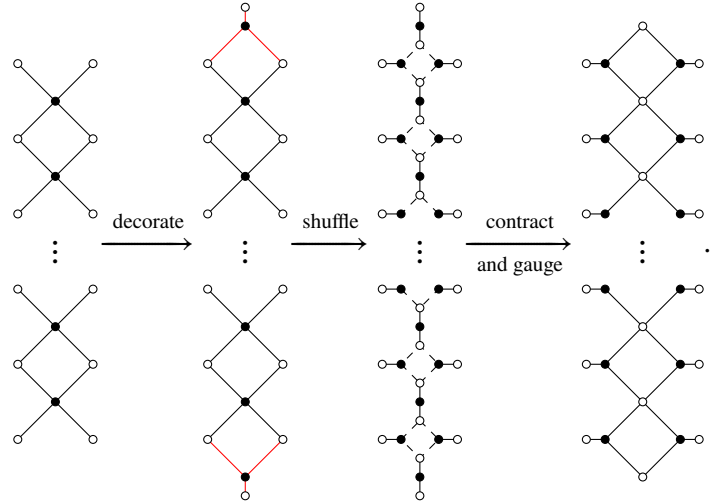
\begin{figure}
\begin{center}

\[
    \begin{tikzpicture}[baseline = (current bounding box).center,scale=0.5]
        \draw (0,0)--(1,1)--(0,2)--(-1,1)--(0,0);
        \draw (0,2)--(1,3);
        \draw (0,2)--(-1,3);
        \draw (0,0)--(1,-1);
        \draw (0,0)--(-1,-1);
        \node at (0,4) {$\vdots$};
        \draw (1,5)--(0,6)--(1,7)--(0,8)--(-1,7)--(0,6)--(-1,5);
        \draw (0,8)--(1,9);
        \draw (0,8)--(-1,9);

        \draw[black, fill=black] (0,0) circle (3pt);
        \draw[black, fill=black] (0,2) circle (3pt);
        \draw[black, fill=black] (0,6) circle (3pt);
        \draw[black, fill=black] (0,8) circle (3pt);
        \draw[black, fill=white] (1,-1) circle (3pt);
        \draw[black, fill=white] (-1,-1) circle (3pt);
        \draw[black, fill=white] (1,1) circle (3pt);
        \draw[black, fill=white] (-1,1) circle (3pt);
        \draw[black, fill=white] (1,3) circle (3pt);
        \draw[black, fill=white] (-1,3) circle (3pt);
        \draw[black, fill=white] (1,5) circle (3pt);
        \draw[black, fill=white] (-1,5) circle (3pt);
        \draw[black, fill=white] (1,7) circle (3pt);
        \draw[black, fill=white] (-1,7) circle (3pt);
        \draw[black, fill=white] (1,9) circle (3pt);
        \draw[black, fill=white] (-1,9) circle (3pt);
    \end{tikzpicture}
    \xrightarrow[]{\text{decorate}}
    \begin{tikzpicture}[baseline = (current bounding box).center,scale=0.5]
        \draw (0,0)--(1,1)--(0,2)--(-1,1)--(0,0);
        \draw (0,2)--(1,3);
        \draw (0,2)--(-1,3);
        \draw (0,0)--(1,-1);
        \draw (0,0)--(-1,-1);
        \node at (0,4) {$\vdots$};
        \draw (1,5)--(0,6)--(1,7)--(0,8)--(-1,7)--(0,6)--(-1,5);
        \draw (0,8)--(1,9);
        \draw (0,8)--(-1,9);
        \draw[red] (-1,-1)--(0,-2)--(1,-1);
        \draw[red] (0,-2)--(0,-2.5);
        \draw[red] (-1,9)--(0,10)--(1,9);
        \draw[red] (0,10)--(0,10.5);

        \draw[black, fill=black] (0,-2) circle (3pt);
        \draw[black, fill=black] (0,0) circle (3pt);
        \draw[black, fill=black] (0,2) circle (3pt);
        \draw[black, fill=black] (0,6) circle (3pt);
        \draw[black, fill=black] (0,8) circle (3pt);
        \draw[black, fill=black] (0,10) circle (3pt);
        \draw[black, fill=white] (0,-2.5) circle (3pt);
        \draw[black, fill=white] (1,-1) circle (3pt);
        \draw[black, fill=white] (-1,-1) circle (3pt);
        \draw[black, fill=white] (1,1) circle (3pt);
        \draw[black, fill=white] (-1,1) circle (3pt);
        \draw[black, fill=white] (1,3) circle (3pt);
        \draw[black, fill=white] (-1,3) circle (3pt);
        \draw[black, fill=white] (1,5) circle (3pt);
        \draw[black, fill=white] (-1,5) circle (3pt);
        \draw[black, fill=white] (1,7) circle (3pt);
        \draw[black, fill=white] (-1,7) circle (3pt);
        \draw[black, fill=white] (1,9) circle (3pt);
        \draw[black, fill=white] (-1,9) circle (3pt);
        \draw[black, fill=white] (0,10.5) circle (3pt);
    \end{tikzpicture}
    \xrightarrow[]{\text{shuffle}}
    \begin{tikzpicture}[baseline = (current bounding box).center,scale=0.5]
        \draw (0,-2)--(0,-2.5);
        \draw[dashed] (0,-0.5)--(0.5,-1)--(0,-1.5)--(-0.5,-1)--(0,-0.5);
        \draw (0,-0.5)--(0,0);
        \draw (0,-1.5)--(0,-2);
        \draw (0.5,-1)--(1,-1);
        \draw (-0.5,-1)--(-1,-1);
        \draw[dashed] (0,0.5)--(0.5,1)--(0,1.5)--(-0.5,1)--(0,0.5);
        \draw (0,1.5)--(0,2);
        \draw (0,0.5)--(0,0);
        \draw (0.5,1)--(1,1);
        \draw (-0.5,1)--(-1,1);
        \draw[dashed] (-0.5,3)--(0,2.5)--(0.5,3);
        \draw (0,2.5)--(0,2);
        \draw (0.5,3)--(1,3);
        \draw (-0.5,3)--(-1,3);
        
        \draw[dashed] (0.5,5)--(0,5.5)--(-0.5,5);
        \draw (0,5.5)--(0,6);
        \draw (0.5,5)--(1,5);
        \draw (-0.5,5)--(-1,5);
        \draw[dashed] (0,6.5)--(0.5,7)--(0,7.5)--(-0.5,7)--(0,6.5);
        \draw (0,7.5)--(0,8);
        \draw (0,6.5)--(0,6);
        \draw (0.5,7)--(1,7);
        \draw (-0.5,7)--(-1,7);
        \draw[dashed] (0,8.5)--(0.5,9)--(0,9.5)--(-0.5,9)--(0,8.5);
        \draw (0,9.5)--(0,10);
        \draw (0,8.5)--(0,8);
        \draw (0.5,9)--(1,9);
        \draw (-0.5,9)--(-1,9);
        \draw (0,10)--(0,10.5);
        \node at (0,4) {$\vdots$};

        \draw[black, fill=black] (0,-2) circle (3pt);
        \draw[black, fill=black] (0.5,-1) circle (3pt);
        \draw[black, fill=black] (-0.5,-1) circle (3pt);
        \draw[black, fill=black] (0,0) circle (3pt);
        \draw[black, fill=black] (0.5,1) circle (3pt);
        \draw[black, fill=black] (-0.5,1) circle (3pt);
        \draw[black, fill=black] (0,2) circle (3pt);
        \draw[black, fill=black] (0.5,3) circle (3pt);
        \draw[black, fill=black] (-0.5,3) circle (3pt);
        \draw[black, fill=black] (0.5,5) circle (3pt);
        \draw[black, fill=black] (-0.5,5) circle (3pt);
        \draw[black, fill=black] (0,6) circle (3pt);
        \draw[black, fill=black] (0.5,7) circle (3pt);
        \draw[black, fill=black] (-0.5,7) circle (3pt);
        \draw[black, fill=black] (0,8) circle (3pt);
        \draw[black, fill=black] (0.5,9) circle (3pt);
        \draw[black, fill=black] (-0.5,9) circle (3pt);
        \draw[black, fill=black] (0,10) circle (3pt);

        \draw[black, fill=white] (0,-2.5) circle (3pt);
        \draw[black, fill=white] (0,-1.5) circle (3pt);
        \draw[black, fill=white] (1,-1) circle (3pt);
        \draw[black, fill=white] (-1,-1) circle (3pt);
        \draw[black, fill=white] (0,-0.5) circle (3pt);
        \draw[black, fill=white] (0,0.5) circle (3pt);
        \draw[black, fill=white] (1,1) circle (3pt);
        \draw[black, fill=white] (-1,1) circle (3pt);
        \draw[black, fill=white] (0,1.5) circle (3pt);
        \draw[black, fill=white] (0,2.5) circle (3pt);
        \draw[black, fill=white] (1,3) circle (3pt);
        \draw[black, fill=white] (-1,3) circle (3pt);
        \draw[black, fill=white] (1,5) circle (3pt);
        \draw[black, fill=white] (-1,5) circle (3pt);
        \draw[black, fill=white] (0,5.5) circle (3pt);
        \draw[black, fill=white] (0,6.5) circle (3pt);
        \draw[black, fill=white] (1,7) circle (3pt);
        \draw[black, fill=white] (-1,7) circle (3pt);
        \draw[black, fill=white] (0,7.5) circle (3pt);
        \draw[black, fill=white] (0,8.5) circle (3pt);
        \draw[black, fill=white] (1,9) circle (3pt);
        \draw[black, fill=white] (-1,9) circle (3pt);
        \draw[black, fill=white] (0,9.5) circle (3pt);
        \draw[black, fill=white] (0,10.5) circle (3pt);
    \end{tikzpicture}
    \xrightarrow[\text{and gauge}]{\text{contract}}
    \begin{tikzpicture}[baseline = (current bounding box).center,scale=0.5]
        \draw (0,0)--(1,1)--(0,2)--(-1,1)--(0,0);
        \draw (0,2)--(1,3);
        \draw (0,2)--(-1,3);
        \draw (0,0)--(1,-1);
        \draw (0,0)--(-1,-1);
        \node at (0,4) {$\vdots$};
        \draw (1,5)--(0,6)--(1,7)--(0,8)--(-1,7)--(0,6)--(-1,5);
        \draw (0,8)--(1,9);
        \draw (0,8)--(-1,9);
        \draw (-1,-1)--(0,-2)--(1,-1);
        \draw (-1,9)--(0,10)--(1,9);
        \draw (1,-1)--(1.5,-1);
        \draw (1,1)--(1.5,1);
        \draw (1,3)--(1.5,3);
        \draw (1,5)--(1.5,5);
        \draw (1,7)--(1.5,7);
        \draw (1,9)--(1.5,9);
        \draw (-1,-1)--(-1.5,-1);
        \draw (-1,1)--(-1.5,1);
        \draw (-1,3)--(-1.5,3);
        \draw (-1,5)--(-1.5,5);
        \draw (-1,7)--(-1.5,7);
        \draw (-1,9)--(-1.5,9);

        \draw[black, fill=white] (0,-2) circle (3pt);
        \draw[black, fill=white] (0,0) circle (3pt);
        \draw[black, fill=white] (0,2) circle (3pt);
        \draw[black, fill=white] (0,6) circle (3pt);
        \draw[black, fill=white] (0,8) circle (3pt);
        \draw[black, fill=white] (0,10) circle (3pt);
        \draw[black, fill=white] (1.5,-1) circle (3pt);
        \draw[black, fill=white] (1.5,1) circle (3pt);
        \draw[black, fill=white] (1.5,3) circle (3pt);
        \draw[black, fill=white] (1.5,5) circle (3pt);
        \draw[black, fill=white] (1.5,7) circle (3pt);
        \draw[black, fill=white] (1.5,9) circle (3pt);
        \draw[black, fill=white] (-1.5,-1) circle (3pt);
        \draw[black, fill=white] (-1.5,1) circle (3pt);
        \draw[black, fill=white] (-1.5,3) circle (3pt);
        \draw[black, fill=white] (-1.5,5) circle (3pt);
        \draw[black, fill=white] (-1.5,7) circle (3pt);
        \draw[black, fill=white] (-1.5,9) circle (3pt);
        \draw[black, fill=black] (1,-1) circle (3pt);
        \draw[black, fill=black] (-1,-1) circle (3pt);
        \draw[black, fill=black] (1,1) circle (3pt);
        \draw[black, fill=black] (-1,1) circle (3pt);
        \draw[black, fill=black] (1,3) circle (3pt);
        \draw[black, fill=black] (-1,3) circle (3pt);
        \draw[black, fill=black] (1,5) circle (3pt);
        \draw[black, fill=black] (-1,5) circle (3pt);
        \draw[black, fill=black] (1,7) circle (3pt);
        \draw[black, fill=black] (-1,7) circle (3pt);
        \draw[black, fill=black] (1,9) circle (3pt);
        \draw[black, fill=black] (-1,9) circle (3pt);
    \end{tikzpicture}.
\]

\end{center}
\caption{The steps involved in shuffling a single column of our graph.}\label{fig:colShuf}
\end{figure}

Using the single column shuffling one may check that we have the following lemma.

\begin{lemma}\label{lem:colshuf}
Shuffling at a column turns adjacent columns from squares into hexagons and vice-versa.
\end{lemma}

\begin{remark}
    Note that our graph will no longer consist of only square faces but may include hexagonal faces as well. This will not effect how we label the faces by elements of $\mathbb{Z}^2$ and we will continue to use the same coordinate system.
\end{remark}

We will also want to allow shuffling of columns on the boundary of our graph. Let $l,r\in\mathbb{Z}_{\ge 0}$ and suppose our graph includes all columns such that $-l< j+k< r$ where $(j,k)\in \mathbb{Z}^2$ are the coordinates of the faces in the column. Then we call \emph{boundary columns} the columns satisfying $j+k=r$ and $j+k=-l$, obtained by adding faces to the the left and right boundary of the graph as 
\[
\begin{tikzpicture}[baseline = (current bounding box).center,scale=0.5]
        \draw (-1,-1)--(0,0)--(-1,1)--(0,2)--(-1,3);
        \draw[red] (0,0)--(1,1)--(0,2)--(1,3);
        \draw[red] (1,1)--(1.5,1);
        \draw[red] (1,3)--(1.5,3);
        \node at (0,4) {$\vdots$};
        \draw (-1,5)--(0,6)--(-1,7)--(0,8)--(-1,9);
        \draw[red] (1,5)--(0,6)--(1,7)--(0,8);
        \draw[red] (1,5)--(1.5,5);
        \draw[red] (1,7)--(1.5,7);
   
        \draw[black, fill=black] (0,0) circle (3pt);
        \draw[black, fill=black] (0,2) circle (3pt);
        \draw[black, fill=black] (0,6) circle (3pt);
        \draw[black, fill=black] (0,8) circle (3pt);
        \draw[black, fill=white] (-1,-1) circle (3pt);
        \draw[black, fill=white] (1,1) circle (3pt);
        \draw[black, fill=white] (-1,1) circle (3pt);
        \draw[black, fill=white] (1,3) circle (3pt);
        \draw[black, fill=white] (-1,3) circle (3pt);
        \draw[black, fill=white] (1,5) circle (3pt);
        \draw[black, fill=white] (-1,5) circle (3pt);
        \draw[black, fill=white] (1,7) circle (3pt);
        \draw[black, fill=white] (-1,7) circle (3pt);
        \draw[black, fill=white] (-1,9) circle (3pt);
        \draw[black, fill=black] (1.5,1) circle (3pt);
        \draw[black, fill=black] (1.5,3) circle (3pt);
        \draw[black, fill=black] (1.5,5) circle (3pt);
        \draw[black, fill=black] (1.5,7) circle (3pt);
    \end{tikzpicture}
\]
for the right boundary and the horizontal flip for the left boundary, where the added edges are colored red. Note that there is no freedom in the placement of dimers in the boundary columns, and the boundary column will always have height two less than the adjacent non-boundary column. All other columns are called \emph{interior columns}.

We consider boundary columns to always be shufflable by decorating the column with an additional square face at the top and bottom, as well as corresponding adjacent degree one vertices, then doing a spider move at each face, contracting and gauging. This turns the boundary column into an interior column and increases the number of columns in the graph by one.  

We may give a slight variation on the shuffling algorithm of the Aztec diamond using column shuffling that highlights the fact that the order of shuffling at non-adjacent columns does not effect the end graph.
\begin{center}
 \textbf{Alternate shuffling algorithm for the Aztec diamond}  
 \begin{enumerate}
 \item Start with  rank-$0$ Aztec diamond.
 \item Choose a shufflable column whose height is not equal to the height of the same column in $AD_n$ and perform the single column shuffling.
 \item Repeat step 2 until you reach an Aztec diamond of rank $n$.
 \end{enumerate}
\end{center}

\subsection{Shuffling and Up-Down paths}\label{sect: shufupdown}

We would like to consider all graphs that can be obtained by starting with the rank-$0$ Aztec diamond and then shuffling the columns in some order until they reach the prescribed heights.
In order to describe these graphs, we first define a certain lattice path that will contain the information of the column heights of our graph. 

By an \emph{up-down path} we mean an assignment for every $c\in \mathbb{Z}$ of a height $H(c)\in\mathbb{Z}$ such that $H(c+1)-H(c)=\pm 1$. We then linearly interpolate between the integer points so that $H:\R\to\R$. By convention, we will represent the graph $AD_0$ by the up-down path such that for each $c\in \mathbb{Z}$ the height $H$ of the path at $c$ is
\[
H(x) = |x|-1.
\]
We have the following proposition relating these paths to our graphs.

\begin{prop} \label{prop:pathtographbij}
    There is a bijection between up-down paths and square-hexagon graphs satisfying the following properties:
    \begin{itemize}
        \item Suppose the graph contains columns $-l<j+k<r$. Then the height of the path at $c\in\{-l+1,-l+2\ldots,r-1\}$ is the height of column $c$ in the graph.
        \item For all $c\le -l$ we have $H(c)=-c-1$ and for all $c\ge r$ we have $H(c)=c-1$. Note that $H(-l)$ and $H(r)$ correspond to the height of the boundary columns of the graph. 
        \item For $c\in\{-l,-l+1,\ldots,r\}$ if $H(c) > H(c+1),H(c-1)$ or $H(c) < H(c+1),H(c-1)$ then column $c$ of the graph consists of square faces. If $H(c-1)<H(c)<H(c+1)$ or $H(c-1)>H(c)>H(c+1)$ then column $c$ of the graph consists of hexagonal faces.
        \item Shuffling at column $c$ corresponds to an upward corner flip of the up-down path at $c$. That is, $H(c)\mapsto H(c)+2$ and segments connecting to $H(c\pm1)$ are changed accordingly.
    \end{itemize}
\end{prop}
\begin{proof}
    Note that starting from $AD_0$ the only possible column shuffle results in $AD_1$. This agrees with the up-down path corresponding to $AD_0$ having only one possible corner flip which results in the path corresponding to $AD_1$. The proposition then follows from induction using Lem. \ref{lem:colshuf}.
\end{proof}

The up-down paths corresponding to the Aztec diamond of rank 3 and the tower graph of rank 3 are shown in Figure \ref{fig:Paths}.

\begin{remark}\label{rmk:railyard}
    The representation of the graphs as an up-down path is related to the Schur process representation as a sequence of interlacing partitions. See \cite{rail-yard,MR3821249} for details. In particular, in \cite{MR3821249} the authors represent their graph by a Young diagram called ``encoded shape". If we restrict the interlacing to only use $\{\preceq',\succeq\}$ and draw the encoded shape in Russian convention, then we recover our up-down path description with $\preceq'$ corresponding to up steps and $\succeq$ corresponding to down steps. For the examples in Figure \ref{fig:Paths}, the Aztec diamond of rank 3 corresponds to the interlacing sequence
    \[
    \emptyset \preceq' \lambda^{(1)} \succeq \lambda^{(2)} \preceq' \lambda^{(3)} \succeq \lambda^{(4)} \preceq' \lambda^{(5)} \succeq \emptyset
    \]
    while the tower graph of rank 3 corresponds to the interlacing sequence
    \[
    \emptyset \preceq' \lambda^{(1)} \preceq' \lambda^{(2)} \succeq \lambda^{(3)} \preceq' \lambda^{(4)} \preceq' \lambda^{(5)} \succeq \lambda^{(6)} \preceq' \lambda^{(7)} \preceq' \lambda^{(8)} \succeq \emptyset.
    \]
    Similar square-hexagon graphs to those we study here were studied using Schur-process techniques in \cite{boutillier2021limit}.
\end{remark}

\begin{figure}
    \centering
   \includegraphics{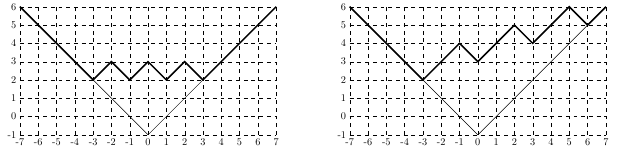}
    \caption{Left: The corresponding up-down path for the Aztec diamond of rank $3$. Right:  The corresponding up-down path for the tower graph of rank $3$.}\label{fig:Paths}
\end{figure}

\begin{figure}
    \centering
    \begin{tabular}{ccc}
    \includegraphics[width=0.4\textwidth]{quadratic0.33N15graph-eps-converted-to.pdf}
    & \quad &
    \includegraphics[width=0.4\textwidth]{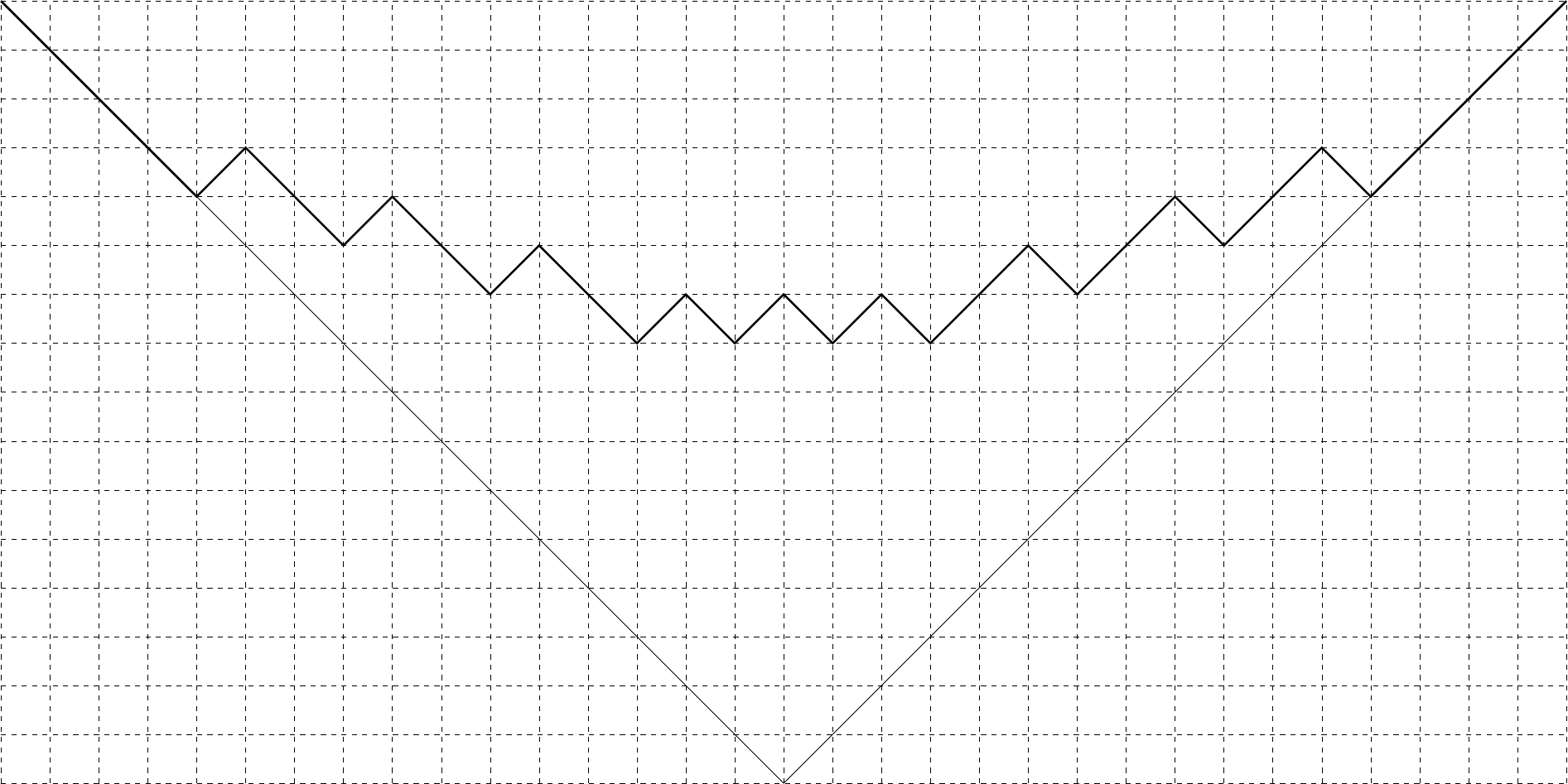}
    \end{tabular}
    \caption{The generalized tower graph of rank 15 approximating $f(x)=\frac{1}{3}x^2+\frac{2}{3}$ and the corresponding up-down path.}
    \label{fig:quadratic}
\end{figure}

\begin{figure}
    \centering
    \begin{tabular}{ccc}
    \includegraphics[width=0.4\textwidth]{sin2N15graph-eps-converted-to.pdf}
    & \quad &
    \includegraphics[width=0.4\textwidth]{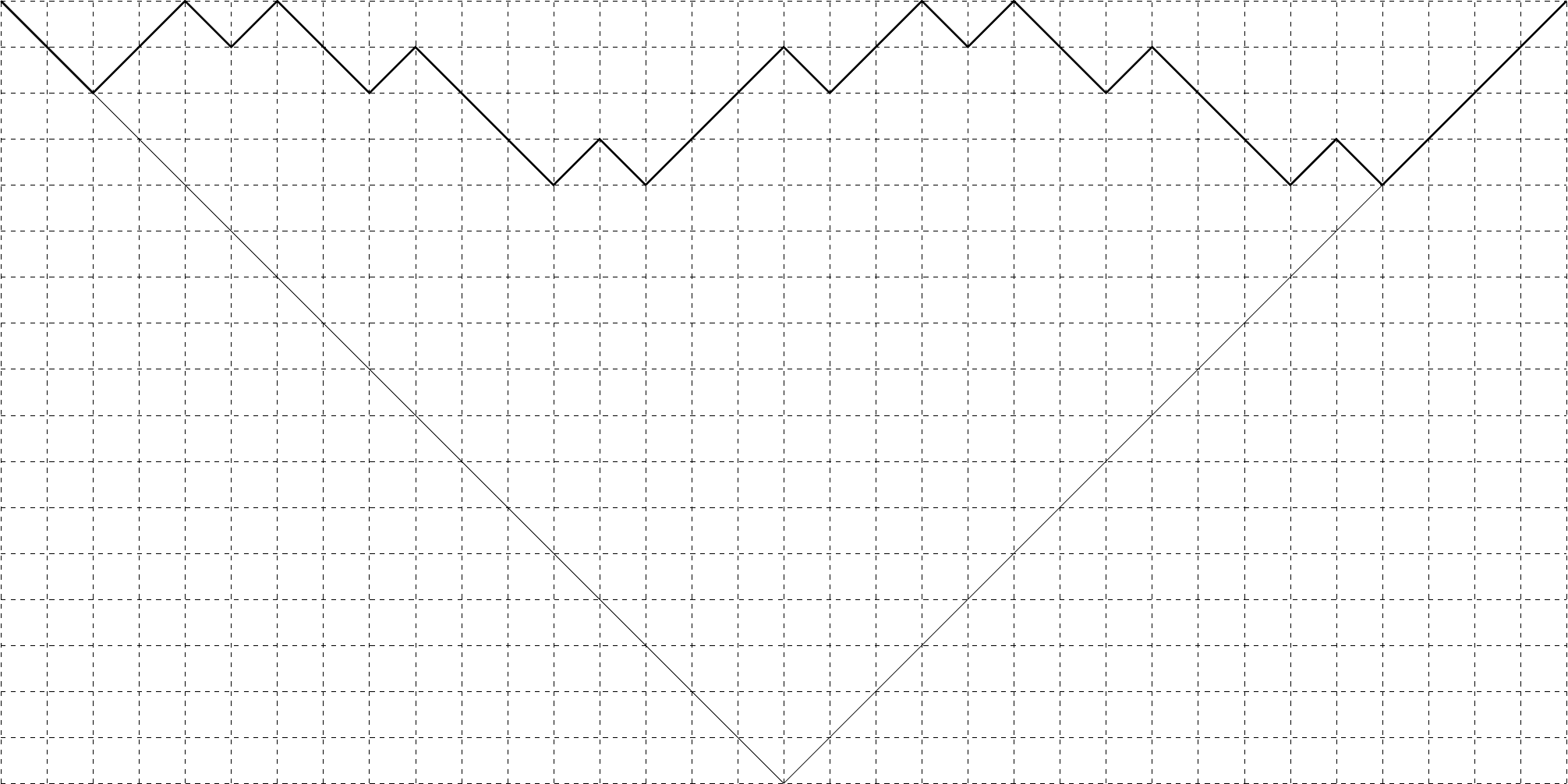}
    \end{tabular}
    \caption{The generalized tower graph of rank 15 approximating $f(x)=\frac{1}{3\pi}\sin(2\pi x) + 1$ and the corresponding up-down path.}
    \label{fig:sin}
\end{figure}

Instead of first starting with a family of graphs in mind, we may instead start with a curve $f$ and use this to define the graphs. To make this precise, consider a function $f:\R\to\R$ satisfying the conditions
\begin{itemize}
    \item $f(-1)=1$,
    \item there exists a smallest $x_*>-1$ such that $f(x_*)=x_*$,
    \item and $|f'(x)| \le 1$ for all $x\in [-1,x_*]$.
\end{itemize}
We then define $\bar f$ by
\begin{equation}\label{eq:barf}
\bar f(x) = 
\begin{cases}
    -x, & x<-1 \\
    f(x), & x\in[-1,x_*] \\
    x, & x>x_*
\end{cases}.
\end{equation}

One can put a partial order on the paths according to the rule that for paths $H$ and $H'$, $H\le H'$ if and only if $H(c)\le H'(c)$ for all $c\in\Z$. We define a sequence of up-down paths $(H_n^f)_{n\in\Z_{\ge 0}}$ approximating $f$ by letting $H_n^f$ be the largest up-down path in this partial order such that
\begin{equation}\label{eq:approxPath}
    \frac{1}{n+1}\left(H_n^f((n+1)x) + 1\right) \le \bar f(x)
\end{equation}
for all $x\in \R$. By Prop. \ref{prop:pathtographbij}, this allows us to define a sequence of graphs whose column heights approximate $f$. We denote these graphs \emph{generalized tower graphs approximating $f$}. We call $n$ the rank of such a graph. Note that we have
\begin{equation}\label{eq:error}
    \left| (n+1)\bar f(x) - H_n^f((n+1)x) + 1\right| \le 2
\end{equation}
for all $x\in \R$.

\begin{example}
    The up-down paths for the Aztec Diamond graphs $AD_n$ approximate the function $f(x)=1$. The up-down paths for the tower graphs approximate the function $f(x) = \frac{1}{3}x + \frac{4}{3}$. Figure \ref{fig:quadratic} shows the up-down path corresponding to the generalized tower graph of rank $n=15$ approximating $f(x) = \frac{1}{3}x^2+\frac{2}{3}$, while Figure \ref{fig:sin} shows the same for $f(x) = \frac{1}{3\pi}\sin(2\pi x) + 1$.
\end{example}

\begin{remark}\label{rmk:2d}
    Rather than column heights one can generalize the above to a two-dimensional version in which we consider the up-down surfaces  that encode the height $H(j,k)$ at each face. For $AD_0$, define the corresponding up-down surface heights to be 
    \[
    H_0(j,k) = \begin{cases}
        |j+k|-1,& j,k\ge 0 \text{ or } j,k \le 0 \\
        |j-k|-1, & j\ge 0,  k\le 0 \text{ or } j\le 0, k\ge 0.
    \end{cases}
    \]
    The analogue of corner flips for the surfaces are local moves such that if $H(j,k) = h-1$ and $H(j\pm1,k) = H(j,k\pm1) = h$, then we may ``flip" at $(j,k)$ and change the height at that face to from $h-1$ to $h+1$. This corresponds to shuffling at the face $(j,k)$ in the graph. Graphs constructed in this way include the $(r,s,t)$-graphs of \cite{MR4782740}. The change in height after a move of shuffling is identical to the change on the stepped initial data surface of the octahedron recurrence. One can then define families of graphs that approximate two-dimensional surfaces $f:\R^2\to\R$ similar to what we have for the generalized tower graphs. 

    The up-down paths and generalized tower graphs correspond to the case for which one always shuffles at every face in a column, so that the two-dimensional surfaces are actual constant in one direction and satisfy $H(j,k) = H(j+k)$. We restrict ourself here to the one-dimensional case for simplicity, but note that the asymptotic analysis in Section \ref{sect: rigidity} will not depend on this fact and applies equally well to the more general case.
\end{remark}

\subsection{Reduced Graphs}

In order to study the t-embeddings we need to consider shuffling on the reduced graphs. Given a graph $G$, by \emph{reduced graph} we mean the graph $G'$ obtained by contracting all degree two vertices of $G$.  The reduced graphs for the Aztec diamond of rank 3 and the tower graph of rank 3 are shown in Figure \ref{fig:reducedgraphs}.  We note that, in particular, the reduced graph $G'$ has the same number of faces as the graph $G$, so we may keep the same labeling of columns. 

In \cite{tower_AD_perfect}, the authors describe how the the shuffling algorithm translates to the reduced graph for the Aztec diamonds and tower graphs. The difference in the shuffling occurs at the boundaries, with Figures \ref{fig:reducedboundaryshuffle} and \ref{fig:reducedboundaryshuffle2} describing how to implement the shuffle at a boundary column in the reduced graph.  Note that the shuffling algorithm for reduced graphs results in the same graph as first taking the corresponding unreduced graph, shuffling, then reducing.

\begin{figure}
    \centering
    \begin{tabular}{ccc}
   \includegraphics{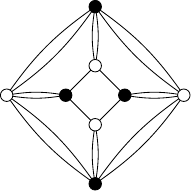}
    &
    \hspace{1cm}
    &
    \includegraphics{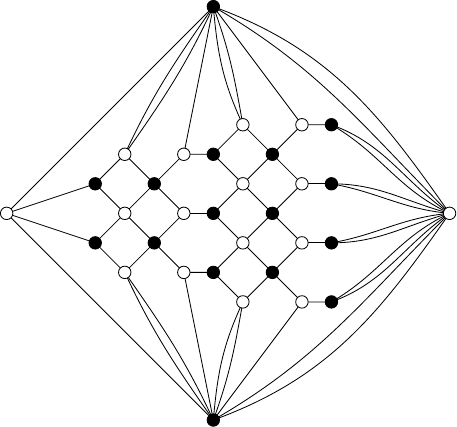}
    \end{tabular}
    \caption{Left: Reduced graph for the Aztec diamond of rank 3. Right: Reduced graph for the tower graph of rank 3.}
    \label{fig:reducedgraphs}
\end{figure}

\begin{figure}
    \centering
    \includegraphics{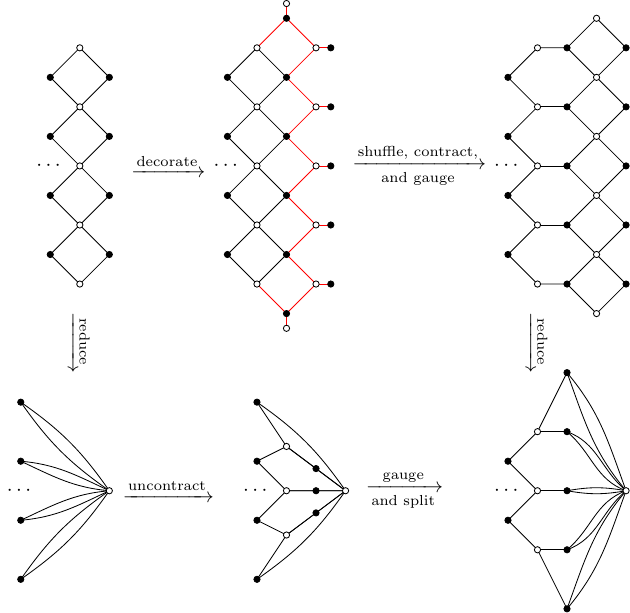} 
    \caption{Shuffling at a boundary column for which the adjacent interior column is made up of square faces. The bottom line shows how to translate the shuffling to the reduced graph. First, one repeatedly uncontracts at the white vertex, each time turning it into a two white vertices connected by a degree 2 black vertex. Then, one multiplies the weight of each edge adjacent to the original white vertex by 2 and splits them into double edges with weight 1.}
    \label{fig:reducedboundaryshuffle}
\end{figure}

\begin{figure}
    \centering \includegraphics{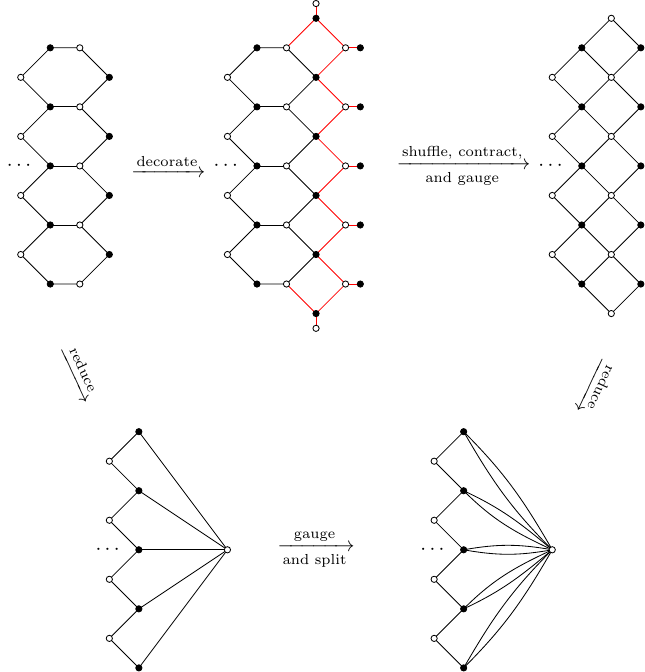} 
       \caption{Shuffling at a boundary column for which the adjacent interior column is made up of hexagonal faces. The bottom line shows how to translate the shuffling to the reduced graph. In this case, one multiplies the weight of each edge adjacent to the leftmost white vertex by 2 and splits them into double edges with weight 1.}
    \label{fig:reducedboundaryshuffle2}
\end{figure}

The following proposition follows from the reduced graph shuffling algorithm.
\begin{prop} \label{prop:boundary4}
    All the reduced generalized tower graphs have degree 4 boundary.
\end{prop}

\section{t-Embeddings and Generalized Tower Graphs}
\label{sect: t_embedding}
In this section we relate the t-embeddings of the generalized tower graphs to that of the Aztec diamond (Lem. \ref{lem: gen_tower_AD_coord}). This will serve as main lemma allowing for the asymptotic analysis in Section \ref{sect: rigidity}. To relate the various t-embeddings, we will use the connection with the shuffling algorithm. It is shown in \cite{circle_pattern} that perfect t-embedding and the origami map  are invariant under the elementary transformations in Figure  \ref{fig:elt_moves}. We refer to Section 2.14 of \cite{tower_AD_perfect} for details of the following proposition.
\begin{prop}[\cite{circle_pattern, tower_AD_perfect}]
\label{prop:t_embedding_invariant}
    Perfect t-embeddings of $G$ are preserved under elementary transformations of ${G}$. More precisely,
    \begin{enumerate}
        \item  replacement of a single edge with weight $\nu_1+\nu_2$ by parallel edges with weights $\nu_1, \nu_2$ corresponds to adding a point dividing the corresponding edge of the t-embedding into two segments with proportion $\left[\nu_1: \nu_2\right]$. Merging double edges corresponds to removing a degree 2 vertex of the t-embedding, which due to the angle condition has to lie on a line with its two adjacent vertices.
        \item Contracting a degree 2 vertex corresponds to removing a degree two face of a t-embedding which can be seen as a diagonal of a face of the t-embedding. Splitting a vertex of degree $d_1+d_2$ to two vertices of degrees $d_1+1$ and $d_2+1$ and adding a degree two vertex between them corresponds to adding a diagonal to the corresponding face of the t-embedding with respect to the structure of the split graph.
        \item Performing a spider move corresponds to a central move of points of t-embedding. The transformation is given by,
        $$
        \frac{\left(u_1-u\right)\left(u_3-u\right)}{\left(u_2-u\right)\left(u_4-u\right)}=\frac{\left(u_1-\tilde{u}\right)\left(u_3-\tilde{u}\right)}{\left(u_2-\tilde{u}\right)\left(u_4-\tilde{u}\right)}
        $$
        where $u$ and $\tilde u$ are points on Figure \ref{fig:embedding_moves}.
    \end{enumerate}
\end{prop}
An example of these transformations is shown in Figure  \ref{fig:embedding_moves}. Let us focus briefly on the central move (rightmost image in Figure \ref{fig:embedding_moves}). Note that a face $(j,k)$ in the generalized tower graphs is shufflable only if $j+k+n$ is odd, where $n$ is the height of column $j+k$. Define the set $\Lambda^+:=\{(j,k,n)|j+k+n \text{ is odd}\}$. Suppose $(j,k,n-1)\in\Lambda^+$, then the central move says,
\[
\mathcal{T}_{n+1}(j,k) + \mathcal{T}_{n-1}(j,k) = \frac{1}{2}\left( \mathcal{T}_{n}(j+1,k) + \mathcal{T}_{n}(j-1,k) + \mathcal{T}_{n}(j,k+1) + \mathcal{T}_{n}(j,k-1) \right)
\]
where $\mathcal{T}_{n}(j,k)$ is the value of the t-embedding for face $(j,k)$ for the graph of rank $n$, with similar formulas holding for the other elementary moves. This gives us a recursive formula for computing the t-embedding provided we know the appropriate boundary conditions for the recurrence.

\begin{figure}[htbp]
    \centering
        \includegraphics[width=\textwidth]{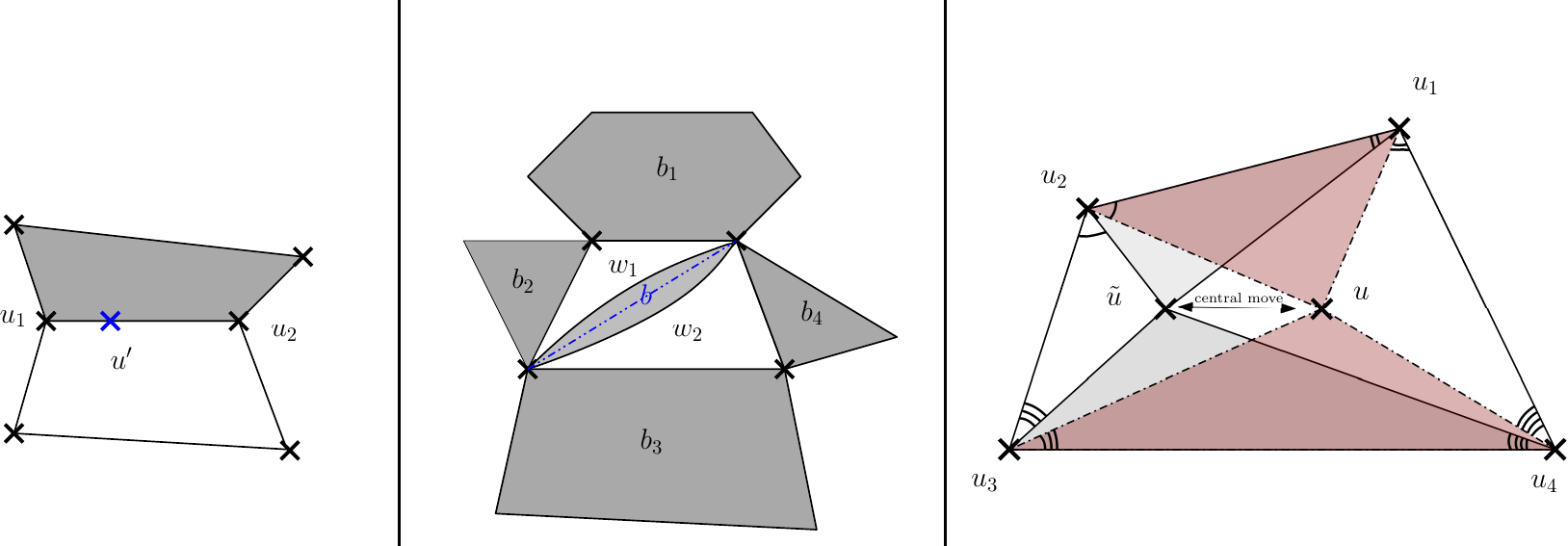}
    \caption{Transformation in the t-embedding corresponding to the transformation in $G$}
    \label{fig:embedding_moves}
\end{figure}

The general form of this recurrence is given in \cite{chelkak2021fluctuations}. We reproduce it here. Let $F: \Lambda^+\to \C$. Fix complex parameters $(b_0,b_{SE},b_{NE},b_{NW},b_{SW})$. The following conditions determine $F$ uniquely\footnote{In \cite{chelkak2021fluctuations}, solutions $F$ to the recurrence are said to solve the \emph{discrete wave equation in the cone}.}:
\begin{enumerate}
    \item $F(j,k,n)=0$ if $|j|+|k|\ge n$;
    \item Boundary condition at tip of the cone: $F(0,0,1)=b_0$;
    \item Boundary conditions at the edges: for each $n\ge 1$,
    \[
    \begin{aligned}
        F(-n,0,n+1) = \frac{1}{2}\left(F(-n+1,0,n)+b_{NW}\right), & \quad F(0,n,n+1) = \frac{1}{2}\left(F(0,n-1,n)+b_{NE}\right), \\
        F(0,-n,n+1) = \frac{1}{2}\left(F(0,-n+1,n)+b_{SW}\right), & \quad F(n,0,n+1) = \frac{1}{2}\left(F(n+1,0,n)+b_{SW}\right);
    \end{aligned}
    \]
    \item In the bulk: $|j|+|k|\le n$ 
    \[
    \begin{aligned}
    &F(j,k,n+1) + F(j,k,n-1) =\\
    &\frac{1}{2}\left(F(j+1,k,n-1) + F(j+1,k,n-1) + F(j,k+1,n-1) + F(j,k-1,n-1)\right).
    \end{aligned}
    \]
\end{enumerate}

By comparing the general recursion relation to the recursion relation one gets from the shuffling algorithm, one can prove the following proposition. Let $\cT_n(j,k) \in \C$ be the point in the perfect t-embedding corresponds to the face $(j,k)$ in $AD'_{n+1}$.

\begin{prop}[Proposition 2.7 of \cite{chelkak2021fluctuations}]
    \label{prop: AD_t_embedding}
    \hfill
    \begin{enumerate}
    \item The t-embedding of the reduced Aztec diamond $\mathcal{T}_n(j,k)$ solves the above recursion relation with $(b_0,b_{SE},b_{NE},b_{NW},b_{SW})=(0,1,i,-1,-i)$. 
    \item The origami map of the reduced Aztec diamond $\mathcal{O}_n(j,k)$ solves the above recursion relation with $(b_0,b_{SE},b_{NE},b_{NW},b_{SW})=(0,1,i,1,i)$.
    \end{enumerate}
\end{prop}

\begin{remark}\label{rmk:tembshuf}
    We would like to stress that in the recursive formula, each update is given by a local relation depending only on a single face and its neighbors (possibly including a boundary value) and corresponds to a local transformation of the corresponding graph coming from the shuffling algorithm. As the order of which one shuffles the faces in the shuffling algorithm does not effect the resulting graph, the same is true for the order in which one does updates for the t-embedding and origami map. 
\end{remark}

Similarly, one can show that the origami map also preserves under the elementary transformation in Figure \ref{fig:elt_moves}.

\begin{prop}[\cite{circle_pattern}]
    The local transformations in Figure \ref{fig:elt_moves} preserves the origami maps. Namely, the image of $u$ and $\tilde u$ in Figure \ref{fig:embedding_moves} under the origami map ${\mathcal O}$ satisfies, 
    $$
\frac{\left(\mathcal{O}\left(u_1\right)-\mathcal{O}(u)\right)\left(\mathcal{O}\left(u_3\right)-\mathcal{O}(u)\right)}{\left(\mathcal{O}\left(u_2\right)-\mathcal{O}(u)\right)\left(\mathcal{O}\left(u_4\right)-\mathcal{O}(u)\right)}=\frac{\left(\mathcal{O}\left(u_1\right)-\mathcal{O}(\tilde{u})\right)\left(\mathcal{O}\left(u_3\right)-\mathcal{O}(\tilde{u})\right)}{\left(\mathcal{O}\left(u_2\right)-\mathcal{O}(\tilde{u})\right)\left(\mathcal{O}\left(u_4\right)-\mathcal{O}(\tilde{u})\right)}
$$
\end{prop}
\begin{proof}
    See Corollary 2.15 of \cite{tower_AD_perfect} for the detailed proof and visualization. 
\end{proof}

\subsection{Convergence to maximal surface}
Recall that to one of the assumptions in Thm. \ref{thm:Main_thm_CLR21} is the convergence of $(\mathcal{T}_n,\mathcal{O}'_n)$ to a maximal surface in $\R^{2,1}$. In the case of the Aztec diamond, the following proposition is conjectured in \cite{chelkak2021fluctuations} and proven in \cite{tower_AD_perfect} using the same asymptotic expansion of the edge probabilities we will use in Cor. \ref{cor: embedding_ori_prob}.
\begin{prop}
    \label{prop: convergence_max_surface}
     Let $\mathcal{T}_n(j, k)$ denote the image of the vertex $(j, k) \in\left(AD_{n+1}^{\prime}\right)^*$ under the perfect t-embedding of $AD_{n+1}^{\prime}$ and $\mathcal{O}_n^{\prime}(j, k)$ denote the image of this vertex under the origami map. Then,
$$
\begin{aligned}
\mathcal{T}_n(j, k) & =z(j / n, k / n)+o(1) \\
\mathcal{O}_n^{\prime}(j, k) & =\vartheta(j / n, k / n)+o(1)
\end{aligned}
$$
where $z(x, y)$ and $\vartheta(x, y)$ are smooth functions on the rescaled Aztec domain $\mathcal A:=\{(x,y): |x|+|y|<1\}$, defined by
$$
\begin{aligned}
z(x, y) & :=\Psi_E(x, y)+i \Psi_N(x, y)-\Psi_W(x, y)-i \Psi_S(x, y) \\
\vartheta(x, y) & :=\frac{1}{\sqrt{2}}\left(\Psi_E(x, y)-\Psi_N(x, y)+\Psi_W(x, y)-\Psi_S(x, y)\right)
\end{aligned}
$$
where the function $\Psi_{(.)}$ is given by
$$
\begin{aligned}
\Psi_E(x, y)&:=\int_0^1 \Psi_0(x-s, y, 1-s) d s \\
\Psi_N(x, y)&:=\int_0^1 \Psi_0(x, y-s, 1-s) d s \\
\Psi_W(x, y)&:=\int_0^1 \Psi_0(x+s, y, 1-s) d s \\
\Psi_S(x, y)&:=\int_0^1 \Psi_0(x, y+s, 1-s) d s 
\end{aligned}
$$
with
$$
\Psi_0(x, y, t):= \begin{cases}\frac{1}{\pi}\left(t^2-2 x^2-2 y^2\right)^{-1 / 2} & \text { if } x^2+y^2<\frac{1}{2} t^2 \\ 0 & \text { otherwise }\end{cases}
$$
and the o(1) error is uniform for $(j / n, k / n)$ in compact subsets of $\mathcal{A}$.
\end{prop}
\begin{cor}
    \label{cor: AD_ori_convergence}
 On compact subsets of $\lozenge$
$$
\mathcal{O}_n^{\prime}(z) \rightarrow \vartheta(z)
$$
uniformly as $n\to \infty$, where the graph of $\vartheta$ is the maximal surface $S_{\lozenge} \subset \mathbb{R}^{2,1}$ with boundary contour $C_{\lozenge}$.
\end{cor}

The proof for Prop. \ref{prop: convergence_max_surface} and Cor. \ref{cor: AD_ori_convergence} are written in detail in Section 5 of \cite{tower_AD_perfect} so we will not repeat here.

\subsection{Perfect t-embeddings of generalized tower graphs}

As the generalized tower graphs are also constructed by shuffling, Rmk. \ref{rmk:tembshuf} and Prop. \ref{prop:boundary4} ensure that perfect t-embeddings also exist, although in general they do not satisfy such a simple recursive formula as that of the Aztec diamond. However, we may relate the t-embeddings of the generalized tower graphs to those of the Aztec diamond through the following lemma. 

\begin{lemma}
\label{lem: gen_tower_AD_coord}
    Let $\widetilde \cT_n(j,k)$ be the t-embedding of the reduced generalized tower graph of rank $n+1$ approximating $f$. Then 
\[
\begin{aligned}
     \widetilde \cT_n(j,k) &= \cT_{\rr}(j,k) \\
     \widetilde \cO_n(j,k) &= \cO_{\rr}(j,k) \\
\end{aligned}
\]
with $r(n;j,k)+j+k$ odd. Moreover, 
\[
\rr=(n+1) \bar f\left(\frac{j+k}{n+1}\right)+\epsilon(j,k)
\]
with $|\epsilon(j,k)|\le 2 $. 
\end{lemma}
\begin{proof}
    Note that when constructing the t-embedding (or origami map) via the shuffling algorithm, the image of a face only changes when there is local transformation of $(G^f_n)'$ at that face. In particular, when a column $c$ has been shuffled until it reaches its desired height $H_n^f(c)$, the embedding of the faces in that column no longer change. In light of Rmk. \ref{rmk:tembshuf}, the embedding of faces in column $c$ of the  $(G^f_n)'$ will be the same as those of a reduced Aztec diamond whose the height at column $c$ is equal to $H_n^f(c)$. In fact, we have two choices here since, for $r+j+k$ odd, both the Aztec diamond of rank $r$ and rank $r-1$ have the same height at column $j+k$.
        
    Let $\rr$ be the choice of for which $r(n;j,k)+j+k$ is odd. We have
    \[
    \widetilde \cT_n(j,k) = \cT_{r(n;j,k)}(j,k).
    \]
    From Eqn. \eqref{eq:error}, we have
    \[
     0 \le (n+1) \bar f\left(\frac{j+k}{n+1}\right) - H_n^f(j+k) \le 1.
    \] 
    Letting $\epsilon(j,k) = r(n;j,k) - (n+1) \bar f\left(\frac{j+k}{n+1}\right)$ and noting that $r(n;j,k)$ differs from $H_n^f(j+k)$ by at most one, we
    have
    \[
    \rr=(n+1) \bar f\left(\frac{j+k}{n+1}\right)+\epsilon(j,k)
    \]
    with $|\epsilon(j,k)|\le2$, as desired.
\end{proof}
Note that the points in the embedding are always from faces in the interior columns of the graph, so we may drop the bar on the $f$ without issue. We will use only $f$ instead of $\bar f$ whenever referring to $\rr=(n+1) \bar f\left(\frac{j+k}{n+1}\right)+\epsilon(j,k)$ from now on.

An example of the perfect t-embedding for the Aztec diamond and tower graph of rank 3 are given in Figure \ref{fig:perfecttembEx}. The perfect t-embeddings of the reduced generalize tower graphs of rank 15 approximating $f(x) =\frac{1}{3}x^2+\frac{2}{3}$ and $f(x)=\frac{1}{3\pi}\sin(2\pi x)+1$ are drawn in Figure \ref{fig:perfecttembEx2}.

\begin{figure}
    \centering
    \[
    \includegraphics[width=0.2\textwidth]{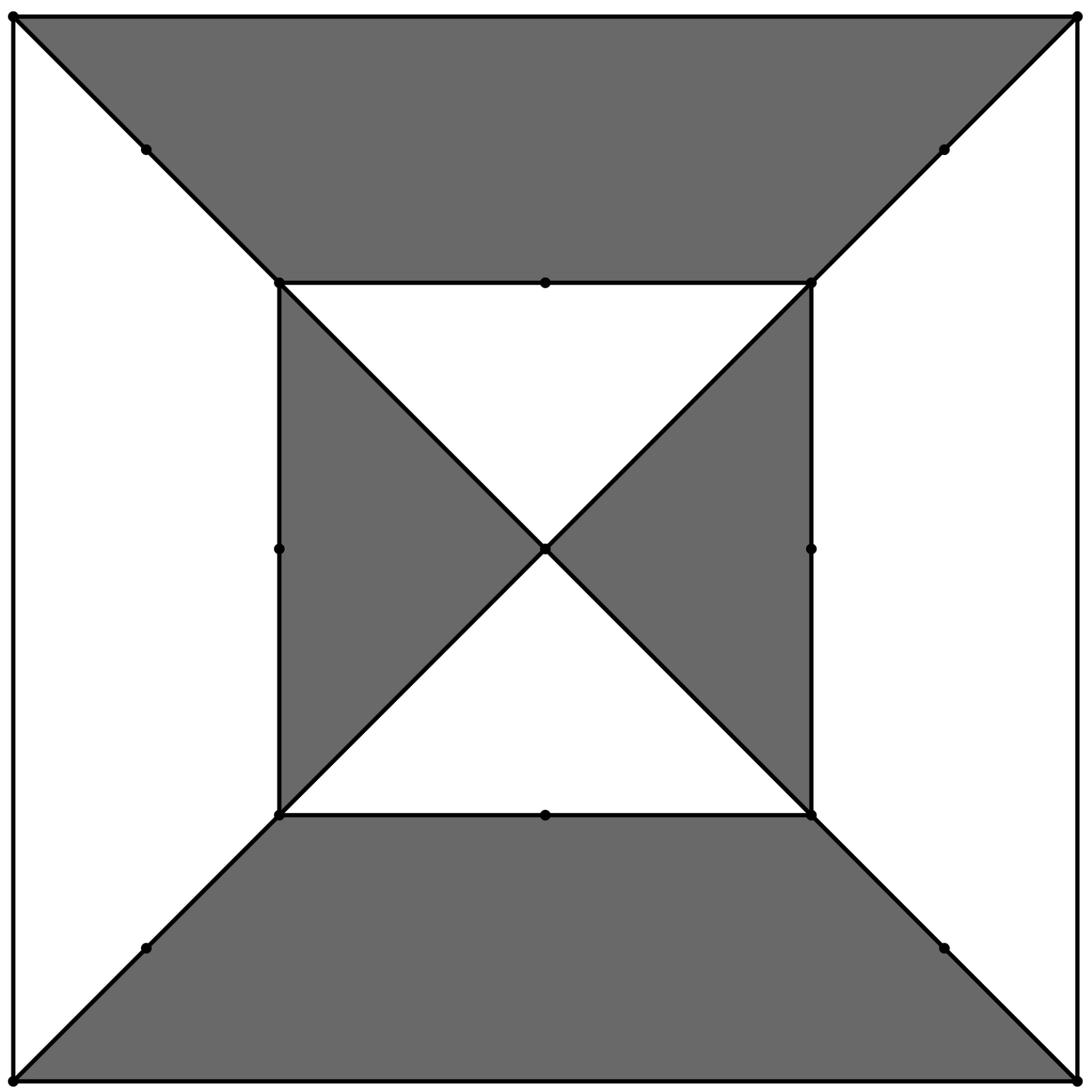} 
   \qquad \qquad
    \includegraphics[width=0.2\textwidth]{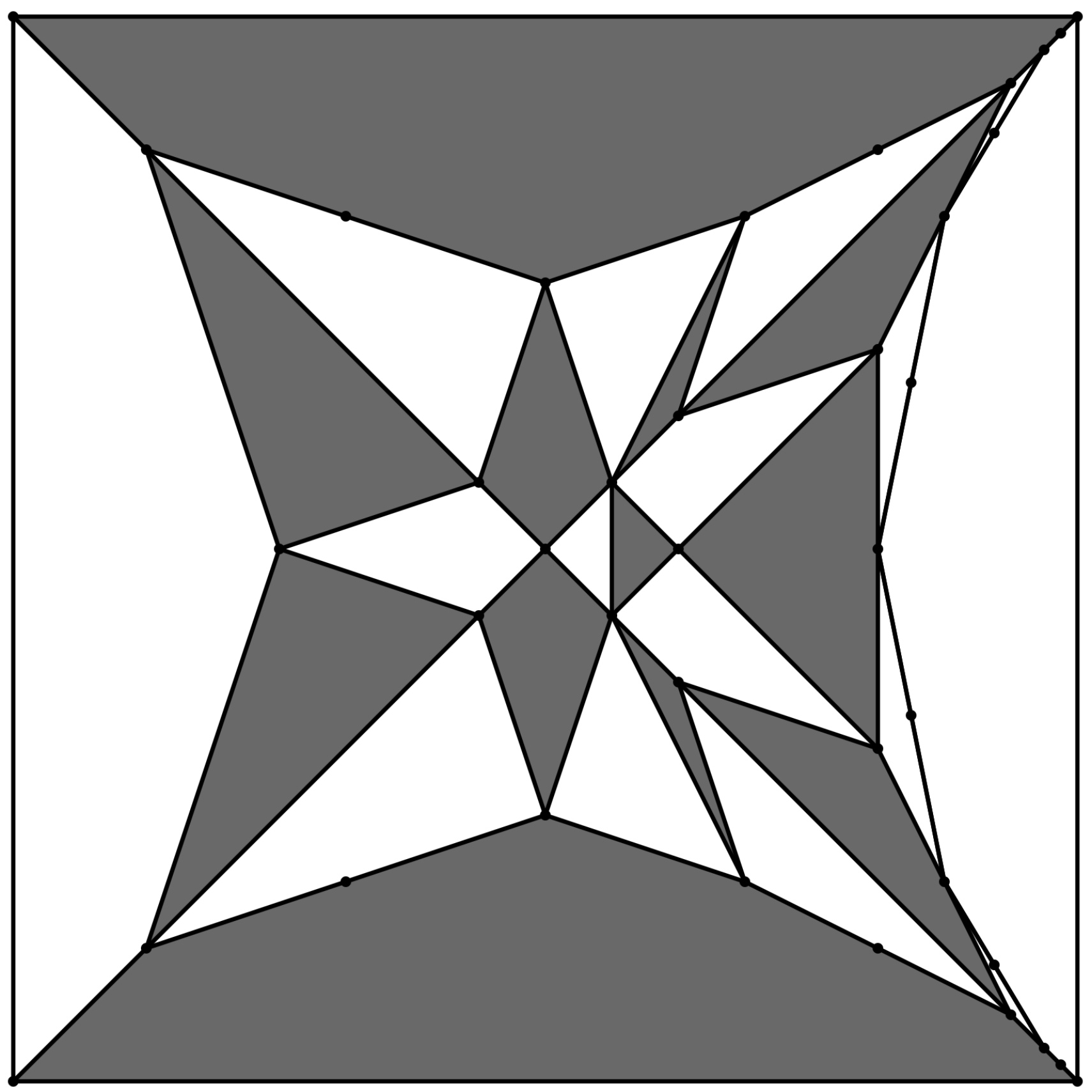} 
    \]
    \caption{Left: The perfect t-embedding of the reduced Aztec diamond of rank 3. Right:  The perfect t-embedding of the reduced tower graph of rank 3.}
    \label{fig:perfecttembEx}
\end{figure}

\begin{figure}
    \centering
    \[
    \includegraphics[width=0.2\textwidth]{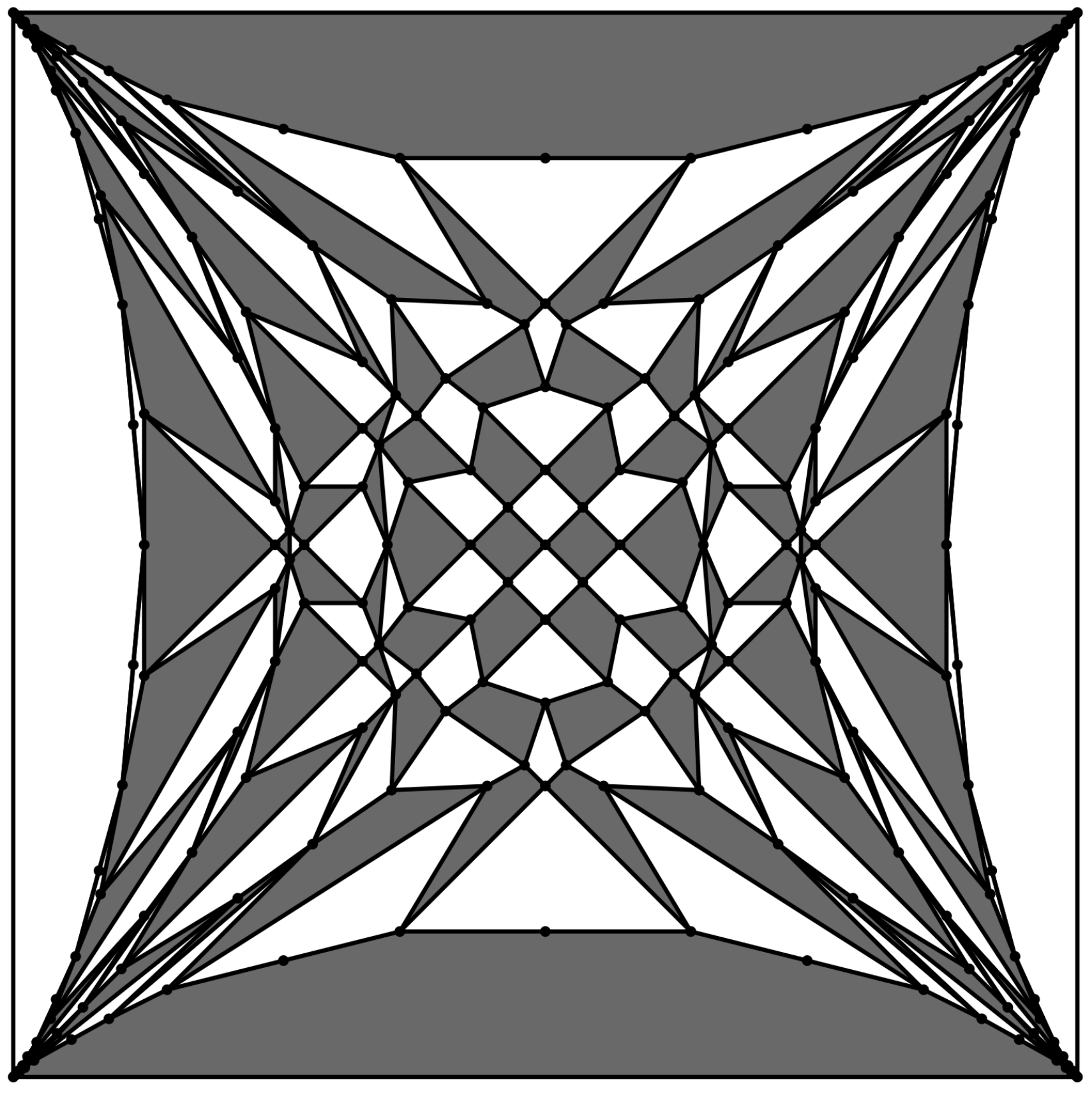} 
   \qquad \qquad
    \includegraphics[width=0.2\textwidth]{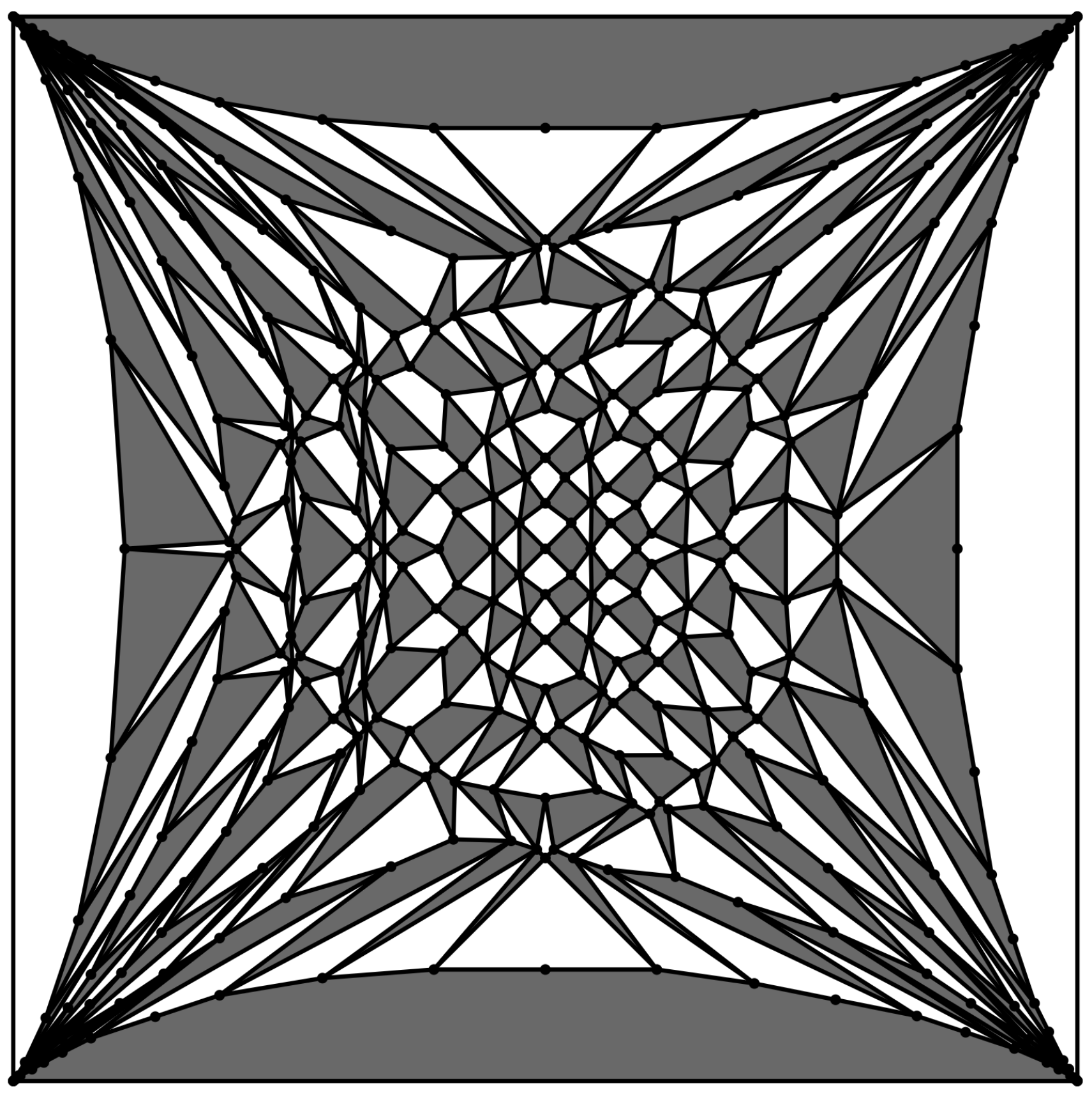} 
    \]
    \caption{Left: The perfect t-embedding of the reduced generalized tower graph of rank 15 approximating $f(x)=\frac{1}{3}x^2+\frac{2}{3}$. Right:  The perfect t-embedding of the reduced generalized tower graph of rank 15 approximating $f(x)=\frac{1}{3\pi}\sin(2\pi x) +1$.}
    \label{fig:perfecttembEx2}
\end{figure}

Recall that in \cite{chelkak2021fluctuations}, it is shown that the pair of t-embedding and origami map $(\cT_n, \cO'_n)$ of the Aztec diamond converges to the maximal surface $\{(z(x, y), \vartheta(x, y)):(x, y) \in \lozenge\}$ in the Minkowski space $\mathbb{R}^{2,1}$. Prop. \ref{prop: convergence_max_surface}, together with Lem. \ref{lem: gen_tower_AD_coord}, show that $(\mathcal{\widetilde T}_n, \mathcal{\widetilde O}'_n)$ converges uniformly to the same surface. 

\begin{cor}
    \label{cor: gen_tower_limit_surface}
    We have that
    \[
    \left(\mathcal{T}_n(\lfloor n x \rfloor, \lfloor n y \rfloor), \mathcal{O}_n^{\prime}(\lfloor n x \rfloor, \lfloor n y \rfloor)\right) \to \left(z\left(\frac{x}{f(x+y)},\frac{y}{f(x+y)}\right),\theta\left(\frac{x}{f(x+y)},\frac{y}{f(x+y)}\right)\right)
    \]
    converges uniformly on compact subset of the limiting generalized tower graph
    \[
    \begin{aligned}
        &-1 < x+y < x_* \\
        &-f(x+y) < x-y < f(x+y)
    \end{aligned}
    \]
    where $f(x_*)=x_*$. Moreover, as $(x,y)$ range over the rescaled generalized tower graph, the points $\left(\frac{x}{f(x+y)},\frac{y}{f(x+y)}\right)$ range over the limiting Aztec diamond $|x|+|y|<1$.
\end{cor}
\begin{proof}
    From Lem. \ref{lem: gen_tower_AD_coord}, we may write
    \[
    \left(\widetilde{\mathcal{T}}_n(j, k), \widetilde{\mathcal{O}}_n^{\prime}(j, k)\right) = \left(\mathcal{T}_{r(n;j,k)}(j,k),\mathcal{O}_{r(n;j,k)}(j,k)\right).
    \]
    Using Prop. \ref{prop: convergence_max_surface} this becomes
    \[
    \begin{aligned}
     \left(\widetilde{\mathcal{T}}_n(j, k), \widetilde{\mathcal{O}}_n^{\prime}(j, k)\right) &\;= \left(z\left(\frac{j}{r(n;j,k)},\frac{k}{r(n;j,k)}\right),\vartheta\left(\frac{j}{r(n;j,k)},\frac{k}{r(n;j,k)}\right)\right) + o(1) \\
     &\;\to \left(z\left(\frac{x}{f(x+y)},\frac{y}{f(x+y)}\right),\vartheta\left(\frac{x}{f(x+y)},\frac{y}{f(x+y)}\right)\right)
     \end{aligned}
    \]
    where $\left(\frac{j}{n+1},\frac{k}{n+1}\right)\to (x,y)$ and the convergence is uniform.

    Note that as $-1 < x+y < x_*$, we have
    \[
    -1 < \frac{x+y}{f(x+y)}  < 1  
    \]
    since $f(-1)=-1$ and $f(x_*)=x_*$. We also see that as $-f(x+y) < x-y < f(x+y)$ and it immediately follows that
    \[
    -1 < \frac{x-y}{f(x+y)}  < 1.
    \]
\end{proof}

By connecting the t-embeddings of the generalized tower graphs to those of the Aztec diamond, we may leverage the known exact formulas for the Aztec diamond to study the limiting behavior of the generalized tower graphs. Our main result is the following theorem whose proof of the rigidity assumption we delay to Section \ref{sect: rigidity}.

\begin{thm}
    \label{thm: main_thm_generalized_tower}
    Let $\widetilde{ \mathcal{T}}_n$ be the sequence of perfect t-embeddings of the reduced uniformly weighted generalized tower graph approximating $f$, as defined in Section \ref{sect: gen_shuffling}, with corresponding origami maps $\widetilde{\mathcal{O}}_n^{\prime}$ given in Lem. \ref{lem: gen_tower_AD_coord} . The assumptions \ref{hyp: domain}, \ref{hyp:convergence}, and \ref{hyp:technical} of Thm. \ref{thm:Main_thm_CLR21} hold for the sequence $\widetilde{\mathcal{T}}_n$, with the choice $\displaystyle\delta_n=\frac{\log n}{n}$.
\end{thm}

\begin{proof}
    Assumption \ref{hyp: domain} holds immediately since for every $n$, the domain covered by $\widetilde{\cT}_n$ is just a triangulation of $\lozenge:=\{z=x+\mathrm i y; |x|+|y|<1\}.$ Assumption \ref{hyp:convergence} is the content of  Cor. \ref{cor: gen_tower_limit_surface}. Assumption \ref{hyp:technical} is the content of Section \ref{sect: rigidity} where we show that the rigidity assumption holds for the generalized tower graphs, which in turn implies the LIP and EXP-FAT assumptions. 
\end{proof}

As a result, we have that the gradients of
$n$-point correlation functions of the height functions on $(G^f_n)^*$ converges to that of the Gaussian free field as $n \to \infty$ by Thm. \ref{thm:Main_thm_CLR21}.

\section{Rigidity}\label{sect: rigidity}

In this section, we complete the verification of the rigidity assumption for our family of generalized tower graphs. Our method of proof follows the same form as the case of Aztec diamond in \cite{tower_AD_perfect}, utilizing the local edge probability expression in \cite{CEP, MR3325273} and steepest descent/saddle point asymptotic analysis, see \cite{deBruijn81} for exposition of the saddle-point/steepest descent method. The main technical difference is that we must consider edges between adjacent hexagonal faces and the error term.

To begin, we summarize the relationship between the value of $\cT_{n}(j,k)$ of the reduced Aztec diamond $AD'_{n+1}$ and the probability of the event where a specific face $(j,k)$ in $AD_n$ contains a dimer edge. Specifically, denote $p_E(j, k, n), p_N(j, k, n), p_W(j, k, n), p_S(j, k, n)$ to be the probabilities that the face $(j,k)$ contains a covered edge on the east, north, west and south edge, respectively, see Fig. \ref{fig:edge_prob}. Note that even though $p_E, p_N, p_W, p_S$ are functions on the faces of $AD_n$ and $\cT_n(j,k)$ is the function on the face of $AD'_{n+1}$, all inner faces of $AD'_{n+1}$ are in the same correspondence with $AD_n$. Thus, we can write $\cT_n(j,k)$ as function of $p_E(j,k,n), p_N(j,k,n), p_W(j,k,n), p_S(j,k,n)$, 

\begin{figure}
    \centering
    \includegraphics[width=0.3\linewidth]{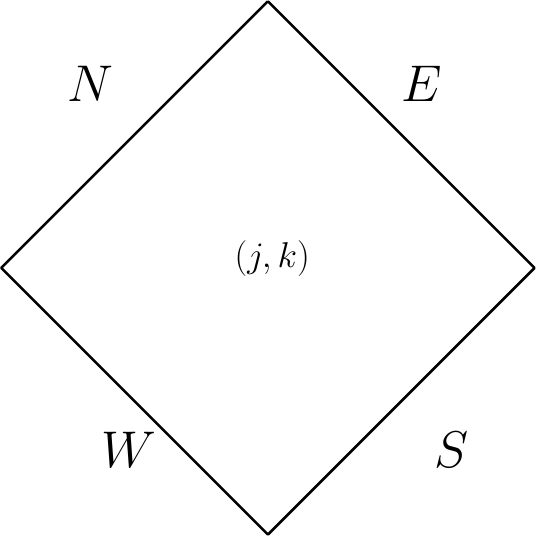}
    \caption{The convention of edges around a square face indexed $(j,k)$ }
    \label{fig:edge_prob}
\end{figure}

\begin{cor}[\cite{chelkak2021fluctuations, tower_AD_perfect}]
    \label{cor: embedding_ori_prob}
Let $p_E, p_S, p_N$ and $p_W$ be edge probabilities as defined above, for $j+k+n$ odd we have
\begin{equation}
    \mathcal{T}_n(j, k)=p_E(j, k, n)+\mathrm{i} p_N(j, k, n)-p_W(j, k, n)-\mathrm{i} p_S(j, k, n) 
\end{equation}
and
\begin{equation}
    \mathcal{O}_n(j, k)=p_E(j, k, n)+\mathrm{i} p_N(j, k, n)+p_W(j, k, n)+\mathrm{i} p_S(j, k, n). 
\end{equation}
\end{cor}

We use the same convention as in \cite{tower_AD_perfect} for the Aztec diamond with Kasteleyn weights given by 
\begin{equation}
K_{\left(j_1+\frac{1}{2}, k_1+\frac{1}{2}\right),\left(j_2+\frac{1}{2}, k_2-\frac{1}{2}\right)}= \begin{cases}(-1)^{k_1+n+1}, & \left(j_2, k_2\right)=\left(j_1, k_1\right) \\ (-1)^{k_1+n+1} \mathrm{i}, & \left(j_2, k_2\right)=\left(j_1-1, k_1+1\right) \\ (-1)^{k_1+n}, & \left(j_2, k_2\right)=\left(j_1, k_1+2\right) \\ (-1)^{k_1+n}\mathrm{i}, & \left(j_2, k_2\right)=\left(j_1+1, k_1+1\right) \\ 0, & \text { otherwise. }\end{cases}
\end{equation}
With this one can write the edge probabilities for the Aztec diamond as 
\begin{equation}\begin{aligned}
    p_E(j,k,n)&=K_{(j+\frac{1}{2},k+\frac{1}{2}),(j+\frac{1}{2},k-\frac{1}{2})}K^{-1}_{(j+\frac{1}{2},k-\frac{1}{2}),(j+\frac{1}{2},k+\frac{1}{2})}=(-1)^{k+n+1}K^{-1}_{(j+\frac{1}{2},k-\frac{1}{2}),(j+\frac{1}{2},k+\frac{1}{2})}\\
    p_N(j,k,n)&=K_{(j+\frac{1}{2},k+\frac{1}{2}),(j-\frac{1}{2},k+\frac{1}{2})}K^{-1}_{(j-\frac{1}{2},k+\frac{1}{2}),(j+\frac{1}{2},k+\frac{1}{2})}=(-1)^{k+n+1} \mathrm{i}K^{-1}_{(j-\frac{1}{2},k+\frac{1}{2}),(j+\frac{1}{2},k+\frac{1}{2})}\\
    p_W(j,k,n)&=K_{(j-\frac{1}{2},k-\frac{1}{2}),(j-\frac{1}{2},k+\frac{1}{2})}K^{-1}_{(j-\frac{1}{2},k+\frac{1}{2}),(j-\frac{1}{2},k-\frac{1}{2})}=(-1)^{k+n-1}K^{-1}_{(j-\frac{1}{2},k+\frac{1}{2}),(j-\frac{1}{2},k-\frac{1}{2})}\\
    p_S(j,k,n)&=K_{(j-\frac{1}{2},k-\frac{1}{2}),(j+\frac{1}{2},k-\frac{1}{2})}K^{-1}_{(j+\frac{1}{2},k-\frac{1}{2}),(j-\frac{1}{2},k-\frac{1}{2})}=(-1)^{k+n-1}\mathrm{i}K^{-1}_{(j+\frac{1}{2},k-\frac{1}{2}),(j-\frac{1}{2},k-\frac{1}{2})}
\end{aligned}  
\end{equation}
for those faces $(j,k,n)$ such that $j+k+n$ and $k+n$ are odd. Combining the above expressions with the form of the inverse Kasteleyn matrix given by Chhita-Johansson-Young in  \cite{MR3325273},
$$
\begin{aligned}
& K_{\left(j_1+\frac{1}{2}, k_1-\frac{1}{2}\right),\left(j_2+\frac{1}{2}, k_2+\frac{1}{2}\right)}^{-1} \\
& \qquad= \begin{cases}f_1\left(\left(j_1, k_1\right),\left(j_2, k_2\right)\right), & j_1+k_1<j_2+k_2+2 \\
f_1\left(\left(j_1, k_1\right),\left(j_2, k_2\right)\right)-f_2\left(\left(j_1, k_1\right),\left(j_2, k_2\right)\right), & \text { otherwise }\end{cases}
\end{aligned}
$$
where
$$
\begin{aligned}
f_1\left(\left(j_1, k_1\right),\left(j_2, k_2\right)\right)= & \frac{(-1)^{\frac{1}{2}\left(k_1+k_2+2 n\right)}}{(2 \pi \mathrm{i})^2} \int_{\mathcal{E}_2} \int_{\mathcal{E}_1} \frac{w^{\frac{1}{2}\left(j_2+k_2+n+1\right)}}{z^{\frac{1}{2}\left(j_1+k_1+n+1\right)}} \\
& \times \frac{(1+z)^{\frac{1}{2}\left(k_1-j_1+n-1\right)}(z-1)^{\frac{1}{2}\left(j_1-k_1+n+1\right)}}{(1+w)^{\frac{1}{2}\left(k_2-j_2+n+1\right)}(w-1)^{\frac{1}{2}\left(j_2-k_2+n+1\right)}} \frac{\mathrm{d} z \mathrm{~d} w}{w-z}, \\
f_2\left(\left(j_1, k_1\right),\left(j_2, k_2\right)\right)= & \frac{(-1)^{\frac{1}{2}\left(k_1+k_2+2 n\right)}}{2 \pi \mathrm{i}} \int_{\mathcal{E}_1} \frac{(1+z)^{\frac{1}{2}\left(j_2+k_2-j_1-k_1\right)} z^{\frac{1}{2}\left(k_2-j_2-k_1+j_1\right)}}{(z+2)^{\frac{1}{2}\left(k_2-j_2-k_1+j_1+2\right)}} \mathrm{d} z
\end{aligned}
$$
and $\mathcal{E}_1$ is a small circle around 0 and $\mathcal{E}_2$ is a small circle around 1, both with radius strictly less than $1 / 2$ and oriented in positive direction. With this notation, we have that the edge probabilities are given by
\begin{equation}
\resizebox{0.9\textwidth}{!}{$
\begin{aligned}
    p_E(j,k,n)&=(-1)^{k+n+1}K^{-1}_{(j+\frac{1}{2},k-\frac{1}{2}),(j+\frac{1}{2},k+\frac{1}{2})}=(-1)^{k+n+1}f_1\left((j, k),(j, k)\right),\\
    p_N(j,k,n)&=(-1)^{k+n+1}\mathrm iK^{-1}_{(j-\frac{1}{2},k+\frac{1}{2}),(j+\frac{1}{2},k+\frac{1}{2})}=(-1)^{k+n+1}\mathrm if_1\left((j-1,k+1), (j,k)\right),\\
    p_W(j,k,n)&=(-1)^{k+n-1}K^{-1}_{(j-\frac{1}{2},k+\frac{1}{2}),(j-\frac{1}{2},k-\frac{1}{2})}\\
    &=(-1)^{k+n-1}\left(f_1((j-1,k+1),(j-1,k-1))-f_2((j-1,k+1),(j-1,k-1))\right),\\
    p_S(j,k,n)&=(-1)^{k+n-1}iK^{-1}_{(j+\frac{1}{2},k-\frac{1}{2}),(j-\frac{1}{2},k-\frac{1}{2})}\\
    &=(-1)^{k+n-1}\mathrm i \left(f_1((j,k),(j-1,k-1))-f_2((j,k),(j-1,k-1))\right).
\end{aligned} 
$}
\end{equation}
In the case of Aztec diamond, the following contour integral formula for $\mathcal{T}_n$ is shown in \cite{tower_AD_perfect}. 
\begin{lemma}[Lemma 6.3, \cite{tower_AD_perfect} ]
\label{lem: t_embedding_coord_AD}
    For $(x, y) \in \mathcal A$, define the action function as 
    \[
    S(z,x,y):=\frac{1}{2}(1+x+y) \log z-\frac{1}{2}(1-x+y) \log (z+1)-\frac{1}{2}(1+x-y) \log (z-1)
    \]
    where
$$
G_{\mathcal{T}}(z, w):=\frac{(1-\mathrm{i})(z-\mathrm{i})(w-\mathrm{i})}{(z-1)(z+1) w}.
$$
Let $(j, k) \in\left(AD_{n+1}^{\prime}\right)^*$ with $j+k+n$ odd, and set $(x, y)=\frac{1}{n+1}(j, k)$. Then the perfect t-embedding
$$
\mathcal{T}_n(j, k)=\frac{1}{(2 \pi \mathrm{i})^2} \int_{\mathcal{E}_2} \int_{\mathcal{E}_1} \exp{\left\{(n+1)(S(w)-S(z))\right\}} G_{\mathcal{T}}(z, w) \frac{\mathrm d z \mathrm d w}{z-w}-1
$$
where $\mathcal{E}_1$ and $\mathcal{E}_2$ are positively oriented circles around 0 and 1, respectively, with radius strictly less than $1 / 2$, and we denote $S(z):=S(z ; x, y)$.
\end{lemma}
Applying Lem. \ref{lem: t_embedding_coord_AD} together with Lem. \ref{lem: gen_tower_AD_coord}, the perfect t-embeddings of the generalized tower graphs also admit a double contour integral formula via the edge probability of the Aztec diamond. 

\begin{lemma}
    \label{lem: generalized_double_contour_int}
\begin{equation}
\resizebox{0.9\textwidth}{!}{$
\begin{aligned}
& {p_E}(j, k, \rr)=\frac{-1}{(2 \pi \mathrm{i})^2}\int_{\mathcal{E}_2} \int_{\mathcal{E}_1} \mathrm{e}^{(n+1)(\tS(w)-\tS(z))} \frac{1}{1+z} \left( \frac{w(1+z)(z-1)}{z(1+w)(w-1)}\right)^{\frac{\epsijk+1}{2}} \frac{\mathrm d z \mathrm d w}{z-w}\\
&{p_N}(j, k, \rr) = \frac{-1}{(2\pi\mathrm{i})^2} \int_{\mathcal E_2}\int_{\mathcal{E}_1} \mathrm{e}^{(n+1)(\tS(w)-\tS(z))} \frac{1}{z-1} \left( \frac{w(1+z)(z-1)}{z(1+w)(w-1)}\right)^{\frac{\epsijk+1}{2}} \frac{\mathrm d z \mathrm d w}{z-w}\\
& {p_W}(j, k, \rr) =  \frac{1} {(2\pi\mathrm{i})^2}\int_{\mathcal E_2}\int_{\mathcal{E}_1} \mathrm{e}^{(n+1)(\tS(w)-\tS(z))} \frac{1}{w(z-1)} \left( \frac{w(1+z)(z-1)}{z(1+w)(w-1)}\right)^{\frac{\epsijk+1}{2}} \frac{\mathrm d z \mathrm d w}{z-w}+1 \\
& {p_S}(j, k, \rr) = \frac{1} {(2\pi\mathrm{i})^2}\int_{\mathcal E_2}\int_{\mathcal{E}_1} \mathrm{e}^{(n+1)(\tS(w)-\tS(z))} \frac{1}{w(1+z)} \left( \frac{w(1+z)(z-1)}{z(1+w)(w-1)}\right)^{\frac{\epsijk+1}{2}} \frac{\mathrm d z \mathrm d w}{z-w}
\end{aligned}
$}
\end{equation}
with the modified action
\begin{equation}
\resizebox{0.9\textwidth}{!}{$
    \begin{aligned}
        &\quad \tS(z):=\tS(z; x,y) \\
        &= \frac{1}{2}\left(f(x+y)+x+y\right)\log(z) - \frac{1}{2}\left(f(x+y)-x+y\right)\log(z+1) - \frac{1}{2}\left(f(x+y)+x-y\right)\log(z-1) 
    \end{aligned}
$}
\end{equation}
where $(x,y) = \frac{1}{n+1}(j,k)$. 
\end{lemma}
\begin{proof}
    For the case $ p_E(j,k,\rr)$, we have
\[
    \begin{aligned}
   &\qquad{p_E}(j, k, \rr)=(-1)^{k+\rr+1} f_1((j, k),(j, k)) \\
    &=\frac{(-1)^{2k+2\rr+1}}{(2 \pi \mathrm{i})^2} `\int_{\mathcal{E}_2} \int_{\mathcal{E}_1} \frac{w^{\frac{1}{2}(j+k+\rr+1)}}{z^{\frac{1}{2}(j+k+\rr+1)}} \frac{(1+z)^{\frac{1}{2}(k-j+\rr-1)}(z-1)^{\frac{1}{2}(j-k+\rr+1)}}{(1+w)^{\frac{1}{2}(k-j+\rr+1)}(w-1)^{\frac{1}{2}(j-k+\rr+1)}} \frac{\mathrm d z \mathrm d w}{w-z} \\
    &=\frac{-1}{(2 \pi \mathrm{i})^2}\int_{\mathcal{E}_2} \int_{\mathcal{E}_1} \frac{w^{\frac{1}{2}(j+k+\fjk)}}{z^{\frac{1}{2}(j+k+\fjk)}} \frac{(1+z)^{\frac{1}{2}(k-j+\fjk)}(z-1)^{\frac{1}{2}(j-k+\fjk)}}{(1+w)^{\frac{1}{2}(k-j+\fjk)}(w-1)^{\frac{1}{2}(j-k+\fjk)}}\\
    &\qquad \qquad \qquad \times\frac{1}{1+z} \left( \frac{w(1+z)(z-1)}{z(1+w)(w-1)}\right)^{\frac{\epsijk+1}{2}}  \frac{\mathrm d z\mathrm d w}{w-z} \\
    &=\frac{-1}{(2 \pi \mathrm{i})^2}\int_{\mathcal{E}_2} \int_{\mathcal{E}_1} \mathrm{e}^{(n+1)(\tS(w)-\tS(z))} \frac{1}{1+z} \left( \frac{w(1+z)(z-1)}{z(1+w)(w-1)}\right)^{\frac{\epsijk+1}{2}} \frac{\mathrm d z \mathrm d w}{z-w}
    \end{aligned}
\]
and other cases proceed similarly.
\end{proof}

As a corollary to Lem. \ref{lem: generalized_double_contour_int}, we have the following expression for $\widetilde{\cT}_n(j,k)$.

\begin{cor}
\begin{equation}
\resizebox{0.9\linewidth}{!}{$
    \label{eq: contour T-embedding}
    \begin{aligned}
        \mathcal{\widetilde T}_n(j, k)&=\frac{1}{(2 \pi \mathrm{i})^2} \int_{\mathcal{E}_2} \int_{\mathcal{E}_1} \mathrm e^{(n+1)(\tS(w;\frac{j}{n+1},\frac{k}{n+1})-\tS(z;\frac{j}{n+1},\frac{k}{n+1}))} G_{\mathcal{T}}(z, w) \left( \frac{w(1+z)(z-1)}{z(1+w)(w-1)}\right)^{\frac{\epsijk+1}{2}}\frac{\mathrm d z \mathrm d w}{z-w}-1
    \end{aligned}
$}
\end{equation}
with the contour $\mathcal E_1$ and $\mathcal E_2$ are the same as those in the case of the Aztec diamond in Lem. \ref{lem: t_embedding_coord_AD}.
\end{cor}

\begin{proof}
The result follow immediately from Lem. \ref{lem: generalized_double_contour_int} and the fact that 
\begin{equation}
\resizebox{0.9\linewidth}{!}{$
\begin{aligned}
\label{eq: T expression in p}
 \widetilde{\cT}_n(j,k)&\;= {\mathcal{T}}_{r(n;j,k)}(j, k)\\
 &\;=p_E(j, k, r(n;j,k))+\mathrm i p_N(j, k, r(n;j,k))-p_W(j, k, r(n;j,k))-\mathrm i p_S(j, k, r(n;j,k)).
 \end{aligned}
 $}
\end{equation}
\end{proof}

Note that we may also express  $\widetilde{\cT}_n(j,k)$ as
\begin{equation} \label{eq:tildeTintS}
 \widetilde{\cT}_n(j,k) =
       \frac{1}{(2 \pi \mathrm{i})^2} \int_{\mathcal{E}_2} \int_{\mathcal{E}_1} \mathrm e^{\left(r+1\right)\left(S(w;\frac{j}{r+1},\frac{k}{r+1})-S(z;\frac{j}{r+1},\frac{k}{r+1})\right)} G_{\mathcal{T}}(z, w) \frac{\mathrm d z\mathrm d w}{z-w}-1
\end{equation}
by directly using the contour integral formula for $\cT_n(j,k)$.

Recall that the image of the limit of the perfect t-embedding lies in $\lozenge:=\{|x|+|y|<1\}$. Denote $G^f$ and $\mathfrak{L}$ are the rescaled generalized tower graph indexed by some function $f$ and the interior of the arctic curves (the liquid region) inscribed in $G^f$. The action $\tS(z)$ has critical points that satisfy $\tS'(z)=0$ and are conjugate pairs
\[
z = \frac{-x+y\pm \mathrm i\sqrt{2(x^2+y^2)-f(x,y)^2}}{f(x,y)-x-y}.
\]
Denote the critical point in the upper half-plane $\mathbb H$ by $\displaystyle  \xi$ and $\theta_\xi:=\frac{1}{2}\operatorname{arg}(\tS''(\xi))$. 
\begin{remark}
    \label{remark: arctic_from_crit_point}
    The boundary of the liquid region $\mathfrak L$ can be determined by the critical points of the action $\tS$. Namely, the arctic curves are determined by $\tS'(z)=0$ when $z$ is pure real/double-root, i.e. $f(x,y)^2=2(x^2+y^2)$. Then, $\mathfrak L$ is the domain with boundary contour $f(x,y)^2=2(x^2+y^2)$.
\end{remark}
Let $j, k, n\in \mathbb{Z}$, $\iota, \iota^{\prime}\in \{0, \pm 1\}$ be such that $|\iota|+\left|\iota^{\prime}\right|=1$ or $\io+\io'=0$. We denote the edges in the dual graph $\left(G'^f_n\right)^*$ by $e_{\iota, \iota^{\prime}}(j, k)$, where $e_{\iota, \iota^{\prime}}(j, k)$ is the edge between $(j, k)$ and $\left(j+\iota, k+\iota^{\prime}\right)$. {Note that unlike the Aztec diamond in \cite{tower_AD_perfect}, we include the case $\io+\io'=0$ corresponding to an edge between two hexagonal faces.} We orient the edges so that they have a white vertex to the left. Recall that
$$
d \widetilde{\mathcal{T}}_n\left(e_{\iota, \iota^{\prime}}(j, k)\right)= \pm\left(\widetilde{\mathcal{T}}_n\left(j+\iota, k+\iota^{\prime}\right)-\widetilde{\mathcal{T}}_n(j, k)\right)
$$
where the sign depends on the orientation of the edge. Our goal is to show that the rigidity assumption holds for the perfect t-embedding of the generalized tower graph. We reproduce the assumption below for convenience.

\begin{assumption}
    Given a compact set $\mathcal{K} \subset \mathfrak{L}$, there exist positive constants $N_{\mathcal{K}}, C_{\mathcal{K}}$ and $\varepsilon_{\mathcal{K}}$ which only depend on $\mathcal{K}$, such that for all edges $e_{\iota, \iota^{\prime}}(j, k) \subset \mathcal{K}$,

$$
\frac{\mu_n}{C_{\mathcal{K}}} \leq\left|d \widetilde{\mathcal{T}}_n\left(e_{\iota, \iota^{\prime}}(j, k)\right)\right| \leq \mu_n C_{\mathcal{K}}
$$
for all $n>N_{\mathcal{K}}$. In addition, for $n>N_{\mathcal{K}}$ and $e_{\iota, \iota^{\prime}}(j, k), e_{\tilde{\iota}, \tilde{\iota}^{\prime}}(j, k) \subset \mathcal{K}$ the angle between two adjacent edges $d \mathcal{T}_n\left(e_{\iota, \iota^{\prime}}(j, k)\right)$ and $d \mathcal{T}_n\left(e_{\tilde{\imath},\tilde \iota^{\prime}}(j, k)\right)$, which are adjacent to a common face, is contained in $\left(\varepsilon_{\mathcal{K}}, \pi-\varepsilon_{\mathcal{K}}\right)$. These angles are precisely the angles of the faces of the image of the perfect t-embedding.
\end{assumption}

\subsection{Steepest descent analysis}
To prove that the rigidity assumption holds, we must understand the leading term asymptotic of $d \widetilde{\mathcal{T}}_n\left(e_{\iota, \iota^{\prime}}(j, k)\right)$ as $n\to \infty$. To that end, we make use of Eqn. \eqref{eq: contour T-embedding} and perform a steepest descent analysis. 

In order to understand $d\widetilde{\mathcal{T}}_n$ we fix an edge $e_{\io,\io'}(j,k)$. Note that, in the case $|\io|+|\io'|=1$ in which the edge connects faces in different columns, we may assume the height of column $j+\io+k+\io'$ is greater than that of column $j+k$. If this is not the case we may swap the roles of $(j,k)$ and $(j+\io,k+\io')$ which will only contribute an overall sign to $d\widetilde{\mathcal{T}}_n$. We have the following lemma.

\begin{lemma}
{For $|\io|+|\io'|=1$,} we have
\begin{equation}
\label{eq: perturbed_t_embedd_contour}
\begin{aligned}
      &\, \widetilde{\cT}_{n}(j+\io,k+\io')\\
      &= \frac{1}{(2 \pi \mathrm{i})^2} \int_{\mathcal{E}_2} \int_{\mathcal{E}_1} \mathrm e^{\left(r(n;j,k)+1\right)\left(S(w;\frac{j}{r(n;j,k)+1},\frac{k}{r(n;j,k)+1})-S(z;\frac{j}{r(n;j,k)+1},\frac{k}{r(n;j,k)+1})\right)} \\
      &\qquad \qquad\qquad \times \mathrm e^{S(w;\io,\io')- S(z;\io,\io')} G_{\mathcal{T}}(z, w) \frac{\mathrm d z\mathrm d w}{z-w}-1.
\end{aligned}
\end{equation}
{
For $\io+\io'=0$, we have
\begin{equation}
\label{eq: perturbed_t_embedd_contour2}
\begin{aligned}
      &\, \widetilde{\cT}_{n}(j+\io,k+\io')\\
      &= \frac{1}{(2 \pi \mathrm{i})^2} \int_{\mathcal{E}_2} \int_{\mathcal{E}_1} \mathrm e^{\left(r(n;j,k)+1\right)\left(S(w;\frac{j}{r(n;j,k)+1},\frac{k}{r(n;j,k)+1})-S(z;\frac{j}{r(n;j,k)+1},\frac{k}{r(n;j,k)+1})\right)} \\
      &\qquad \qquad\qquad \times  \left(\frac{(z-1)(w+1)}{(z+1)(w-1)}\right)^{\pm 1} G_{\mathcal{T}}(z, w) \frac{\mathrm d z\mathrm d w}{z-w}-1
\end{aligned}
\end{equation}
where the sign of the exponent matches the sign of $\io$.
}
\end{lemma}

\begin{proof}
Recall that we have $\widetilde{\cT}_{n}(j,k) = \cT_{r(n;j,k)}(j,k)$ with $r(n;j,k)+j+k$ odd. Let $r:=r(n;j,k)$. 

{In the first case,} as we assume the face $(j,k)$ is in a column whose height is less than the height of the column containing the face $(j+\io,k+\io')$, we have
\[
\widetilde{\cT}_{n}(j+\io,k+\io') = \cT_{r}(j+\io,k+\io') = \cT_{r+1}(j+\io,k+\io') 
\]
where the second equality holds for the AD when $r+j+k+\io+\io'$ is even, as is the case here. Note that $r+1+j+k+\io+\io'$ is odd so we may use the contour integral formula for $\cT_{r+1}(j+\io,k+\io')$ to write
\[
\begin{aligned}
      &\, \widetilde{\cT}_{n}(j+\io,k+\io')=
       \frac{1}{(2 \pi \mathrm{i})^2} \int_{\mathcal{E}_2} \int_{\mathcal{E}_1} \mathrm e^{\left(r+2\right)\left(S(w;\frac{j+\io}{r+2},\frac{k+\io'}{r+2})-S(z;\frac{j+\io}{r+2},\frac{k+\io'}{r+2})\right)} G_{\mathcal{T}}(z, w) \frac{\mathrm d z\mathrm d w}{z-w}-1.
\end{aligned}
\]
The result then follows from the fact that
\begin{equation*}
    \label{eq:action recursion}
    (n+2)S(z; \frac{j+\io}{n+2},\frac{k+\io'}{n+2})=(n+1)S(z;\frac{j}{n+1},\frac{k}{n+1})+ S(z;\io,\io').
\end{equation*}

{
In the second case, we note that
\begin{equation*}
\begin{aligned}
(n+1)S(z; \frac{j+\io}{n+1},\frac{k+\io'}{n+1}) =&\; (n+1)S(z; \frac{j}{n+1},\frac{k}{n+1})\\
&\; -\frac{\io'-\io}{2} \log(z+1) - \frac{\io-\io'}{2} \log(z-1).
\end{aligned}
\end{equation*}
The result then follows from applying the above to the contour integral formula for $\cT_r(j+\io,k+\io')$ observing that $\io-\io'=\pm2$.
}
\end{proof}

Using this we obtain a contour integral formula for $d \widetilde{\mathcal{T}}_n\left(e_{\iota, \iota^{\prime}}(j, k)\right)$.

\begin{cor}
The difference of the t-embedding along an edge can be written as
\begin{equation}
    \label{eq:dT equation factor}
     \begin{aligned}
         d \widetilde{\mathcal{\cT}}_n\left(e_{\iota, \iota^{\prime}}(j, k)\right)&= \pm \frac{(1-\mathrm i)}{(2 \pi \mathrm{i})^2} \left(\int_{\mathcal{E}_2}\mathrm e^{(n+1)\tS(w)}\frac{w-\mathrm i}{w}\left (\frac{w}{(1+w)(w-1)}\right )^{\frac{\epsilon(j,k)+1}{2}} G_{\io,\io'}(w)\operatorname{d}w\right)\\
         &\times \left(\int_{\mathcal{E}_1} \mathrm{e}^{-(n+1)\tS(z)}\frac{z-\mathrm i}{(z-1)(z+1)}\left(\frac{(1+z)(z-1)}{z}\right)^{\frac{\epsilon(j,k)+1}{2}}H_{\io,\io'}(z)  \operatorname{d}z\right).
    \end{aligned} 
\end{equation}
where
$$
\begin{aligned}
G_{1,0}(w)&=\frac{1}{w-1}, \quad G_{0,1}(w) =\frac{-1}{w+1},
\quad G_{-1,0}(w)=\frac{1}{w+1}, \quad G_{0,-1}(w)=\frac{1}{w-1},\\
H_{1,0}(z)&= H_{0,1}(z)=\frac{1}{z},
, \quad H_{-1,0}(z)=H_{0,-1}(z)=1, \\
{
G_{1,-1}(w)} &= \frac{2}{w-1}, \quad G_{-1,1}(w)=\frac{-2}{w+1},\quad H_{1,-1}(z) = \frac{1}{z+1}, \quad \text{and} \quad H_{-1,1}(z) = \frac{1}{z-1}.
\end{aligned}
$$
\end{cor}

\begin{proof}
    Using integral expressions in \eqref{eq:tildeTintS} and \eqref{eq: perturbed_t_embedd_contour}, in the case $|\io|+|\io'|=1$,  we may write $d \widetilde{\mathcal{\cT}}_n\left(e_{\iota, \iota^{\prime}}(j, k)\right)$ as
    
\begin{equation*}
    \label{eq: contour_integral_single_edge}
    \resizebox{0.9\linewidth}{!}{$
    \begin{aligned}
    & \operatorname{d}\widetilde{\cT}_n(e_{\io,\io'}(j,k))=\pm\frac{1}{(2\pi \mathrm{i})^2}\int_{\mathcal E_2}\int_{\mathcal E_1}\mathrm{e}^{(r+1)(S(w;\frac{j}{r+1},\frac{k}{r+1})-S(z;\frac{j}{r+1},\frac{k}{r+1}))}G_{\cT}(z,w)\left(\mathrm{e}^{S(w;\io,\io')-S(z;\io,\io')}-1\right)\frac{\operatorname{d}z\operatorname{d}w}{z-w}\\
    &=\frac{1}{(2\pi \mathrm{i})^2}\int_{\mathcal E_2}\int_{\mathcal E_1}\mathrm{e}^{(n+1)(\tS(w;x,y)-\tS(z;x,y))}G_{\cT}(z,w)\left(\mathrm{e}^{S(w;\io,\io')-S(z;\io,\io')}-1\right)\left( \frac{w(1+z)(z-1)}{z(1+w)(w-1)}\right)^{\frac{\epsijk+1}{2}}\frac{\operatorname{d}z\operatorname{d}w}{z-w}
    \end{aligned}
    $}
\end{equation*}
where the last equality is obtained by applying (\ref{eq: contour T-embedding}). Notice that the term $\displaystyle\left(\mathrm{e}^{S(w;\io,\io')-S(z;\io,\io')}-1\right)$ simplifies to
\[
\begin{cases}
    \frac{z-w}{z(w-1)}, & \io=1,\io'=0 \\
    \frac{z-w}{w+1}, & \io=-1,\io'=0 \\
    -\frac{z-w}{z(w+1)}, & \io=0,\io'=1 \\
    \frac{z-w}{w-1}, & \io=0,\io'=-1
\end{cases}
\]
which cancels out the factor of $z-w$ in the denominator. 

{
In the case $\io+\io'=0$, we have the integral expression
\begin{equation*}
    \resizebox{0.9\linewidth}{!}{$
    \begin{aligned}
    & \operatorname{d}\widetilde{\cT}_n(e_{\io,\io'}(j,k))\\
    &=\frac{1}{(2\pi \mathrm{i})^2}\int_{\mathcal E_2}\int_{\mathcal E_1}\mathrm{e}^{(n+1)(\tS(w;x,y)-\tS(z;x,y))}G_{\cT}(z,w)\left(\left(\frac{(z-1)(w+1)}{(z+1)(w-1)}\right)^{\pm 1}-1\right)\left( \frac{w(1+z)(z-1)}{z(1+w)(w-1)}\right)^{\frac{\epsijk+1}{2}}\frac{\operatorname{d}z\operatorname{d}w}{z-w}.
    \end{aligned}
    $}
\end{equation*}
The term $\left(\frac{(z-1)(w+1)}{(z+1)(w-1)}\right)^{\pm 1}-1$ simplifies to
\[
\begin{cases}
    \frac{2(z-w)}{(z+1)(w-1)}, & \io=1,\io'=-1 \\
    \frac{-2(z-w)}{(z-1)(w+1)}, & \io=-1,\io'=1
\end{cases}
\]
where, again, the factor of $z-w$ in the denominator cancels.
}

{In all cases}, the two integrals factor and we can write the double contour integral formula as the product of two single integrals. 
\end{proof}

This expression is similar to (6.4) of \cite{tower_AD_perfect} up to the error term and thus the verification of the rigidity assumption will proceed similarly. We now perform the steepest descent/saddle point approximation separately for each factor. From here on, we denote $\tS(z; x,y)=\tS(z)$ since there is no more subtlety with the parameter.

\begin{lemma}
    \label{lemma: steepest_descent_lemma}
    \begin{equation}
      \label{eq: asymptotic of w integral 
}
\begin{aligned}
&\quad I_w(n):=\int_{\mathcal{E}_2}\mathrm e^{(n+1)\tS(w)}\frac{w-\mathrm i}{w}\left (\frac{w}{(1+w)(w-1)}\right )^{\frac{\epsilon(j,k)+1}{2}} G_{\io,\io'}(w)\operatorname{d}w\\
        &\sim G(\xi)\mathrm e^{(n+1) \tS(\xi)}\sqrt{\frac{\pi}{n+1}}\sqrt{\frac{2}{|\tS''(\xi)|}}e^{-i\theta_\xi}+G(\bar \xi)e^{(n+1) \tS(\bar \xi)}\sqrt{\frac{\pi}{n+1}}\sqrt{\frac{2}{|\tS''(\bar \xi)|}}e^{-i\theta_{\bar \xi}}
\end{aligned}
    \end{equation}
where $$G(w):=\frac{w-\mathrm i}{w}\left (\frac{w}{(1+w)(w-1)}\right )^{\frac{\epsilon(j,k)+1}{2}}  G_{\io,\io'}(w).$$
\end{lemma}

 \begin{proof}
     We perform the steepest descent asymptotic analysis on the integral $I_w(m)$. First, recall that the action $\tS(z)$ has two critical points/saddle points that are simple and also conjugate pairs.
     And since $\tS(z)$ is analytic in the upper half plane $\mathbb H$, the paths of steepest ascent and descent are the level lines
of $\operatorname{Im}(\tS(z))=\operatorname{Im}(\tS(\xi))$. These paths can only cross the real line at 0. We now describe
these paths. For the saddle point $\xi$ on the upper-half plane, the steepest ascent contour start from $0$, passes through $\xi$ and goes to $\infty$. The complete steepest ascent includes the reflection of the mentioned path in the lower half plane. We will denote the steepest ascent contour in $w$ as $\gamma_w$. Note that the contour $\mathcal E_2$ only has a singularity at $w=1$. Thus, we can deform $\mathcal E_2$ to the steepest descent contour without picking up the residue at $w=1$. The Taylor expansion near the saddle point $\xi$ of the action $\tS(w)$ gives,
     \[
     \tS(w) = \tS(\xi)+\frac{1}{2}\tS''(\xi)(w-\xi)^2 +O(w^3)
     \]
     since the first derivative vanishes at $\xi$. Writing $ \tS(w)= u(w)+\mathrm i  v(w)$,  then along the steepest descending path, the integral $I_\xi(n)$ is approximated, 
     \begin{equation}
         \label{eq: I(n) approx}
    \begin{aligned}
    I_\xi(n)&=\mathrm e^{(n+1)\tS(\xi)}\int_{\mathcal E_2}G(w)\mathrm e^{(n+1)(\tS(w)-\tS(\xi))}\operatorname{d}w\\
    &\sim G(\xi)\mathrm e^{(n+1)\tS(\xi)}\int_{\mathcal E_2}\mathrm e^{\mathrm i (n+1)(v(w)-v(\xi))}\mathrm e^{(n+1)(u(w)-u(\xi))}\operatorname{d}w \\
    &\sim G(\xi)\mathrm e^{(n+1) \tS(\xi)}\int_{\gamma_w}\mathrm e^{(n+1)(u(w)-u(\xi))}\mathrm dw
    \end{aligned}    
     \end{equation}
     where $\gamma_w$ is the steepest descent contour. Note that along this contour, the imaginary part of $\widetilde\cS$ stays constant so $v(w)=v(\xi)$. Let $s^2=u(\xi)-u(w)$, then we can rewrite \eqref{eq: I(n) approx} as
    \begin{equation}
        \label{eq: I(n) approx 2}
        \begin{aligned}
            I_\xi(n) \sim G(\xi)\mathrm e^{(n+1)\tS(\xi)}\int_{-\infty}^\infty \mathrm e^{-(n+1)s^2}\left(\frac{\mathrm dw}{\mathrm ds}\right)\mathrm ds.
        \end{aligned}
    \end{equation}
    Note that we have the Gaussian integral $\displaystyle \int_{-\infty}^\infty \mathrm e^{-(n+1)s^2} \ds=\sqrt{\frac{\pi}{n+1}}$, so the leading term of the asymptotic expression near the saddle point where $s=0$ gives 
    \[
    I_\xi(n) \sim G(\xi)\mathrm e^{(n+1)\tS(\xi)}\left(\frac{\mathrm dw}{\mathrm ds}\right)_{s=0} \sqrt{\frac{\pi}{n+1}}.
    \]
    On the other hand, since $v(w)$ is constant along the steepest descent path, $-s^2=u(w)-u(\xi)=\tS(w)-\tS(\xi)$. By taking derivative with respect of $w$ twice, we have 
    \[-2s\left(\frac{\ds}{\dw}\right)=\tS'(w),\]
    \[
    -2\left(\frac{\ds}{\dw}\right)\left(\frac{\ds}{\dw}\right)-2s\frac{\mathrm d^2s}{\dw^2}=\tS''(w)
    \]
    and so $\displaystyle \left(\frac{\dw}{\ds}\right)_{s=0}=\left(\frac{-2}{\tS''(\xi)}\right)^{1/2}$. 
   This square root has some ambiguity in sign. To avoid this issue, denote $\displaystyle \theta_{SDP}:=\operatorname{arg}\left(\frac{\mathrm d w}{\mathrm d s}\right)_{s=0}$. Then, 
  $ 
   \displaystyle \left(\frac{\dw}{\ds}\right)_{s=0}=\left(\frac{2}{|\tS''(\xi)|}\right)^{1/2}e^{i\theta_{SDP}}$. 
   On the other hand, recall that we define $\theta_\xi:=\frac{1}{2}\operatorname{arg}(\tS''(\xi))$. Let $\tS''(\xi)=R\mathrm e^{\mathrm 2i\theta_\xi}$ and $w-\xi=r\mathrm e^{\mathrm i \theta_{SDP}}$, we have,
   \[
     \tS(w)-\tS(\xi) = \frac{1}{2}\tS''(\xi)(w-\xi)^2 +O(w^3)\sim\frac{1}{2}(Rr^2)\mathrm e^{2i(\theta_\xi+\theta_{SDP})}.
     \]
    But note that since we are on the steepest descent path $\gamma_w$, $2(\theta_\xi+\theta_{SDP})=\pm \pi$ and thus $\theta_{SDP}=-\theta_\xi\pm\frac{\pi}{2}$ (where the sign is dependent on the direction of the contour) and we have the asymptotic of $I_\xi(n)$ given by, 
    \begin{equation}
        \label{eq: I(n) approx 3}
        I_\xi(n)\sim G(\xi)\mathrm e^{(n+1) S(\xi)}\sqrt{\frac{\pi}{n+1}}\sqrt{\frac{2}{|S''(\xi)|}}\mathrm e^{-i\theta_\xi}.
    \end{equation}
 Recall that the contour $\mathcal E_2$ has also passes through the other critical point $\bar \xi$ on the lower half-plane. The analysis is similar and we have
 \begin{equation}
        \label{eq: I(m) approx 4 conjugate}
        I_{\bar \xi}(n)\sim G(\bar \xi)\mathrm e^{(n+1)S(\bar \xi)}\sqrt{\frac{\pi}{n+1}}\sqrt{\frac{2}{|S''(\bar \xi)|}}\mathrm e^{-i\theta_{\bar \xi}}.
    \end{equation}
Altogether, we have that the leading order asymptotic of the integral $I_w(n)$ is given by
$$
\begin{aligned}
I_w(n)&\;=I_\xi(n)+I_{\bar \xi}(n)\\
&\;\sim G(\xi)\mathrm e^{(n+1) \tS(\xi)}\sqrt{\frac{\pi}{n+1}}\sqrt{\frac{2}{|\tS''(\xi)|}}\mathrm e^{-i\theta_\xi}+G(\bar \xi)\mathrm e^{(n+1) \tS(\bar \xi)}\sqrt{\frac{\pi}{n+1}}\sqrt{\frac{2}{|\tS''(\bar \xi)|}}\mathrm e^{-i\theta_{\bar \xi}}.
\end{aligned}
$$
 \end{proof}

Similarly, we can also perform the steepest descent analysis for the integral in $z$ along the contour $\mathcal E_1$. 

\begin{cor}
    \label{cor: steepest descent in z}
    $$I_z(n):=\int_{\mathcal{E}_1} \mathrm{e}^{-(n+1)\tS(z)}\frac{z-\mathrm i}{(z-1)(z+1)}\left(\frac{(1+z)(z-1)}{z}\right)^{\frac{\epsilon(j,k)+1}{2}}H_{\io,\io'}(z)  \operatorname{d}z.$$
    As $n \to \infty$, we have 
    \begin{equation}
    \begin{aligned}
        &I_z(n)=I_\xi(n)+I_{\bar \xi}(n) \\
        &\sim H(\xi)\mathrm e^{-(n+1) \tS(\xi)}\sqrt{\frac{\pi}{n+1}}\sqrt{\frac{2}{|\tS''(\xi)|}}\mathrm e^{-i\theta_\xi}-H(\bar \xi)\mathrm e^{-(n+1) \tS(\bar \xi)}\sqrt{\frac{\pi}{n+1}}\sqrt{\frac{2}{|\tS''(\bar \xi)|}}\mathrm e^{-i\theta_{\bar \xi}}.
    \end{aligned}
    \end{equation}
where 
\[
H(z):= \frac{z-\mathrm i}{(z-1)(z+1)}\left(\frac{(1+z)(z-1)}{z}\right)^{\frac{\epsilon(j,k)+1}{2}}H_{\io,\io'}(z).
\]
\end{cor}
\begin{proof}
    The analysis is similar to Lem. \ref{lemma: steepest_descent_lemma} except for the contour of integration. For the contour integral in $z$, we have the steepest descent contour from $-1$, which passes through $\xi$ and $1$ together with its reflection on the lower half-plane. 
\end{proof}
 
Combining Lem. \ref{lemma: steepest_descent_lemma} and Cor. \ref{cor: steepest descent in z}, we have the asymptotic expansion of the edge length. 

\begin{lemma}
\label{lemma: edge length asymptotic}
    \begin{equation}
    \label{eq: edge length asymptotic}
    \resizebox{0.9\textwidth}{!}{$
        \begin{aligned}
            d \widetilde{\mathcal{\cT}}_n\left(e_{\iota, \iota^{\prime}}(j, k)\right) &\sim \frac{1 -\mathrm i}{(2\pi \mathrm i)^2} \left(G(\xi)e^{(n+1) \tS(\xi)}\sqrt{\frac{\pi}{n+1}}\sqrt{\frac{2}{|\tS''(\xi)|}}e^{-i\theta_\xi}+G(\bar \xi)e^{(n+1) \tS(\bar \xi)}\sqrt{\frac{\pi}{n+1}}\sqrt{\frac{2}{|\tS''(\bar \xi)|}}e^{-i\theta_{\bar \xi}}\right)\\
            &\times \left(H(\xi)e^{-(n+1) \tS(\xi)}\sqrt{\frac{\pi}{n+1}}\sqrt{\frac{2}{|\tS''(\xi)|}}e^{-i\theta_\xi}-H(\bar \xi)e^{-(n+1) \tS(\bar \xi)}\sqrt{\frac{\pi}{n+1}}\sqrt{\frac{2}{|\tS''(\bar \xi)|}}e^{-i\theta_{\bar \xi}}\right)\\
            &= \frac{1 -\mathrm i}{(2\pi \mathrm i)^2}\frac{\pi}{n+1}\frac{2}{|\tS''(\xi)|} \times \Bigg{(} \frac{1}{e^{2\mathrm i\theta_\xi}}G(\xi)H(\xi)+e^{2\mathrm i (n+1) \operatorname{Im(\tS(\xi))}}G(\bar \xi)H(\xi)\\
          &-\frac{1}{e^{2\mathrm i \theta_{\bar\xi}}}G(\xi)H(\bar \xi)-e^{2\mathrm i (n+1) \operatorname{Im(\tS(\xi))}}G(\xi)H(\bar \xi)\Bigg ).\\
        \end{aligned} 
        $}
    \end{equation}
\end{lemma}
\begin{proof}
    The equality comes from the facts $\displaystyle e^{(n+1)(\tS(\xi)-\tS(\bar \xi))}=e^{2\mathrm i (n+1)\operatorname{Im}(\tS(\xi))}$ and $|\tS(\xi)|=|\tS(\bar \xi)|$. Since the critical points are in conjugate pairs $\xi , \bar \xi$ and the contour is symmetric with respect to the real line, we have the final expression as in \eqref{eq: edge length asymptotic}
\end{proof}
 
\subsection{Edges and angles bounds -  rigidity condition}
All proofs and results in this section are identical to those in Section 6.3 of \cite{tower_AD_perfect}, with only some differences in notation to accommodate the error terms coming from approximating the function $f(x,y)$ in our shuffling algorithm. We follow their proof and explanation closely and only change the computation wherever the error terms appear. 

We would like to use the asymptotic / leading terms in Lem. \ref{lemma: edge length asymptotic} to prove the boundedness of the lengths and angles and confirm the rigidity assumption. Recall that 
$$G(w)=\frac{w-\mathrm i}{w}\left (\frac{w}{(1+w)(w-1)}\right )^{\frac{\epsilon(j,k)+1}{2}} G_{\io,\io'}(w),$$
and 
$$ H(z)= \frac{z-\mathrm i}{(z-1)(z+1)}\left(\frac{(1+z)(z-1)}{z}\right)^{\frac{\epsilon(j,k)+1}{2}}H_{\io,\io'}(z).$$

Let 
\[
\al = \frac{1}{e^{2\mathrm i\theta_\xi}}, \qquad  \beta = e^{-2\mathrm i (n+1)\operatorname{Im(\tS(\xi))}}.
\]
Then $H$ and $G$ satisfy the following identities:
\begin{equation}
    \label{eq: 2 summands}
    \resizebox{0.9\textwidth}{!}{$
    \begin{aligned}
         \al G(\xi)H(\xi)+\beta G(\bar \xi)H(\xi)&= \frac{(\xi -\ci)}{(\xi -1)(\xi +1)}H_{\io,\io'}(\xi)\left(\frac{(\xi-\ci)G_{\io,\io'}(\xi)}{\xi}\al+\frac{(\bar \xi -\ci)G_{\io,\io'}(\bar \xi)}{\bar \xi }\left(\frac{\bar \xi (\xi^2-1)}{\xi(\bar \xi^2-1)}\right)^{\frac{\epsilon(j,k) +1}{2}}\beta\right),\\
        \bar\beta G(\xi)H(\bar \xi)+\bar \al G(\bar \xi)H(\bar \xi)&= \frac{(\bar\xi -\ci)}{(\bar\xi -1)(\bar\xi +1)}H_{\io,\io'}(\bar\xi)\left(\frac{(\xi-\ci)G_{\io,\io'}(\xi)}{\xi}\left(\frac{ \xi (\bar\xi^2-1)}{\bar\xi(\xi^2-1)}\right)^{\frac{\epsilon(j,k) +1}{2}}\bar\beta+\frac{(\bar \xi -\ci)G_{\io,\io'}(\bar \xi)}{\bar \xi }\bar\al\right).
    \end{aligned}  
    $}
\end{equation}
Notice that 
\begin{equation*}  
\begin{aligned}
&\left(\frac{(\xi-\ci)G_{\io,\io'}(\xi)}{\xi}\left(\frac{ \xi (\bar\xi^2-1)}{\bar\xi(\xi^2-1)}\right)^{\frac{\epsilon(j,k) +1}{2}}\bar\beta+\frac{(\bar \xi -\ci)G_{\io,\io'}(\bar \xi)}{\bar \xi }\bar\al\right)\al\beta\left(\frac{\bar \xi (\xi^2-1)}{\xi(\bar \xi^2-1)}\right)^{\frac{\epsilon(j,k) +1}{2}}\\
&=\left(\frac{(\xi-\ci)G_{\io,\io'}(\xi)}{\xi}\al+\frac{(\bar \xi -\ci)G_{\io,\io'}(\bar \xi)}{\bar \xi }\left(\frac{\bar \xi (\xi^2-1)}{\xi(\bar \xi^2-1)}\right)^{\frac{\epsilon(j,k) +1}{2}}\beta\right).
\end{aligned}
\end{equation*}
since $|\al|=|\beta|=1$. We have 
\[
\begin{aligned}
    &\al G(\xi)H(\xi)+\beta G(\bar \xi)H(\xi) \\
    & \quad = \al\beta\left(\frac{\bar \xi (\xi^2-1)}{\xi(\bar \xi^2-1)}\right)^{\frac{\epsilon(j,k) +1}{2}}\frac{(\xi-\ci)(\bar \xi^2 -1)H_{\io,\io'}(\xi)}{(\bar\xi-\ci)(\xi^2 -1)H_{\io,\io'}(\bar \xi)}(\bar\beta G(\xi)H(\bar \xi)+\bar \al G(\bar \xi)H(\bar \xi)).
\end{aligned}
\]
and thus
\begin{equation}
    \label{eq: factoring_4_terms}
    \resizebox{0.9\textwidth}{!}{$
    \begin{aligned}
        &\quad\al G(\xi)H(\xi)+\beta G(\bar \xi)H(\xi)-\bar\beta G(\xi)H(\bar \xi)-\bar \al G(\bar \xi)H(\bar \xi)\\
        &=\left(\al\beta\left(\frac{\bar \xi (\xi^2-1)}{\xi(\bar \xi^2-1)}\right)^{\frac{\epsilon(j,k) +1}{2}}\frac{(\xi-\ci)(\bar \xi^2 -1)H_{\io,\io'}(\xi)}{(\bar\xi-\ci)(\xi^2 -1)H_{\io,\io'}(\bar \xi)}-1\right)(\bar\beta G(\xi)H(\bar \xi)+\bar \al G(\bar \xi)H(\bar \xi))\\
        &=\left(\al\beta\left(\frac{\bar \xi (\xi^2-1)}{\xi(\bar \xi^2-1)}\right)^{\frac{\epsilon(j,k) +1}{2}}\frac{(\xi-\ci)(\bar \xi^2 -1)H_{\io,\io'}(\xi)}{(\bar\xi-\ci)(\xi^2 -1)H_{\io,\io'}(\bar \xi)}-1\right)\left(\al\bar\beta\frac{\xi-\ci}{\bar \xi-\ci}\frac{G_{\io,\io'}(\xi)}{G_{\io,\io'}(\bar \xi)}\frac{\bar \xi}{\xi}\left(\frac{ \xi (\bar\xi^2-1)}{\bar\xi(\xi^2-1)}\right)^{\frac{\epsilon(j,k) +1}{2}}+1\right)\\
        &\qquad \times\frac{(\bar \xi-\ci)^2H_{\io,\io'}(\bar \xi)G_{\io,\io'}(\bar \xi)\bar \al}{(\bar \xi^2-1)\bar \xi}.
    \end{aligned}
    $}
\end{equation}

Finally, we have the boundedness for the edges in \emph{rigidity assumption}, with $\displaystyle \mu_n=\frac{1}{n+1}$.

\begin{lemma}
\label{lemma: rigid_edge_bounds}
    Let $\mathcal{K} \subset \mathcal{L}$ be a compact set. There exist positive constants $N_{\mathcal{K}}, C_{\mathcal{K}}>1$ such that
$$
\frac{1}{n+1} \frac{1}{C_{\mathcal{K}}} \leq\left|d {\widetilde\cT}_n\left(e_{\iota, \iota^{\prime}}(j, k)\right)\right| \leq \frac{1}{n+1} C_{\mathcal{K}}
$$
for all edges $e_{\iota, \iota^{\prime}}(j, k) \subset \mathcal{K}$ and all $n>N_{\mathcal{K}}$.
\end{lemma}

\begin{proof}
    We set $(x, y)=\frac{1}{n+1}(j, k)$. Note that for $(x, y) \in \mathcal{K}$, $\xi=\xi(x, y)$ is bounded away from the real line, and $\left|\tS^{\prime \prime}(\xi)\right|$ is bounded away from zero and infinity. Since the approximation as $n \to \infty$ in Lem. \ref{lemma: edge length asymptotic} is uniform on compact subsets, it is sufficient to prove that
$$
\al G(\xi)H(\xi)+\beta G(\bar \xi)H(\xi)-\bar\beta G(\xi)H(\bar \xi)-\bar \al G(\bar \xi)H(\bar \xi)
$$
is bounded away from zero and infinity on $\mathcal{K}$.

Since $\xi(x, y)\in \mathbb H$ and bounded away from $\mathbb{R}$, the last factor in (\ref{eq: factoring_4_terms}) is uniformly bounded away from zero and infinity. Note that
$$
\left|\alpha \beta \left(\frac{\bar \xi (\xi^2-1)}{\xi(\bar \xi^2-1)}\right)^{\frac{\epsilon(j,k) +1}{2}}\frac{(\bar{\xi}^2-1)}{(\xi^2-1)} \frac{H_{\iota, \iota^{\prime}}(\xi)}{H_{\iota, \iota^{\prime}}(\bar{\xi})}\right|=1=\left|\alpha \bar{\beta} \left(\frac{ \xi (\bar\xi^2-1)}{\bar\xi(\xi^2-1)}\right)^{\frac{\epsilon(j,k) +1}{2}}\frac{\bar{\xi}}{\xi} \frac{G_{\iota, \iota^{\prime}}(\xi)}{G_{\iota, \iota^{\prime}}(\bar{\xi})}\right|,
$$
and
$$
\left|\frac{\xi-\mathrm{i}}{\bar{\xi}-\mathrm{i}}\right|<c<1
$$
for some constant $c$ that only depends on $\mathcal{K}$ and the function $f$. Hence, there exists a constant $\widetilde{C}>1$ so that for $(x, y) \in \mathcal{K}$ and fixed $f$,
$$
\frac{1}{\widetilde{C}} \leq|\al G(\xi)H(\xi)+\beta G(\bar \xi)H(\xi)-\bar\beta G(\xi)H(\bar \xi)-\bar \al G(\bar \xi)H(\bar \xi)| \leq \widetilde{C}
$$
which proves the statement.
\end{proof}
The following lemma will complete our rigidity confirmation.
\begin{lemma}
    \label{lemma: angle_bound}
    Let $\mathcal{K} \subset \mathcal{L}$ be a compact set. Then there exist constants $N_{\mathcal{K}}>0$ and $\varepsilon_{\mathcal{K}} \in(0, \pi)$, such that if $n>N_{\mathcal{K}}$, and $\frac{1}{n+1} e_{\iota_1, \iota_1^{\prime}}(j, k), \frac{1}{n+1} e_{\iota_2, \iota_2^{\prime}}(j, k)$ are contained in $\mathcal{K}$, the angle between two edges $d \widetilde{\mathcal{T}}_n\left(e_{\iota_1, \iota_1^{\prime}}(j, k)\right)$ and $d \widetilde{\mathcal{T}}_n\left(e_{\iota_2, \iota_2^{\prime}}(j, k)\right)$, which are adjacent to a common face, is contained in $\left(\varepsilon_{\mathcal{K}}, \pi-\varepsilon_{\mathcal{K}}\right)$.
\end{lemma}

\begin{proof}
In the case $|\io|+|\io'|=1$, it is enough to prove the statement for $j+k+n$ odd since a face of $\widetilde{\mathcal{T}}_n\left(\left(G'^f_{n+1}\right)^* \cap \mathcal{K}\right)$ is a convex quadrilateral with edge lengths bounded from above and below, by Lem. \ref{lemma: rigid_edge_bounds}.  If we know the statement is true for $j+k+n$ odd, then it means that two opposite angles in the quadrilateral are bounded from $\pi$ and 0. It follows that the other two angles in the quadrilateral are bounded away from 0. Recall the angle condition in the definition of perfect t-embeddings, Def. \ref{def:t_embedding}, which states that the sum of all angles around a vertex corresponding to the corner of a black (white) face is equal to $\pi$. In the case where the valency of a vertex is four, if all angles is bounded away from $0$, we conclude from the angle condition that the all angles is bounded away from $\pi$.

By the above discussion, we assume that $j+k+n$ is odd. We also assume that $(\iota_1, \iota_1^{\prime})$ follows by ($\iota_2, \iota_2^{\prime}$) when viewed as point on the unit circle oriented in positive direction. As we are first considering the case when $|\io|+|\io'|=1$ the vertex is adjacent to only 4 other edges.

    We fixed some new notation for the leading-order term $e^{\ci \theta}$ as $n \to \infty$, where $\theta$ is the angle between two edges $d \widetilde{\mathcal{T}}_n\left(e_{\iota_1, \iota_1^{\prime}}(j, k)\right)$ and $d \widetilde{\mathcal{T}}_n\left(e_{\iota_2, \iota_2^{\prime}}(j, k)\right)$. Set
    \begin{equation}
        \begin{aligned}
             A_i&:=G_{\io_i, \io'_i}(\xi)H_{\io_i, \io'_i}(\xi)\frac{(\xi -\ci)^2}{(\xi -1)(\xi +1)\xi}\\
            B_i&:=G_{\io_i, \io'_i}(\bar \xi)H_{\io_i, \io'_i}(\xi)\frac{(\xi -\ci)(\bar \xi -\ci)}{(\xi -1)(\xi +1)\bar \xi}\left(\frac{\bar \xi (\xi^2-1)}{\xi(\bar \xi^2-1)}\right)^{\frac{\epsilon(j,k) +1}{2}}\\
            \widetilde A_i&:=G_{\io_i, \io'_i}(\bar \xi)H_{\io_i, \io'_i}(\bar \xi)\frac{(\bar \xi -\ci)^2}{(\bar \xi -1)(\bar \xi +1)\bar \xi}\\
            \widetilde B_i&:=G_{\io_i, \io'_i}(\xi)H_{\io_i, \io'_i}(\bar \xi)\frac{(\bar \xi -\ci)(\xi -\ci)}{(\bar \xi -1)(\bar \xi +1) \xi}\left(\frac{ \xi (\bar\xi^2-1)}{\bar\xi(\xi^2-1)}\right)^{\frac{\epsilon(j,k) +1}{2}}
        \end{aligned}
    \end{equation}
    Then, similar to the proof of Lemma 6.6 in \cite{tower_AD_perfect}, we want to consider the imaginary part of the quantity 
    \begin{equation}\label{eq: angle_leading_term}
        \begin{aligned}
            \frac{\overline{\left(A_1 \alpha+B_1 \beta-\widetilde{A}_1 \bar{\alpha}-\widetilde{B}_1 \bar{\beta}\right)}\left(A_2 \alpha+B_2 \beta-\widetilde{A}_2 \bar{\alpha}-\widetilde{B}_2 \bar{\beta}\right)}{\left|A_1 \alpha+B_1 \beta-\widetilde{A}_1 \bar{\alpha}-\widetilde{B}_1 \bar{\beta}\right|\left|A_2 \alpha+B_2 \beta-\widetilde{A}_2 \bar{\alpha}-\widetilde{B}_2 \bar{\beta}\right|}
        \end{aligned}
    \end{equation}
   since this is the leading term of $e^{\ci\theta}$. For consecutive points, we see from the Cor. \ref{eq:dT equation factor} that either,
$$
\frac{H_{\iota_1, \iota_1^{\prime}}(\xi)}{H_{\iota_1, \iota_1^{\prime}}(\bar{\xi})}=\frac{H_{\iota_2, \iota_2^{\prime}}(\xi)}{H_{\iota_2, \iota_2^{\prime}}(\bar{\xi})} \quad \text { or } \quad \frac{G_{\iota_1, \iota_1^{\prime}}(\xi)}{G_{\iota_1, \iota_1^{\prime}}(\bar{\xi})}=\frac{G_{\iota_2, \iota_2^{\prime}}(\xi)}{G_{\iota_2, \iota_2^{\prime}}(\bar{\xi})}.
$$
In the case $\displaystyle \frac{H_{\iota_1, \iota_1^{\prime}}(\xi)}{H_{\iota_1, \iota_1^{\prime}}(\bar{\xi})}=\frac{H_{\iota_2, \iota_2^{\prime}}(\xi)}{H_{\iota_2, \iota_2^{\prime}}(\bar{\xi})}$, we have that 
    $$\begin{aligned}
&H_{\iota_1, \iota_1^{\prime}}(z)=H_{\iota_2, \iota_2^{\prime}}(z)=\frac{1}{z}, \quad G_{\iota_1, \iota_1^{\prime}}(w)=\frac{1}{w-1} \quad \text { and } \quad G_{\iota_2, \iota_2^{\prime}}(w)=-\frac{1}{w+1}
\end{aligned}$$
or
$$\begin{aligned}
&H_{\iota_1, \iota_1^{\prime}}(z)=H_{\iota_2, \iota_2^{\prime}}(z)=1, \quad G_{\iota_1, \iota_1^{\prime}}(w)=\frac{1}{w+1} \quad \text { and } \quad G_{\iota_2, \iota_2^{\prime}}(w)=\frac{1}{w-1}.
\end{aligned}$$
From (\ref{eq: factoring_4_terms}), let us compute the product in the numerator of (\ref{eq: angle_leading_term}) given by
\begin{equation}
\label{eq: angle_leading_term_explicit}
\resizebox{0.9\textwidth}{!}{$
    \begin{aligned}
        \overline{\left(\al\beta\left(\frac{\bar \xi (\xi^2-1)}{\xi(\bar \xi^2-1)}\right)^{\frac{\epsilon(j,k) +1}{2}}\frac{(\xi-\ci)(\bar \xi^2 -1)H_{\io_1,\io'_1}(\xi)}{(\bar\xi-\ci)(\xi^2 -1)H_{\io_1,\io_1'}(\bar \xi)}-1\right)\left(\al\bar\beta\frac{\xi-\ci}{\bar \xi-\ci}\frac{G_{\io_1,\io'_1}(\xi)}{G_{\io_1,\io'_1}(\bar \xi)}\frac{\bar \xi}{\xi}\left(\frac{ \xi (\bar\xi^2-1)}{\bar\xi(\xi^2-1)}\right)^{\frac{\epsilon(j,k) +1}{2}}+1\right)
  \frac{(\bar \xi-\ci)^2H_{\io_1,\io'_1}(\bar \xi)G_{\io_1,\io'_1}(\bar \xi)\bar \al}{(\bar \xi^2-1)\bar \xi}}\\
  \times \left(\al\beta\left(\frac{\bar \xi (\xi^2-1)}{\xi(\bar \xi^2-1)}\right)^{\frac{\epsilon(j,k) +1}{2}}\frac{(\xi-\ci)(\bar \xi^2 -1)H_{\io_2,\io'_2}(\xi)}{(\bar\xi-\ci)(\xi^2 -1)H_{\io_2,\io'_2}(\bar \xi)}-1\right)\left(\al\bar\beta\frac{\xi-\ci}{\bar \xi-\ci}\frac{G_{\io_2,\io'_2}(\xi)}{G_{\io_2,\io'_2}(\bar \xi)}\frac{\bar \xi}{\xi}\left(\frac{ \xi (\bar\xi^2-1)}{\bar\xi(\xi^2-1)}\right)^{\frac{\epsilon(j,k) +1}{2}}+1\right)
  \frac{(\bar \xi-\ci)^2H_{\io_2,\io'_2}(\bar \xi)G_{\io_2,\io'_2}(\bar \xi)\bar \al}{(\bar \xi^2-1)\bar \xi}.
\end{aligned}
$}
\end{equation}
We have, since $|\al|=1=|\beta|$, $H_{\io_1,\io'_1}=H_{\io_2,\io'_2}$, 
\begin{equation}
\label{eq: last_factor}
\resizebox{0.9\textwidth}{!}{$
\begin{aligned}
    &\overline{\frac{(\bar \xi-\ci)^2H_{\io_1,\io'_1}(\bar \xi)G_{\io_1,\io'_1}(\bar \xi)\bar \al}{(\bar \xi^2-1)\bar \xi}}\times\frac{(\bar \xi-\ci)^2H_{\io_2,\io'_2}(\bar \xi)G_{\io_2,\io'_2}(\bar \xi)\bar \al}{(\bar \xi^2-1)\bar \xi}=\frac{|\bar\xi-\ci|^4|H_{\io_1,\io'_1}(\xi)|^2}{|\bar \xi -1|^2|\bar \xi +1|^2|\xi|^2}G_{\io_1,\io'_1}(\xi)G_{\io_2,\io'_2}(\bar \xi).
\end{aligned}
$}
\end{equation}
On the other hand, the product of the two first terms reads
\begin{equation}
\label{eq: first_factor}
\resizebox{0.9\textwidth}{!}{$
    \begin{aligned}
         &\overline{\left(\al\beta\left(\frac{\bar \xi (\xi^2-1)}{\xi(\bar \xi^2-1)}\right)^{\frac{\epsilon(j,k) +1}{2}}\frac{(\xi-\ci)(\bar \xi^2 -1)H_{\io_1,\io'_1}(\xi)}{(\bar\xi-\ci)(\xi^2 -1)H_{\io_1,\io_1'}(\bar \xi)}-1\right)}\times\left(\al\beta\left(\frac{\bar \xi (\xi^2-1)}{\xi(\bar \xi^2-1)}\right)^{\frac{\epsilon(j,k) +1}{2}}\frac{(\xi-\ci)(\bar \xi^2 -1)H_{\io_2,\io'_2}(\xi)}{(\bar\xi-\ci)(\xi^2 -1)H_{\io_2,\io'_2}(\bar \xi)}-1\right)\\
         &=\left|\left(\frac{\bar \xi (\xi^2-1)}{\xi(\bar \xi^2-1)}\right)^{\frac{\epsilon(j,k) +1}{2}}\frac{\xi-\ci}{\bar\xi-\ci}\right|^2-\bar\al\bar\beta\frac{(\xi-1)(\xi+1)(\bar\xi+i)}{(\bar\xi-1)(\bar\xi+1)(\xi+\ci)}\left(\frac{\xi (\bar\xi-1)(\bar\xi+1)}{ \bar\xi (\xi-1)(\xi+1)}\right)^{\frac{\epsilon(j,k) +1}{2}}\frac{H_{\io_1,\io'_1}(\bar\xi)}{H_{\io_1,\io'_1}(\xi)}\\
         &-\alpha\beta\frac{(\bar\xi-1)(\bar\xi+1)(\xi+i)}{(\xi-1)(\xi+1)(\bar\xi+\ci)}\left(\frac{\bar \xi (\xi-1)(\xi+1)}{ \xi (\bar\xi-1)(\bar \xi+1)}\right)^{\frac{\epsilon(j,k) +1}{2}}\frac{H_{\io_1,\io'_1}(\xi)}{H_{\io_1,\io'_1}(\bar \xi)}+1.
    \end{aligned}
$}
\end{equation}

Note that (\ref{eq: first_factor}) is a real number since the two middle terms are conjugate pairs. Combining \eqref{eq: last_factor} and \eqref{eq: first_factor}, we have that the imaginary part of \eqref{eq: angle_leading_term_explicit} is determined by

\begin{equation}\label{eq: angle_leading_term_simplify}
\resizebox{0.9\textwidth}{!}{$
        \begin{aligned}
           & \operatorname{Im}\left(\frac{\overline{\left(A_1 \alpha+B_1 \beta-\widetilde{A}_1 \bar{\alpha}-\widetilde{B}_1 \bar{\beta}\right)}\left(A_2 \alpha+B_2 \beta-\widetilde{A}_2 \bar{\alpha}-\widetilde{B}_2 \bar{\beta}\right)}{\left|A_1 \alpha+B_1 \beta-\widetilde{A}_1 \bar{\alpha}-\widetilde{B}_1 \bar{\beta}\right|\left|A_2 \alpha+B_2 \beta-\widetilde{A}_2 \bar{\alpha}-\widetilde{B}_2 \bar{\beta}\right|}\right)\\
           &=\frac{\operatorname{Im}\left\{G_{\io_1,\io'_1}(\xi)G_{\io_2,\io'_2}(\bar \xi)\overline{\left(\al\bar\beta\frac{\xi-\ci}{\bar \xi-\ci}\frac{G_{\io_1,\io'_1}(\xi)}{G_{\io_1,\io'_1}(\bar \xi)}\frac{\bar \xi}{\xi}\left(\frac{ \xi (\bar\xi^2-1)}{\bar\xi(\xi^2-1)}\right)^{\frac{\epsilon(j,k) +1}{2}}+1\right)}\left(\al\bar\beta\frac{\xi-\ci}{\bar \xi-\ci}\frac{G_{\io_2,\io'_2}(\xi)}{G_{\io_2,\io'_2}(\bar \xi)}\frac{\bar \xi}{\xi}\left(\frac{ \xi (\bar\xi^2-1)}{\bar\xi(\xi^2-1)}\right)^{\frac{\epsilon(j,k) +1}{2}}+1\right)\right\}}{\overline{\left|\left(\al\bar\beta\frac{\xi-\ci}{\bar \xi-\ci}\frac{G_{\io_1,\io'_1}(\xi)}{G_{\io_1,\io'_1}(\bar \xi)}\frac{\bar \xi}{\xi}\left(\frac{ \xi (\bar\xi^2-1)}{\bar\xi(\xi^2-1)}\right)^{\frac{\epsilon(j,k) +1}{2}}+1\right)\right |}\left|\left(\al\bar\beta\frac{\xi-\ci}{\bar \xi-\ci}\frac{G_{\io_2,\io'_2}(\xi)}{G_{\io_2,\io'_2}(\bar \xi)}\frac{\bar \xi}{\xi}\left(\frac{ \xi (\bar\xi^2-1)}{\bar\xi(\xi^2-1)}\right)^{\frac{\epsilon(j,k) +1}{2}}+1\right)\right |}\\
           &=\operatorname{Im}\left\{\frac{\left(\bar\al\beta\frac{\bar\xi+\ci}{ \xi+\ci}G_{\io_1,\io'_1}(\bar \xi)\frac{\xi}{\bar \xi}\left(\frac{\bar \xi (\xi^2-1)}{\xi(\bar \xi^2-1)}\right)^{\frac{\epsilon(j,k) +1}{2}}+G_{\io_1,\io'_1}(\xi)\right)}{\left|\left(\bar\al\beta\frac{\bar\xi+\ci}{ \xi+\ci}G_{\io_1,\io'_1}(\bar \xi)\frac{\xi}{\bar \xi}\left(\frac{\bar \xi (\xi^2-1)}{\xi(\bar \xi^2-1)}\right)^{\frac{\epsilon(j,k) +1}{2}}+G_{\io_1,\io'_1}(\xi)\right)\right|}\times \frac{\left(\al\bar\beta\frac{\xi-\ci}{\bar \xi-\ci}G_{\io_2,\io'_2}(\xi)\frac{\bar \xi}{\xi}\left(\frac{ \xi (\bar\xi^2-1)}{\bar\xi(\xi^2-1)}\right)^{\frac{\epsilon(j,k) +1}{2}}+G_{\io_2,\io'_2}(\bar \xi)\right)}{\left|\left(\al\bar\beta\frac{\xi-\ci}{\bar \xi-\ci}G_{\io_2,\io'_2}(\xi)\frac{\bar \xi}{\xi}\left(\frac{ \xi (\bar\xi^2-1)}{\bar\xi(\xi^2-1)}\right)^{\frac{\epsilon(j,k) +1}{2}}+G_{\io_2,\io'_2}(\bar \xi)\right)\right |}\right\}.
        \end{aligned}
    $}
    \end{equation}
    
We again consider the product in the numerator of the last equality and notice that expanding this product gives two crossing terms that are each other conjugates. Then, the imaginary part of \eqref{eq: angle_leading_term_simplify} is 
\begin{equation}
    \label{eq:Imaginary_part}
    \resizebox{0.9\textwidth}{!}{$
    \begin{aligned}   &\frac{\operatorname{Im}\left(|\frac{\xi-\ci}{\bar \xi-\ci}|^2G_{\io_1,\io'_1}(\bar \xi) G_{\io_2, \io'_2}(\xi)+G_{\io_1,\io'_1}(\xi)G_{\io_2,\io'_2}(\bar \xi)\right)}{\left|\left(\bar\al\beta\frac{\bar\xi+\ci}{ \xi+\ci}G_{\io_1,\io'_1}(\bar \xi)\frac{\xi}{\bar \xi}\left(\frac{\bar \xi (\xi^2-1)}{\xi(\bar \xi^2-1)}\right)^{\frac{\epsilon(j,k) +1}{2}}+G_{\io_1,\io'_1}(\xi)\right)\right|\times \left|\left(\al\bar\beta\frac{\xi-\ci}{\bar \xi-\ci}G_{\io_2,\io'_2}(\xi)\frac{\bar \xi}{\xi}\left(\frac{ \xi (\bar\xi^2-1)}{\bar\xi(\xi^2-1)}\right)^{\frac{\epsilon(j,k) +1}{2}}+G_{\io_2,\io'_2}(\bar \xi)\right)\right |}\\    
    &=\frac{\left(1-|\frac{\xi-\ci}{\bar\xi-\ci}|^2\right)}{\left|1+\bar \al\beta\frac{\bar\xi+\ci}{ \xi+\ci}\frac{G_{\io_1,\io'_1}(\bar \xi)}{G_{\io_1,\io'_1}( \xi)}\frac{\xi}{\bar \xi}\left(\frac{\bar \xi (\xi^2-1)}{\xi(\bar \xi^2-1)}\right)^{\frac{\epsilon(j,k) +1}{2}}\right|\left |1+\al\bar\beta\frac{\xi-\ci}{\bar \xi-\ci}\frac{G_{\io_2,\io'_2}(\xi)}{G_{\io_2,\io'_2}(\bar \xi)}\frac{\bar \xi}{\xi}\left(\frac{ \xi (\bar\xi^2-1)}{\bar\xi(\xi^2-1)}\right)^{\frac{\epsilon(j,k) +1}{2}}\right |}\times \frac{\operatorname{Im}(G_{\io_1,\io'_1}(\xi)G_{\io_2, \io'_2}(\bar \xi))}{|G_{\io_1,\io'_1}(\xi)||G_{\io_2,\io'_2}(\bar \xi)|}.
    \end{aligned}
    $}
\end{equation}
By the same argument as in the proof of Lem. \ref{lemma: rigid_edge_bounds}, since 
$$
\left|\frac{\xi-\mathrm{i}}{\bar{\xi}-\mathrm{i}}\right|<c<1
$$
the first factor in \eqref{eq:Imaginary_part} is bounded away from zero and infinity. On the other hand, 
$$\begin{aligned}
    \frac{\operatorname{Im}(G_{\io_1,\io'_1}(\xi)G_{\io_2, \io'_2}(\bar \xi))}{|G_{\io_1,\io'_1}(\xi)||G_{\io_2,\io'_2}(\bar \xi)|}=\frac{2\operatorname{Im}(\xi)}{|\xi-1||\xi+1|}>C>0
\end{aligned}$$
for some constant $C$ which finishes the proof for the case $\displaystyle\frac{H_{\iota_1, \iota_1^{\prime}}(\xi)}{H_{\iota_1, \iota_1^{\prime}}(\bar{\xi})}=\frac{H_{\iota_2, \iota_2^{\prime}}(\xi)}{H_{\iota_2, \iota_2^{\prime}}(\bar{\xi})}$. The case $\displaystyle\frac{G_{\iota_1, \iota_1^{\prime}}(\xi)}{G_{\iota_1, \iota_1^{\prime}}(\bar{\xi})}=\frac{G_{\iota_2, \iota_2^{\prime}}(\xi)}{G_{\iota_2, \iota_2^{\prime}}(\bar{\xi})}$ proceeds similarly.

Now we consider the case when $\io+\io'=0$. This corresponds to a vertex of degree six. 
In this case, we will show that for any vertex, any pair of angles belongs to two faces of the same color is bounded away from zero. Then, by the same reason as in the case $|\io|+|\io'|=1$, the condition that the angles corresponding to one color of face around the vertex sum up to $\pi$ implies that the last angles are bounded away from $\pi$. 

\[
H_{0,1}(z)=\frac{1}{z} \quad H_{-1,1}(z)=\frac{1}{z-1} \quad G_{0,1}(w)=\frac{-1}{w+1} \quad G_{-1,1}=\frac{-2}{w+1}.
\]
We want to analyze the imaginary part of (\ref{eq: angle_leading_term}) which when written explicitly still takes the same form as (\ref{eq: angle_leading_term_explicit}). We see that
\begin{equation}
\label{eq: last_factor_hex_case}
\resizebox{0.9\textwidth}{!}{$
\begin{aligned}
    &\overline{\frac{(\bar \xi-\ci)^2H_{\io_1,\io'_1}(\bar \xi)G_{\io_1,\io'_1}(\bar \xi)\bar \al}{(\bar \xi^2-1)\bar \xi}}\times\frac{(\bar \xi-\ci)^2H_{\io_2,\io'_2}(\bar \xi)G_{\io_2,\io'_2}(\bar \xi)\bar \al}{(\bar \xi^2-1)\bar \xi}=\frac{|\bar\xi-\ci|^4}{|\bar \xi -1|^2|\bar \xi +1|^2|\xi|^2}{H_{\io_1,\io'_1}( \xi)}H_{\io_2,\io'_2}(\bar \xi){G_{\io_1,\io'_1}}(\xi)G_{\io_2,\io'_2}(\bar \xi).
\end{aligned}
$}
\end{equation}
The product of two first terms in this case is again a real number. And so we have the imaginary part of \eqref{eq: angle_leading_term_explicit} is determined by
{\footnotesize
\begin{equation}
    \label{eq: angle_leading_term_simplify_hex_case}
    \begin{aligned}
       \frac{\operatorname{Im}\left\{{H_{\io_1,\io'_1}( \xi)}H_{\io_2,\io'_2}(\bar \xi)\left(|\frac{\xi-\ci}{\bar \xi-\ci}|^2G_{\io_1,\io'_1}(\bar \xi) G_{\io_2, \io'_2}(\xi)+G_{\io_1,\io'_1}(\xi)G_{\io_2,\io'_2}(\bar \xi)\right)\right\}}{\left|\left(\bar\al\beta\frac{\bar\xi+\ci}{ \xi+\ci}G_{\io_1,\io'_1}(\bar \xi)\frac{\xi}{\bar \xi}\left(\frac{\bar \xi (\xi^2-1)}{\xi(\bar \xi^2-1)}\right)^{\frac{\epsilon(j,k) +1}{2}}+G_{\io_1,\io'_1}(\xi)\right)\right|\times \left|\left(\al\bar\beta\frac{\xi-\ci}{\bar \xi-\ci}G_{\io_2,\io'_2}(\xi)\frac{\bar \xi}{\xi}\left(\frac{ \xi (\bar\xi^2-1)}{\bar\xi(\xi^2-1)}\right)^{\frac{\epsilon(j,k) +1}{2}}+G_{\io_2,\io'_2}(\bar \xi)\right)\right |}
    \end{aligned}
\end{equation}
}
Note that for \eqref{eq: angle_leading_term_simplify_hex_case}, in all cases, either one of the factors $G_{\io_1,\io'_1}(\xi)G_{\io_2,\io'_2}(\bar \xi)$ or ${H_{\io_1,\io'_1}(\xi)}H_{\io_2,\io'_2}(\bar \xi)$ will be a real number and the other factor will have non-trivial imaginary part bounded away from zero. This finished the proof for the case when $\io+\io'=0$
\end{proof}

\subsection{Simulations}
We end with some examples of large rank. Figure \ref{fig:largeAD} shows a large tiling of the Aztec diamond with the theoretical arctic curve overlaid, as well as the graph $(\mathcal{T}_n,\mathcal{O}_n')$ approximating the limiting maximal surface. Figures \ref{fig:largeTower}, \ref{fig:largeQuadratic}, and \ref{fig:largeSin} show the same for the tower graph, and the generalized tower graphs approximating $f(x) = \frac{1}{3}x^2 + \frac{2}{3}$ and $f(x) = \frac{1}{3\pi}\sin(2\pi x) +1$, respectively. Note that for the tiling we have rotated the image 45 degrees counter-clockwise so that the coordinate axis are aligned horizontally and vertically. The colors in the tiling are used to distinguish the different orientations of tiles. In particular, the corresponding dimer model (without rotating) for any of these graphs contains six possible types of edges:
\[
\begin{tabular}{cccccc}
     \begin{tikzpicture}[baseline = (current bounding box).center]
         \draw (0,0)--(1,1);
         \draw[fill=white] (0,0) circle (3pt);
         \draw[fill=black] (1,1) circle (3pt);
     \end{tikzpicture},
     &
     \begin{tikzpicture}[baseline = (current bounding box).center]
         \draw (0,0)--(1,1);
         \draw[fill=black] (0,0) circle (3pt);
         \draw[fill=white] (1,1) circle (3pt);
     \end{tikzpicture},
     &
     \begin{tikzpicture}[baseline = (current bounding box).center]
         \draw (0,0)--(1,-1);
         \draw[fill=white] (0,0) circle (3pt);
         \draw[fill=black] (1,-1) circle (3pt);
     \end{tikzpicture},
     & 
     \begin{tikzpicture}[baseline = (current bounding box).center]
         \draw (0,0)--(1,-1);
         \draw[fill=black] (0,0) circle (3pt);
         \draw[fill=white] (1,-1) circle (3pt);
     \end{tikzpicture},
     &
     \begin{tikzpicture}[baseline = (current bounding box).center]
         \draw (0,0)--(1,0);
         \draw[fill=white] (0,0) circle (3pt);
         \draw[fill=black] (1,0) circle (3pt);
     \end{tikzpicture},
     &
     and
     \begin{tikzpicture}[baseline = (current bounding box).center]
         \draw (0,0)--(1,0);
         \draw[fill=black] (0,0) circle (3pt);
         \draw[fill=white] (1,0) circle (3pt);
     \end{tikzpicture}
\end{tabular}.
\]
We color the corresponding tiles yellow, blue, red, green, orange, and teal, respectively. 

\begin{figure}
    \centering
    \[
    \includegraphics[width=0.3\linewidth]{AD20.pdf}
    \quad
    \includegraphics[width=0.3\linewidth]{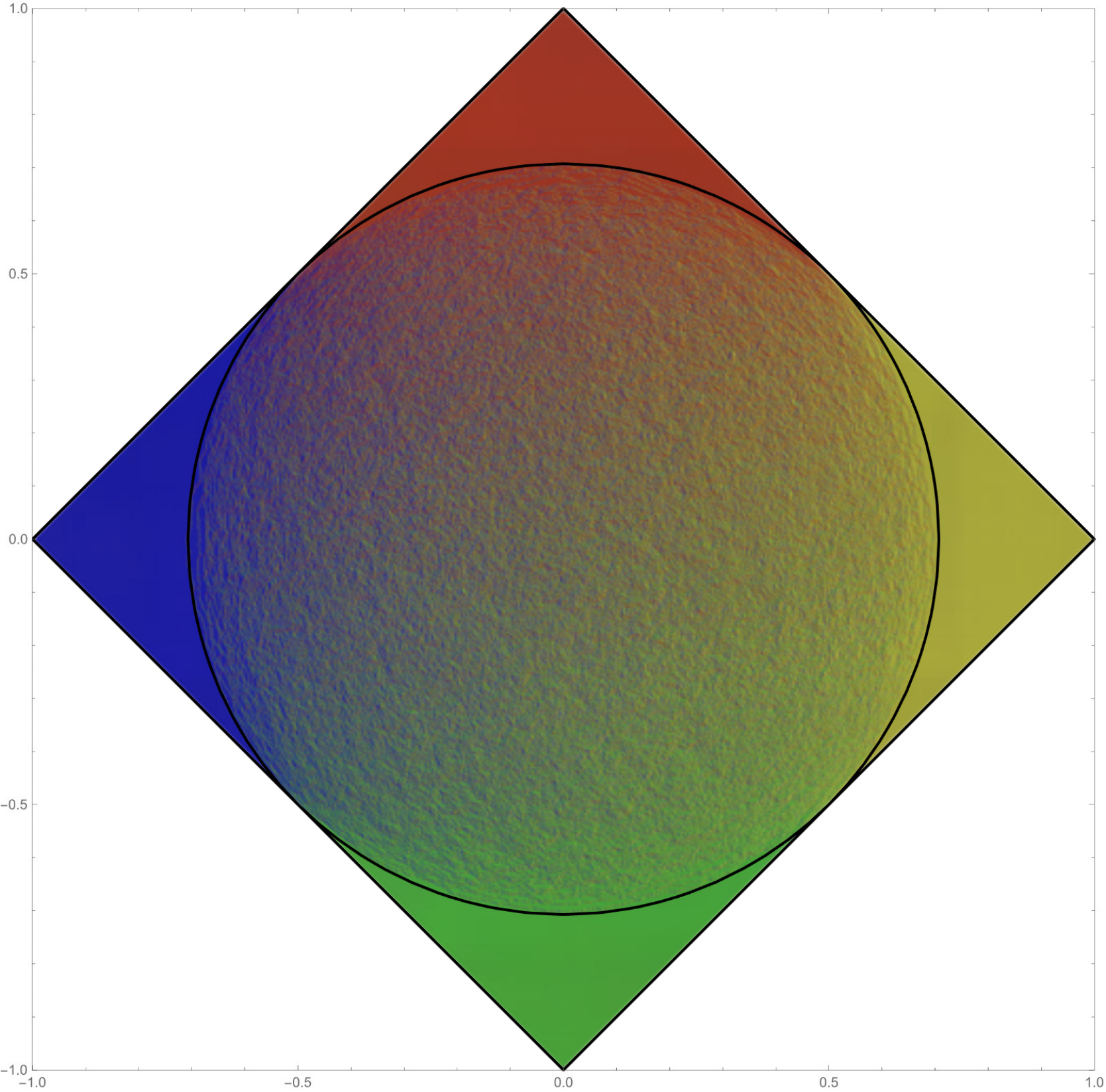}
    \quad
    \includegraphics[width=0.3\linewidth]{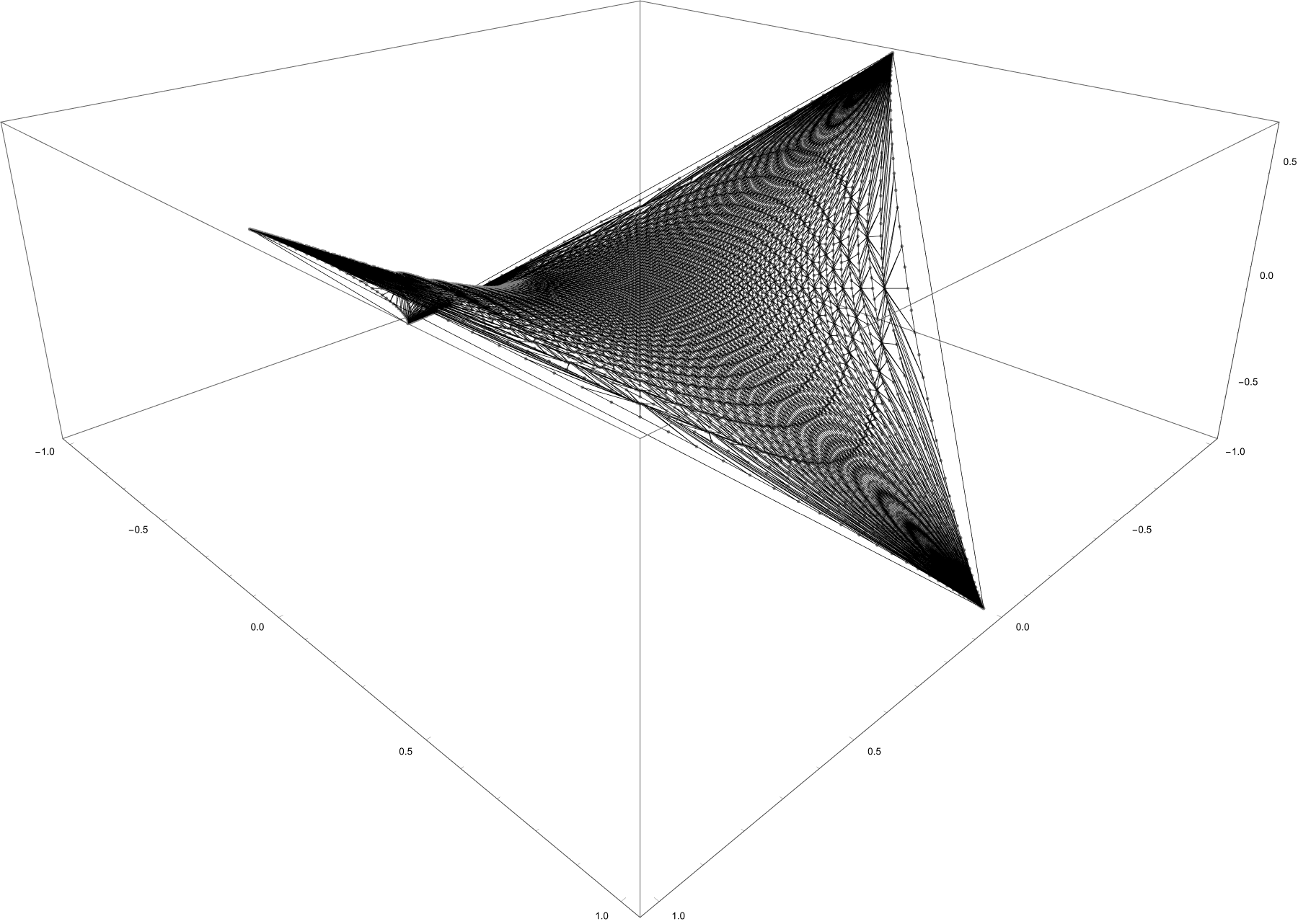}
    \]
    \caption{Left: Domino tiling of rank $n=20$. Center: Domino tiling of the Aztec diamond of rank $n=1000$. The inner black curve is the theoretical arctic curve $x^2+y^2=\frac{1}{2}$. Right: Plot of the graph $(\mathcal{T}_n,\mathcal{O}_n')$ for $n=100$.}
    \label{fig:largeAD}
\end{figure}

\begin{figure}
    \centering
    \[
    \includegraphics[width=0.3\linewidth]{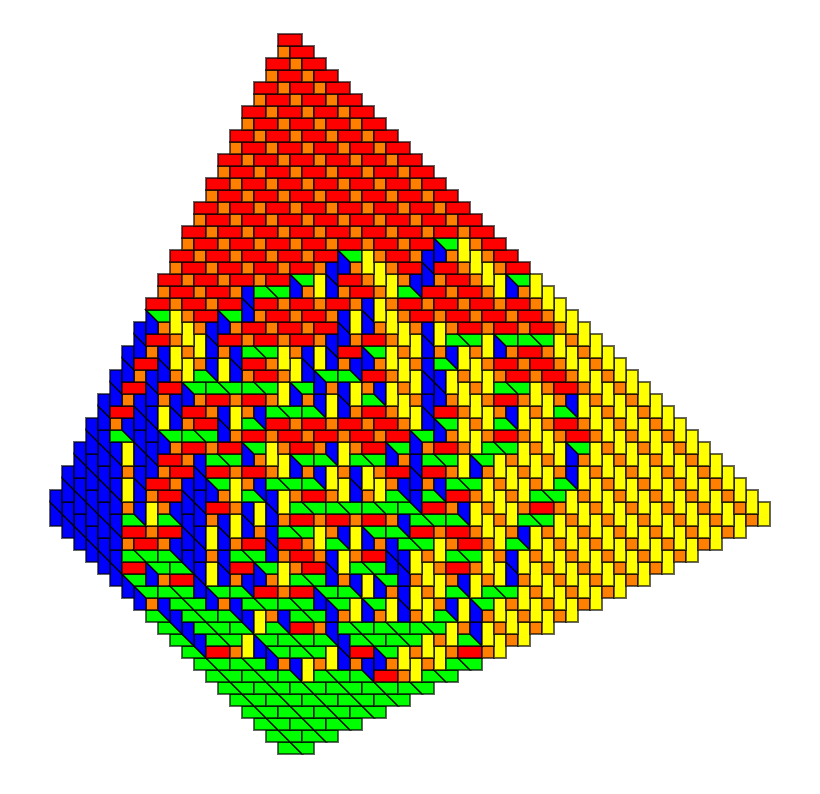}
    \quad
    \includegraphics[width=0.3\linewidth]{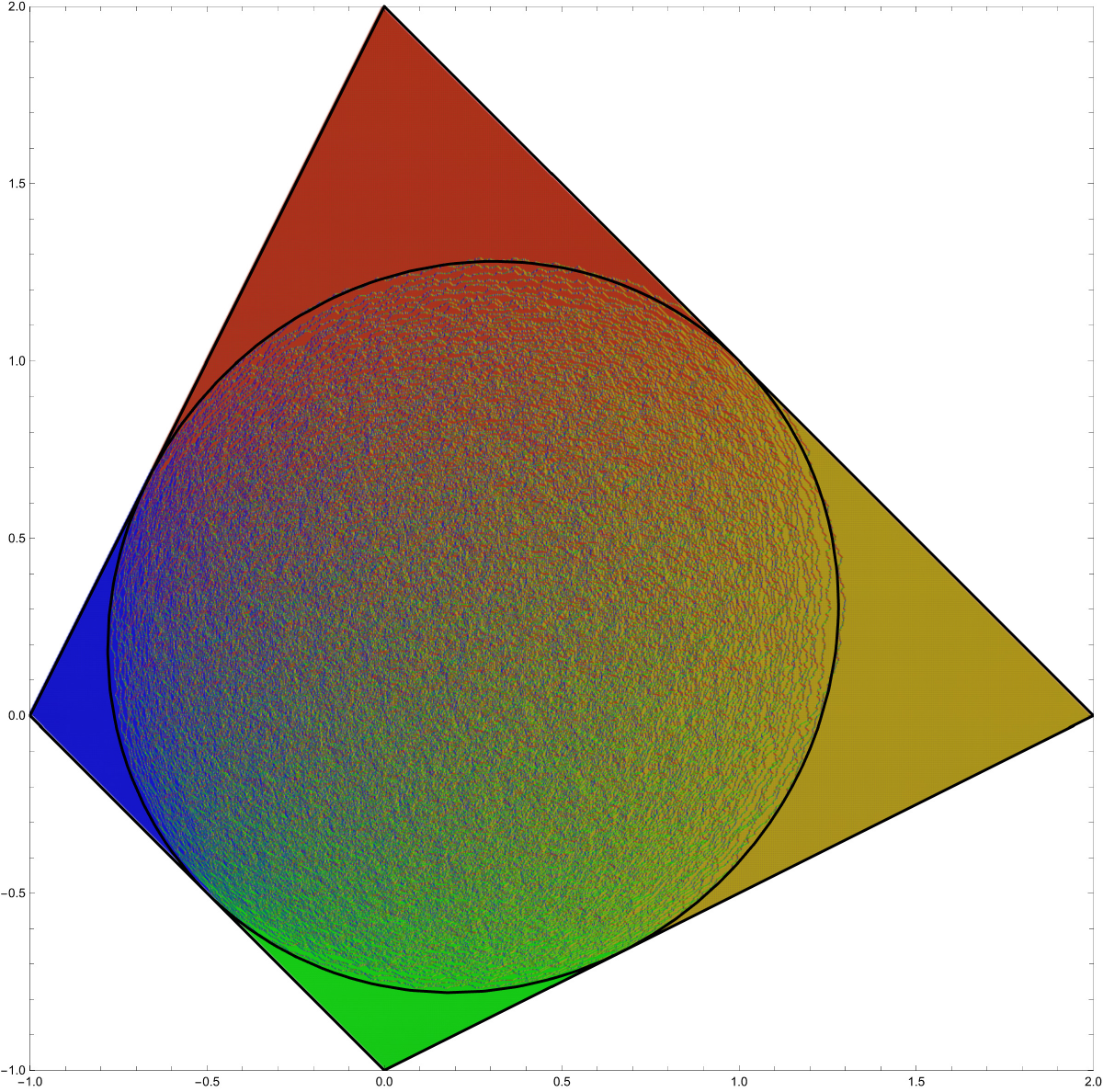}
    \quad
    \includegraphics[width=0.3\linewidth]{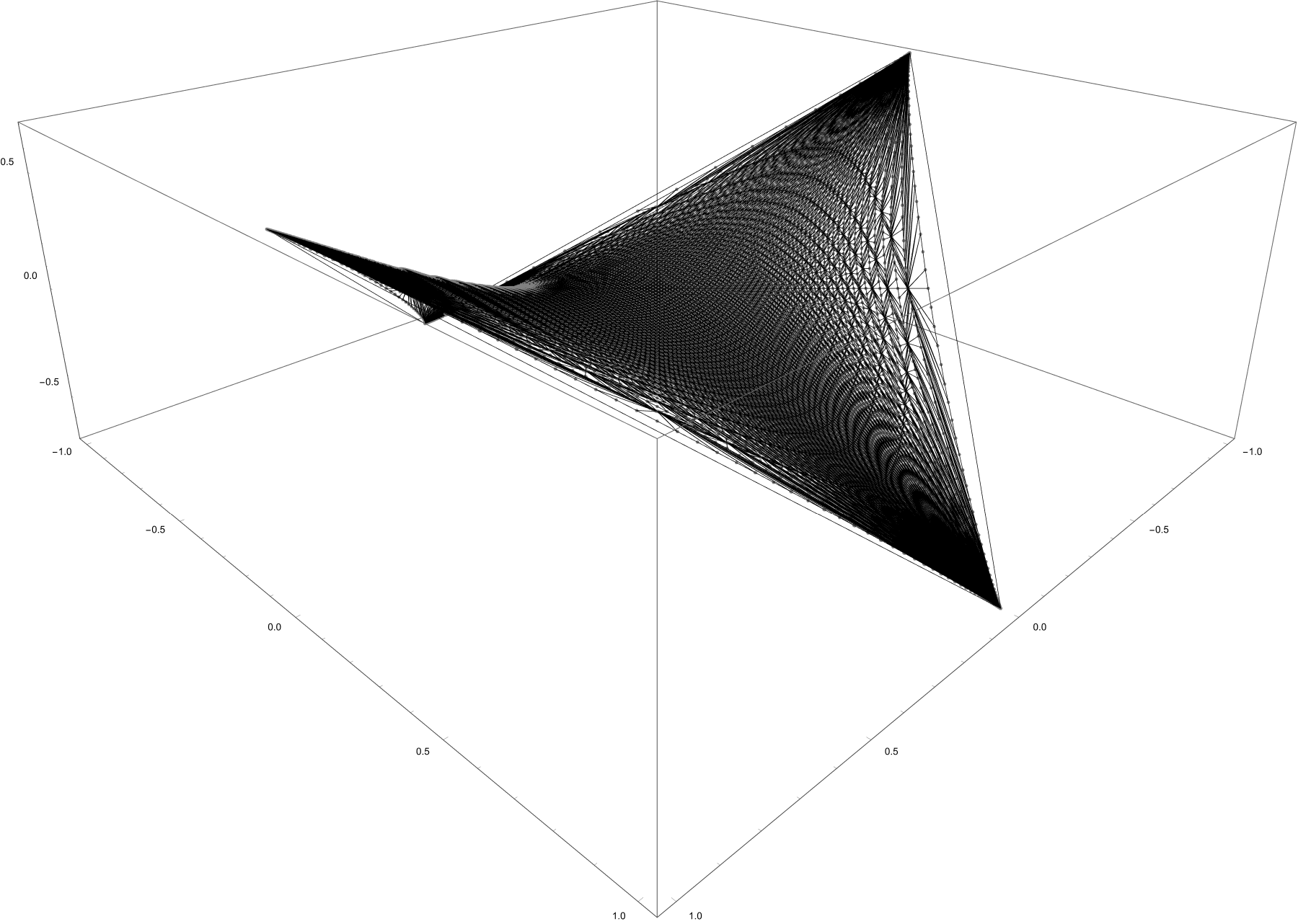}
    \]
    \caption{Left: A tiling of the tower graph of rank $n=20$.  Center: A tiling of the tower graph of rank $n=600$. The inner black curve is the theoretical arctic curve $x^2+y^2=\frac{1}{2}\left(\frac{1}{3}(x+y)+\frac{4}{3}\right)^2$. Right: Plot of the graph $(\mathcal{\widetilde T}_n,\mathcal{\widetilde O}_n')$ for $n=100$.}
    \label{fig:largeTower}
\end{figure}

\newpage

\begin{figure}
    \centering
    \[
    \includegraphics[width=0.3\linewidth]{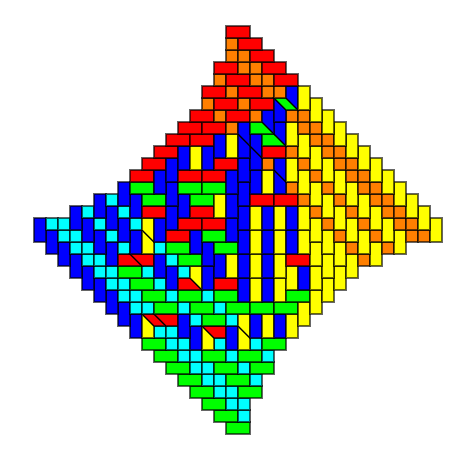}
    \quad
    \includegraphics[width=0.3\linewidth]{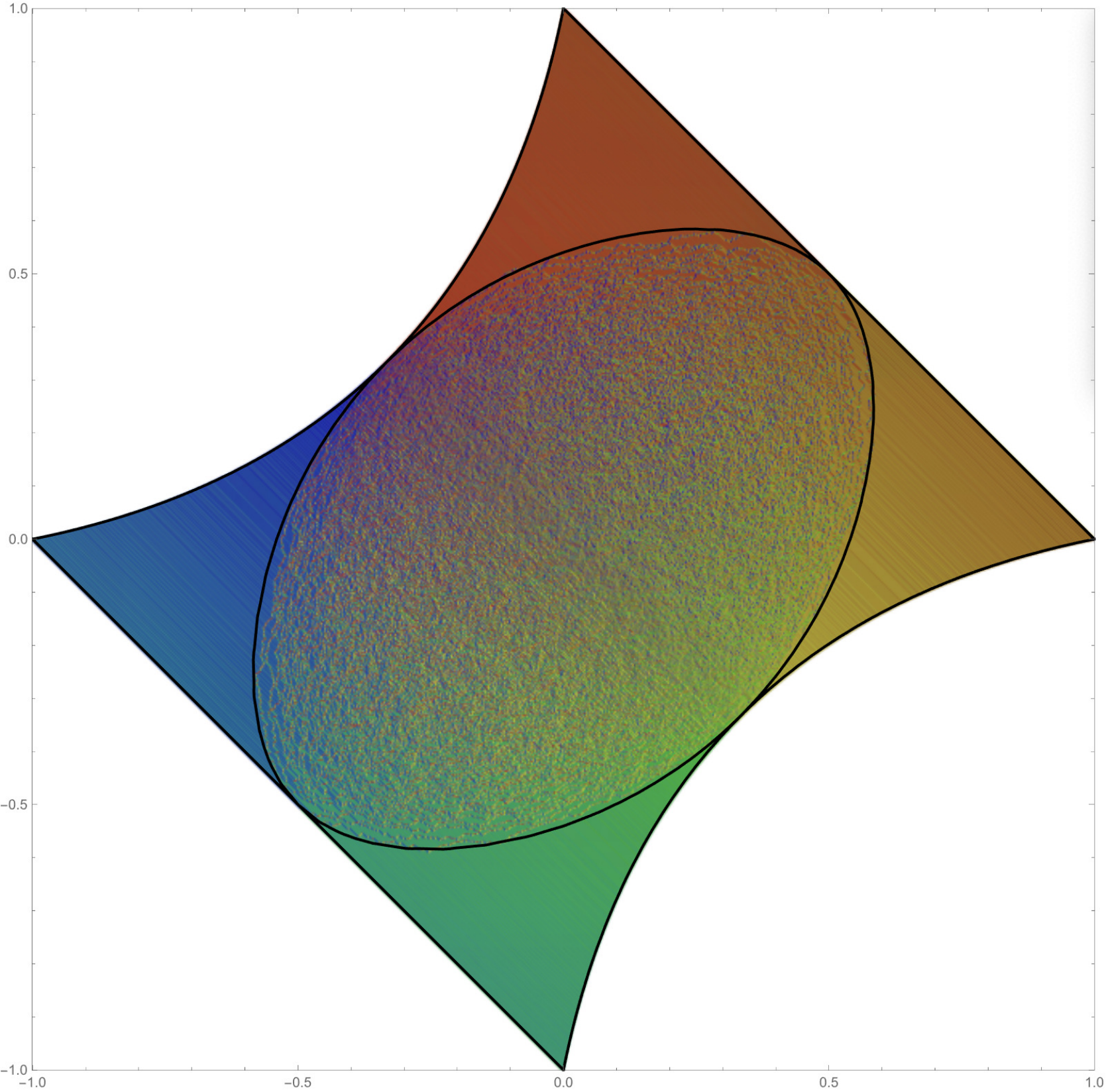}
    \quad
    \includegraphics[width=0.3\linewidth]{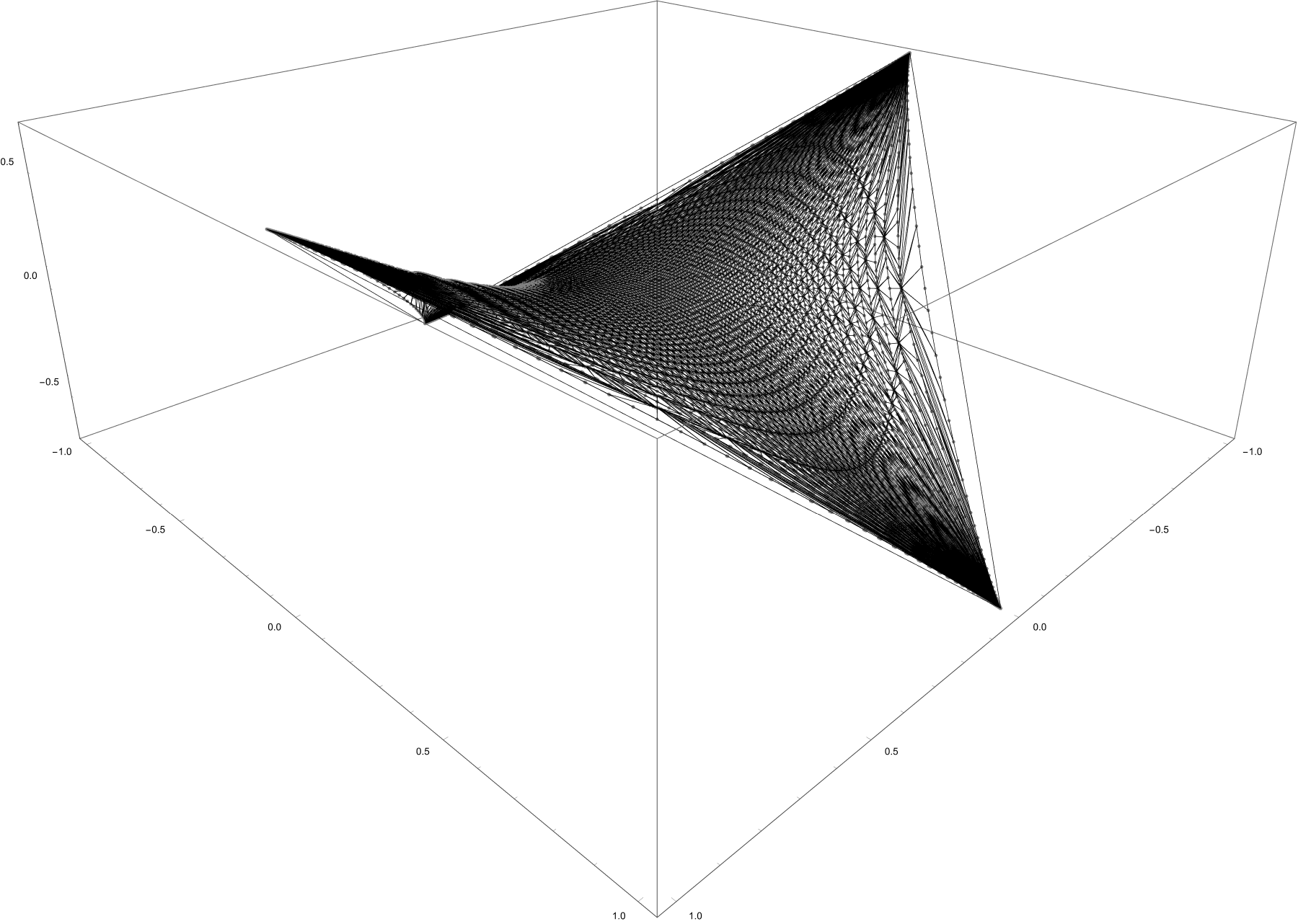}
    \]
    \caption{Left: A tiling of the generalized tower graph of rank $n=20$ approximating $f(x) = \frac{1}{3}x^2 + \frac{2}{3}$. Center: A tiling of the same generalized tower graph of rank $n=1000$. The inner black curve is the theoretical arctic curve $x^2+y^2=\frac{1}{2}f(x+y)^2$. Right:  Plot of the graph $(\mathcal{\widetilde T}_n,\mathcal{\widetilde O}_n')$ for $n=150$.}
    \label{fig:largeQuadratic}
\end{figure}

\begin{figure}
    \centering
    \[
    \includegraphics[width=0.3\linewidth]{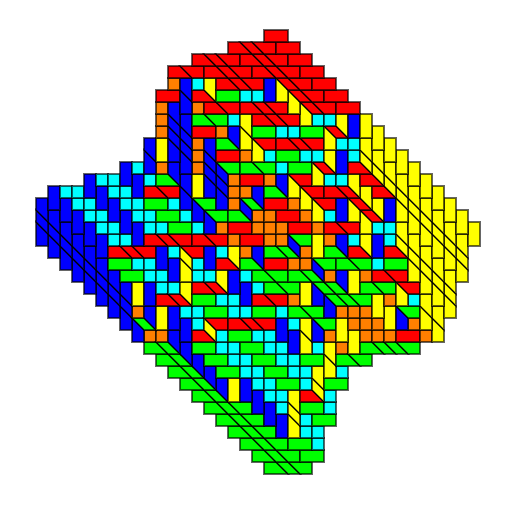}
    \quad
    \includegraphics[width=0.3\linewidth]{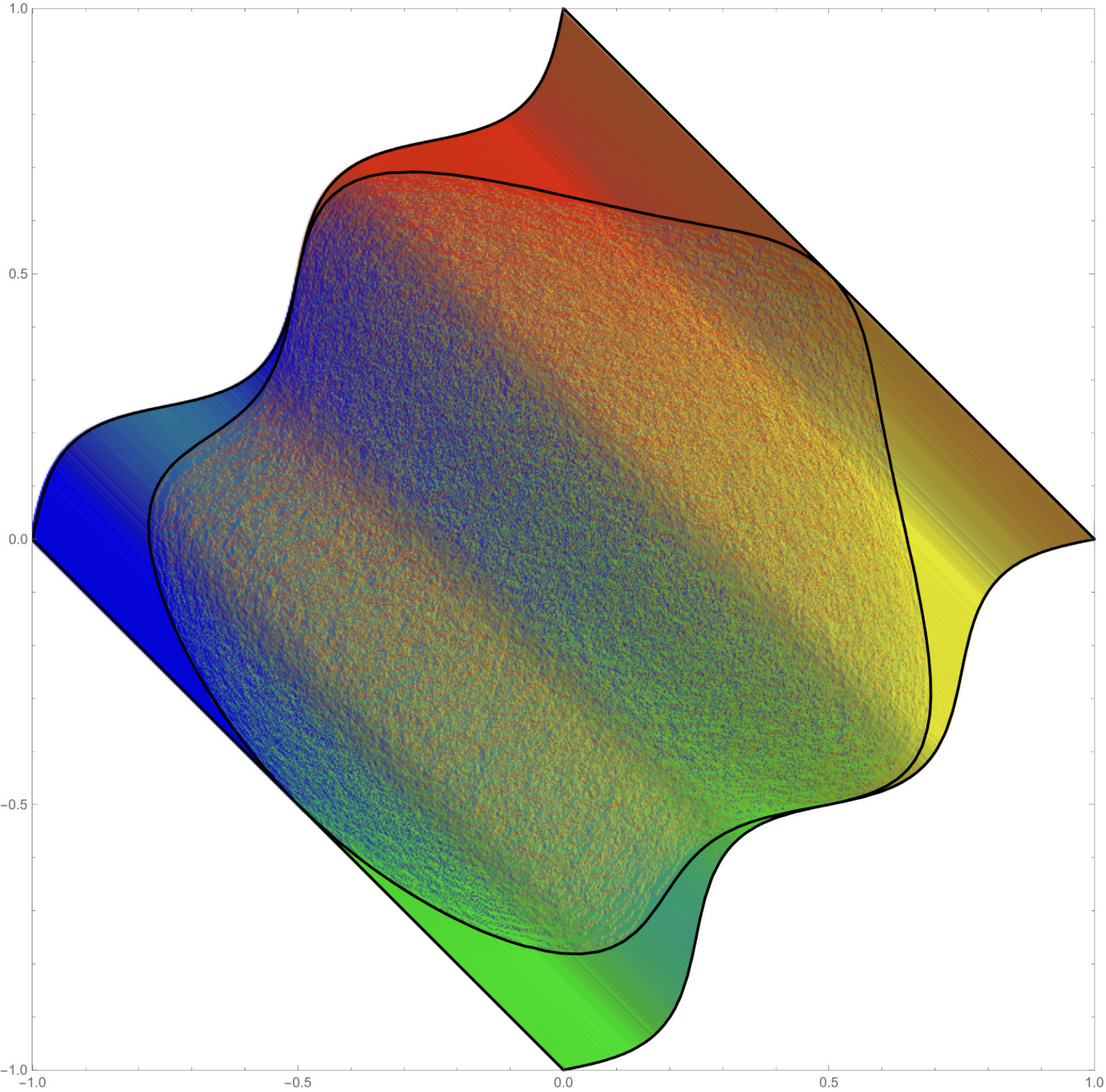}
    \quad
    \includegraphics[width=0.3\linewidth]{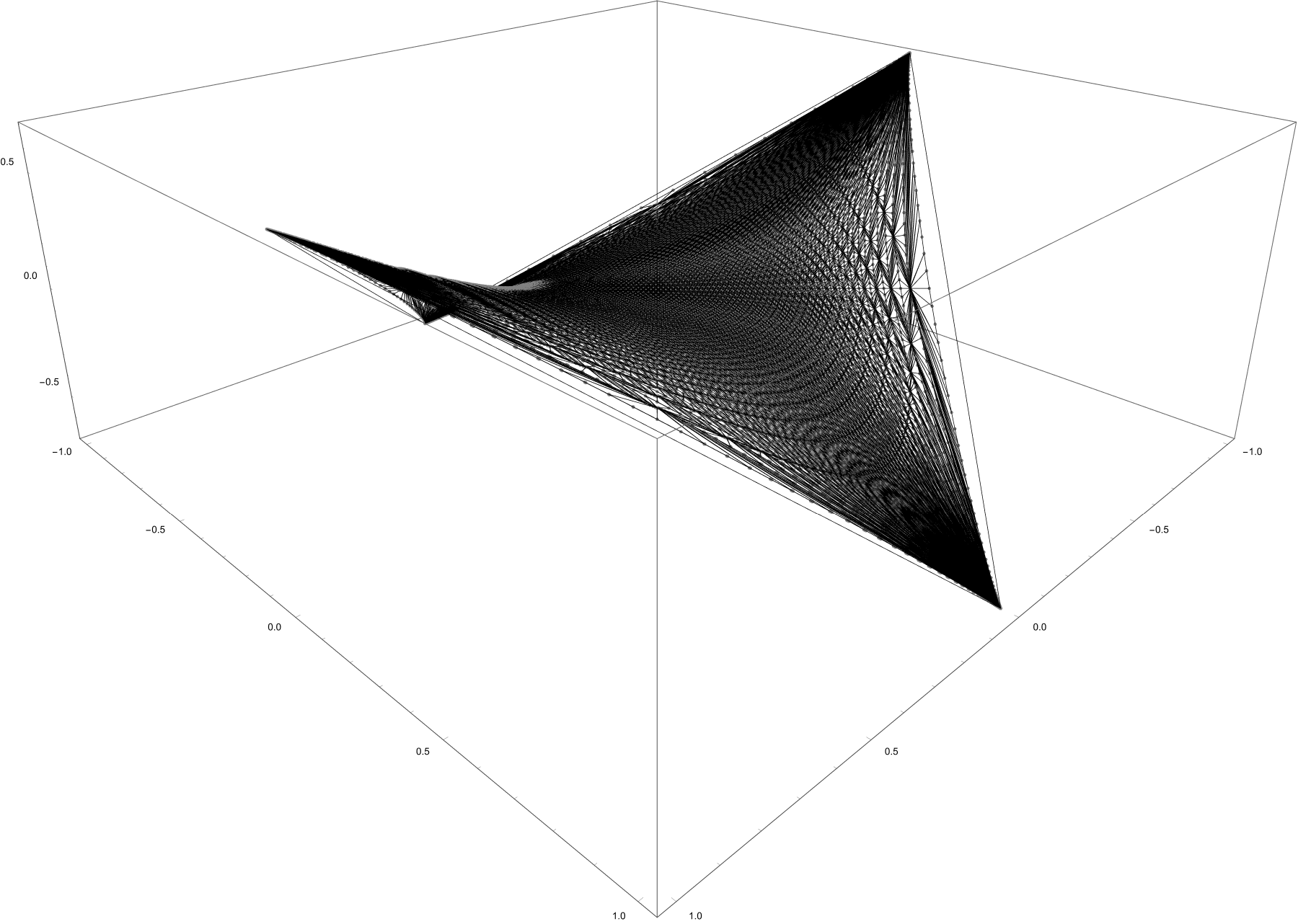}
    \]
    \caption{Left: A tiling of the generalized tower graph of rank $n=20$ approximating $f(x) = \frac{1}{3\pi}\sin(2\pi x) +1$. Center: A tiling of the same generalized tower graph of rank $n=1000$. The inner black curve is the theoretical arctic curve $x^2+y^2=\frac{1}{2}f(x+y)^2$. Right:  Plot of the graph $(\mathcal{\widetilde T}_n,\mathcal{\widetilde O}_n')$ for $n=150$.}
    \label{fig:largeSin}
\end{figure}

\newpage

\newpage
\bibliography{references.bib}
\bibliographystyle{plain}

\end{document}